\documentclass[a4paper,12pt]{article}
\usepackage{graphicx}
\usepackage{jheppub} 
\usepackage{braket}
\usepackage{array}
\usepackage{tcolorbox}
\usepackage{qcircuit}
\usepackage{amsthm}
\usepackage[toc]{glossaries}
\usepackage{subcaption}
\usepackage{cancel}

\newtheorem{theorem}{Theorem}
\numberwithin{theorem}{section}

\newtheorem{definition}[theorem]{Definition}

\newtheorem{lemma}[theorem]{Lemma}

\newtheorem{example}[theorem]{Example}
\let\oldexample\example
\renewcommand{\example}{\oldexample\normalfont}

\newtheorem{atomicexample}{Example}
\let\oldatomicexample\atomicexample
\renewcommand{\atomicexample}{\oldatomicexample\normalfont}

\newcommand{\bigzero}{\mbox{\normalfont\Large\bfseries 0}}
\newcommand{\rvline}{\hspace*{-\arraycolsep}\vline\hspace*{-\arraycolsep}}

\usepackage[backend=biber, natbib=false, bibstyle=alphabetic,citestyle=alphabetic, doi=true, sorting=nyt, giveninits]{biblatex}

\title{\boldmath Understanding holographic error correction via unique algebras and atomic examples}

\author[a]{Jason Pollack,}
\author[a]{Patrick Rall,}
\author[a]{Andrea Rocchetto}

\affiliation[a]{Department of Computer Science, The University of Texas at Austin, \\2317 Speedway,
Austin, TX 78712, USA}

\emailAdd{jasonpollack@gmail.com}
\emailAdd{patrickjrall@gmail.com}
\emailAdd{andrea.rocchetto@outlook.com}

\newcommand{\M}{\mathcal{M}}
\renewcommand{\H}{\mathcal{H}}

\abstract{
We introduce a fully constructive characterisation of holographic quantum error-correcting codes.
That is, given a code and an erasure error we give a recipe to explicitly compute the terms in the RT formula. 
Using this formalism, we employ quantum circuits to construct a number of examples of holographic codes. 
Our codes have nontrivial holographic properties and are simpler than existing approaches built on tensor networks. 
Finally, leveraging a connection between correctable and private systems we prove the uniqueness of the algebra satisfying complementary recovery.
The material is presented with the goal of accessibility to researchers in quantum information with no prior background in holography. 
}

\begin{document} 
\maketitle
\flushbottom

\section{Introduction}

In the last decade, ideas and tools from quantum information and computation have found an increasing number of applications in the efforts to understand the Anti-de Sitter/Conformal Field Theory (AdS/CFT) correspondence~\cite{maldacena1999large} as a holographic quantum theory of gravity. 
Notable examples include the ER=EPR~\cite{maldacena2013cool} conjecture and the proposed resolutions of: the black hole information paradox \cite{almheiri2020page, penington2020entanglement}, the firewall paradox~\cite{harlow2013quantum}, and the wormhole growth paradox in terms of the complexity=volume~\cite{susskind2016computational,aaronson2016complexity, haferkamp2021linear} and complexity=action \cite{brown2016holographic} conjectures.

Central to the connection between quantum gravity and quantum information is the Ryu-Takayanagi (RT) formula.
The RT formula conjectures that the entanglement entropy of a boundary CFT state is dual to the area of a bulk region in AdS~\cite{ryu2006holographic}.
The study of the entanglement properties of the AdS/CFT holographic duality~\cite{almheiri2013black}, spurred by the result of Ryu and Takayanagi, has led to a reformulation of the AdS/CFT correspondence in terms of quantum error-correcting codes~\cite{verlinde2013black, almheiri2015bulk, mintun2015bulk}. This  framework has helped to clarify the relationship between bulk and boundary and proved to be an effective and simple toy model of the AdS/CFT correspondence.

Based on these early results, researchers built toy models that reproduce key features of the correspondence (such as subregion duality, radial commutativity and the RT formula) using quantum error-correcting codes based on tensor networks~\cite{pastawski2015holographic, donnelly2017living}, random tensor networks~\cite{hayden2016holographic}, and approximate Bacon-Shor codes \cite{cao2020approximate}. All these models (which have been recently reviewed in~\cite{jahn2021holographic}) give an explicit bulk-boundary mapping for states and observables. Using techniques from Hamiltonian simulation, \cite{kohler2019toy} showed how the mapping can be extended to local Hamiltonians.

In parallel with the development of increasingly-advanced toy models, Harlow initiated a systematic study of holographic quantum error correction~\cite{harlow2017ryu,Akers:2021fut} (for a pedagogical introduction to these ideas see~\cite{harlow2016jerusalem, harlow2018tasi, rath2020aspects}).
Leveraging the operator algebra quantum error correction framework developed in~\cite{beny2007quantum, beny2007generalization, kribs2005unified, kribs2006operator},~\cite{harlow2017ryu} identified the conditions that make a quantum error-correcting code a good holographic code (that is, a code that reproduces the key features of the AdS/CFT correspondence). In particular,~\cite{harlow2017ryu} showed that standard quantum error-correcting codes such as stabiliser codes~\cite{gottesman1997stabilizer, gottesman2010introduction} or subsystem codes~\cite{poulin2005stabilizer, bravyi2011subsystem}, correct errors ``too well'' to give rise to good holographic codes. This statement can be made precise using the language of finite-dimensional von Neumann algebras, which we review in Section~\ref{sec:vonNeumann}. Consider the algebra of operators that can be reconstructed after the erasure of a region of the boundary: for a good holographic code this algebra is not a factor algebra.

In this paper, we build on the formalism of Harlow to derive new properties and examples of holographic quantum error-correcting codes. Our contributions further sharpen our understanding of the relationship between bulk and boundary and give even simpler examples of holographic codes reproducing key features of the AdS/CFT correspondence.
In particular:

\begin{itemize}
    \item We give new ``atomic'' examples of holographic codes. The key feature of these examples is that they are based on quantum circuits with a minimal number of qubits rather than on the large tensor networks that have appeared in the literature \cite{pastawski2015holographic, donnelly2017living, hayden2016holographic, cao2020approximate}.
    By significantly reducing the complexity of the toy models, we hope to introduce a new tool to identify what features of error correcting codes enable the emergence of holographic states.
    
    \item  We prove new properties of holographic quantum correcting codes. More specifically, we show that the code algebra is the unique von Neumann algebra satisfying complementary recovery (defined below). The proof is obtained by leveraging a connection between quantum error correction and quantum privacy~\cite{crann2016private, kribs2018quantum} which we believe is entirely novel in the context of holographic quantum error correction. The uniqueness of the algebra shows that error correcting codes which satisfy complementary recovery are ``rigid'' in the sense that they are uniquely determined by the requirements of holography.
    
    \item We give a reformulation of key results in holographic quantum error correction which is aimed at experts in quantum information. This might be a desirable feature for researchers with a quantum information background that are venturing into the field and could give people already familiar with these ideas a new angle to think about related problems. 
\end{itemize}

We give a brief presentation of key results from holographic quantum error correction in Section~\ref{sec:overview_HQEC} and a detailed overview of our contributions in Section~\ref{sec:overview_our_contributions}. 

The remainder of this paper is organised as follows.
Section~\ref{sec:holography_background} gives an informal presentation of some of the central concepts in holography for a reader with no prior background on the subject.
In Section~\ref{sec:vonNeumann}, we review some key facts about finite-dimensional von Neumann algebras.
Section~\ref{sec:complementary_recovery} and Section~\ref{sec:examples} contain the bulk of our contributions.
In Section~\ref{sec:complementary_recovery}, we give a reformulation of holographic quantum error correction and prove new properties of the code algebra, while in Section~\ref{sec:examples}, we construct several ``atomic'' examples of holographic codes using quantum circuits.
We conclude in Section~\ref{sec:discussion} with some remarks on the main differences between our work and \cite{harlow2017ryu} and a list of open questions. 

This paper has three Appendices.
In Appendix~\ref{app:privacy}, we review some key notions and results on quantum private systems.
In Appendix~\ref{app:structure_lemma}, we give a new proof of the main theorem in~\cite{harlow2017ryu}.
In Appendix~\ref{app:2x2bacon-shor}, we give a full analysis of the holographic properties of the 2x2 Bacon-Shor code.

\subsection{Overview of holographic quantum error correction}
\label{sec:overview_HQEC}

\begin{figure}
\begin{center}
\includegraphics[width=0.8\textwidth]{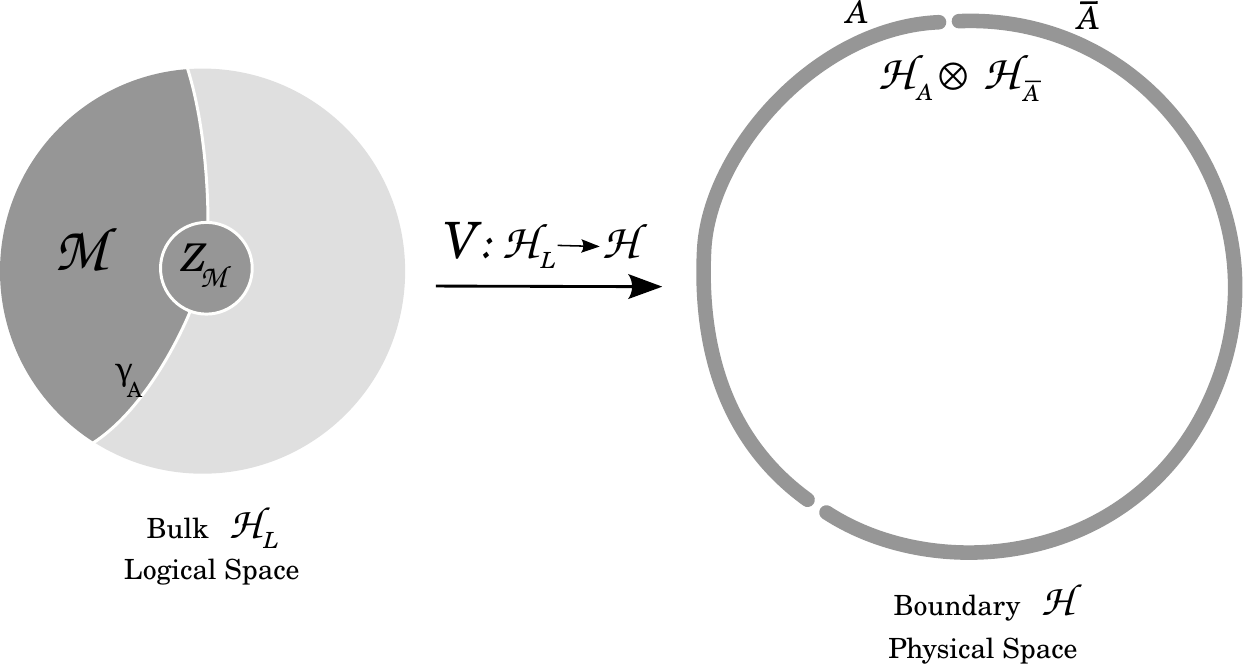}
\end{center}
\caption{\label{fig:notation} Sketch of a holographic quantum error-correcting code in $2+1$ dimensions using our notation, indicating some of the terms in Table~\ref{table:Rosetta}. }
\end{figure}

In holography, we consider a bulk asymptotically-AdS space described by a Hilbert space $\mathcal{H}_L$, surrounded by a boundary CFT with a Hilbert space $\mathcal{H}$. The correspondence manifests via a holographic dictionary $V : \mathcal{H}_L \to \mathcal{H}$, which maps the state from the bulk into the boundary. See Figure~\ref{fig:notation}. The same setup cleanly maps to a quantum error-correcting code. We let $\mathcal{H}_L$ be the logical space, and $\mathcal{H}$ be the physical space. Then $V$ is an isometry that takes the data in $\mathcal{H}_L$ and encodes it in the physical space $\mathcal{H}$.

Our goal is to concretely define what we mean for such a setup to exhibit ``holographic quantum error correction''. We will do this by taking an RT formula and writing it in the notation of quantum error correction. Then we can proceed to derive general properties of such an RT formula, and building specific examples of codes that exhibit one. Having these concrete examples can illuminate the relationship between bulk and boundary, and generally make AdS/CFT easier to understand.

We begin with a classical RT formula. Say $A$ is a subregion of the boundary space, splitting the boundary into a bipartition $A$-$\bar A$. Then, the classical RT formula states\footnote{This is the version of the formula that holds in a \emph{static} geometry, i.e.\ one that can be described by a time-independent metric. In a time-dependent geometry the extremization is more subtle, and is described by the maximin prescription \cite{Engelhardt:2014gca} for the HRT formula \cite{Hubeny:2007xt}. In particular, the geodesics are not confined to a fixed spatial slice of the bulk but instead live inside the \emph{entanglement wedge} of the boundary region; see Section \ref{sub:wedges} for further discussion. Furthermore, the minimization over geodesics should include only the geodesics homologous to A; see Footnote \ref{fn:homology} for a discussion.}
 that, in a holographic state corresponding to a (2+1)-dimensional classical bulk geometry:
 \begin{align}
    \text{entanglement across } A\text{-}\bar A \propto \min_{\gamma_{A}} \text{Area}(\gamma_{A}), 
\end{align}
where $\gamma_A$ is a geodesic in the (negatively curved, gravitating) bulk whose endpoints are the same as those of $A$ on the boundary. In $d+1$ dimensions, the geodesic is replaced by a $(d-1)$-dimensional extremal surface ending on (and homologous to) the boundary subregion; `area' denotes a codimension 2 quantity, and hence it is actually a length for a (2+1)-dimensional bulk. The RT formula connects a geometrical quantity, an area, with a quantum-mechanical quantity: the entanglement entropy. (Readers who find the RT formula unfamiliar are invited to consult Section~\ref{sec:holography_background} for a more detailed exposition of the quantum-gravitational setting where the formula arises.)

If the bulk is itself a quantum system that can be in a mixed state, then we must be more careful in defining the left-hand side of the equation: we only care about the entanglement entropy stemming from the holographic dictionary $V$, but not any entropy from the bulk degrees of freedom. Thus, we must subtract off the entropy from the bulk state. Say $\rho_L$ is a state in the bulk $\mathcal{H}_L$, and $\rho = V \rho_L V^\dagger$ is its encoded state on the boundary $\mathcal{H}$. The subregion $A$ induces a factorization\footnote{In a conformal field theory this factorization may be subtle: to ensure the theory factorizes we can introduce edge modes \cite{Donnelly:2016auv}. Throughout this paper we will follow convention and assume without comment that the boundary theory does indeed factorize, which is already necessary to define the left-hand side of the RT formula.} of the boundary into $\mathcal{H} = \mathcal{H}_{A} \otimes \mathcal{H}_{\bar A}$. We can then say $\rho_{A}$ is the reduced state obtained by taking $\rho$ and tracing out $\mathcal{H}_{\bar A}$. Now we can phrase a quantum RT formula as:
\begin{align}
\text{entropy of } \rho_A - \text{entropy of } \rho_L \text{ visible from }A  = \min_{\gamma_{A}} \text{Area}(\gamma_{A}). 
\end{align}

We can define the entropy of $\rho_A$ via the von Neumann entropy $S(\rho_A)$. The other quantities are more challenging to define. The geometry itself may be a superposition, so that the area actually corresponds to an observable $L_A$ on the bulk Hilbert space $\mathcal{H}_L$. The area contribution to the RT formula is then the expectation $\langle L_A \rangle_{\rho_L} = \text{Tr}(\rho_L L_A)$. It can be reconstructed by a bulk observer given access to either subregion. The state $\rho_L$ describes the state of the bulk, and thus also captures the superposition over geometries. We are left with:
\begin{align}
S(\rho_A) = \text{entropy of } \rho_L \text{ visible from }A   + \text{Tr}(\rho_L L_A).
\end{align}

To make the ``entropy of $\rho_L$ visible from $A$'' rigorous, we will need some tools from the quantum error correction literature. Our goal is to identify a collection of operators $\mathcal{M}$ on $\mathcal{H}_L$ that exactly capture what we can see given only access to the boundary subregion $A$. Then we can use this family of operators to define the entropy. Some language developed in the quantum error correction framework from \cite{beny2007quantum, beny2007generalization, kribs2005unified, kribs2006operator} is particularly useful for this purpose. These papers are concerned with what kinds of observables in $\mathcal{H}_L$ are affected by a general quantum error channel. Here we restrict their language to erasure errors, since we just want to erase the subregion $\bar A$.

\begin{definition} Say $V:\mathcal{H}_L \to \mathcal{H}$ is an encoding isometry, and $A$ is a subregion that induces the factorization $\mathcal{H} = \mathcal{H}_{A} \otimes \mathcal{H}_{\bar A}$. 

We say a bulk operator $O_L \in \mathcal{L}(\H_L)$ is \textbf{correctable from $A$} if there is some boundary operator $O$ with support only on $A$ that lets us access $O_L$ via $V^\dagger O V = O_L$.

If all the operators with support only on $A$ commute with $O_L$ after projection with $V^\dagger$, then the observable corresponding to $O_L$ cannot be measured from $A$. In this case we say $O_L$ is \textbf{private from $A$}. 

\end{definition}

We are looking for a collection of operators $\mathcal{M}$ that exactly capture what degrees of freedom are visible from $A$. These are all correctable from $A$. To make sure we are not missing any operators, we want the `mirror image' of this condition to be true from $\bar A$. If $\mathcal{M}'$ is the set of all operators in the bulk that commute with $\mathcal{M}$ (also known as the `commutant' or the `normalizer of $\mathcal{M}$ in $\mathcal{L}(\H_L)$'), then we want $\mathcal{M}'$ to be correctable from $\bar A$. The center $Z_\mathcal{M} := \mathcal{M}\cap\mathcal{M'}$, which contains the area operator, is correctable from both regions. We call this condition `complementary recovery'\footnote{Holography experts: note that in classical holographic states, the equivalent of this condition is that access to a boundary subregion $A$ allows the reconstruction of bulk operators in (at least) the causal wedge of $A$, while access to the complement $\bar A$ allows the reconstruction of operators in the causal wedge of $\bar A$. See Figure \ref{fig:subregion}; although the union of these two causal wedges does not cover the entire spacetime, when $A$ is a \emph{spatial} subregion and the boundary state is pure it \emph{does} contain the entirety of a spatial slice of the bulk. (If the boundary state is mixed, the union won't cover an entire spatial slice; for example, there could be a black hole horizon beyond which the boundary-anchored geodesics will not penetrate.)}:

\begin{definition} An encoding isometry $V$, a subregion $A$, and collection of operators $\mathcal{M}$ satisfy \textbf{complementary recovery} if $\mathcal{M}$ is correctable from $A$, and $\mathcal{M}'$ is correctable from $\bar A$.
\end{definition}

If we can find such a collection of operators $\mathcal{M}$, then we have exactly captured the degrees of freedom in $\mathcal{H}_L$ that are visible from $A$. Now all that is left to do is to use $\mathcal{M}$ to define an entropy on $\rho_L$. It turns out that there is a very natural way of doing this if $\mathcal{M}$ is closed under multiplication, in which case $\mathcal{M}$ is a von Neumann algebra. In this case there is a natural generalization of the entanglement entropy called the `algebraic entropy' $S(\mathcal{M},\rho_L)$ (which we review in Section~\ref{sec:vonNeumann}). This finally lets us define what we mean by `entropy of $\rho_L$ visible from $A$', and state an RT formula in quantum mechanical language: 
 \begin{align}
S(\rho_A) = S(\mathcal{M},\rho_L)   + \text{Tr}(\rho_L L_A). \label{eqn:intro_rt_formula}
\end{align}

In fact, a key result of \cite{harlow2017ryu} is that the existence of a von Neumann algebra $\mathcal{M}$ implies that the code satisfies an RT formula.

\begin{theorem}\label{thm:simple_complementaritytoRT}\textbf{From Theorem~5 of \cite{harlow2017ryu}.} Say an encoding isometry $V$, a subregion $A$, and a von Neumann algebra $\mathcal{M}$ satisfy complementary recovery. Then there is an area operator $L_A$ such that (\ref{eqn:intro_rt_formula}) holds.
\end{theorem}

A summary of the notation from this discussion is to be found in Table~\ref{table:Rosetta}. In section~\ref{sec:vonNeumann} we give an introduction to von Neumann algebras and their properties. Then, in Section~\ref{sec:complementary_recovery} we present the above discussion in more mathematical detail and also outline some of our main results.

\begin{table}
\small
\begin{center}
    \begin{tabular}{p{3cm}|p{4.2cm}|p{4.2cm}}
     \textbf{Symbol} & \textbf{Quantum Error Correction \mbox{Interpretation}} & \textbf{Holographic \mbox{Interpretation}} \\
    \hline 
    \hline
    $\mathcal{H}$ & physical space & boundary CFT  \\
    \hline
    $\mathcal{H}_{L}$ & logical space & bulk asympotically AdS space \\
    \hline
    $V :\mathcal{H}_{L} \rightarrow  \mathcal{H}$ & encoding isometry & AdS/CFT dictionary \\
    \hline
    $\ket{\psi}_L$, $\rho_L$, $O_L$ & state, operator in the logical space $\mathcal{H}_L$ &  state, operator in the bulk \\
    \hline
    $\mathcal{H}_A$, where $\mathcal{H} = \mathcal{H}_A \otimes \mathcal{H}_{\bar A}$ & the region of $\mathcal{H}$ that remains after the erasure of $\bar A$  & a subregion of the boundary complementary to $\bar A$ \\
    \hline
    $\ket{\psi}_A$, $\rho_A$, $O_A$ & state, operator in $\mathcal{H}_A$ & boundary state, operator in the $A$ subregion \\
    \hline
    $\mathcal{M} $ & algebra of operators protected from the erasure of $\bar A$ & algebra of bulk operators in the entanglement wedge of $A$ \\
    \hline
    $Z_\mathcal{M}$ & algebra of operators protected from the erasure of either $A$ or $\bar A$  & bulk operators that can be reconstructed from either $A$ or $\bar A$  \\
    \end{tabular}
\end{center}
\caption{\textbf{A Rosetta stone for symbols and their quantum error correction and holographic interpretations}. First column: main symbols used throughout the paper. Second column: interpretation of the symbol in the language of quantum error correction. Third column: holographic interpretation of the symbol.}
\label{table:Rosetta}
\end{table}

\subsection{Overview of our contributions}
\label{sec:overview_our_contributions}

 The central goal of our work is to build concrete examples and analysis techniques using the work of \cite{harlow2017ryu} as a starting point. We begin with some general results that aid in the analysis of holographic codes. First, we show that the von Neumann algebra of interest to holography is unique, and that there is a direct way of computing it. Then, we give several `atomic' examples of quantum error correction codes that manifest holographic features despite only possessing very few qubits.

\begin{theorem} \label{thm:simple_vn_algebra}\textbf{What is the von Neumann algebra?} Say $V$ is an encoding isometry and $A$ is a subregion. If there exists a von Neumann algebra $\mathcal{M}$ such that $V,A,\mathcal{M}$ obey complementary recovery, then it is unique.

Furthermore, let $\mathcal{M}$ be exactly the set of operators in the bulk that are correctable from $A$. If it is closed under multiplication, then it is the unique von Neumann algebra with complementary recovery. Otherwise, no such algebra exists.
\end{theorem}

This theorem is a direct consequence of a result from the quantum error correction literature:
\begin{lemma} A von Neumann algebra $\mathcal{M}$ is correctable from $A$ if and only if it is private from $\bar A$.
\end{lemma}

The main idea is that complementary recovery restricts $\mathcal{M}$ from both sides: on the one hand $\mathcal{M}$ must be correctable from $A$ so it cannot contain too many operators. But on the other hand, since $\mathcal{M}'$ is correctable from $\bar A$, we must have that $\mathcal{M}'$ is private from $A$. So $\mathcal{M}$ must be large enough so that its commutant remains small enough to be private. In particular, when $\mathcal{M}$ is correctable from $A$  there is then no proper subalgebra of $\mathcal{M}$ whose commutant is correctable from $\bar A$. We comment on the implications of this fact for bulk reconstruction in the Discussion.

The uniqueness theorem implies a concrete `recipe' for analyzing quantum error-correcting codes and determining their RT formulae. In section~\ref{sec:complementary_recovery} we present Theorem~\ref{thm:simple_vn_algebra}, as well as a mostly self-contained derivation of Theorem~\ref{thm:simple_complementaritytoRT}. With all these mathematical tools together, the section culminates in a series of step-by-step instructions for analyzing a quantum error-correcting code.

In section~\ref{sec:examples} we then practice this recipe on several examples. We construct these examples by building quantum circuits for the encoding isometry $V$. These examples are designed to flesh out the different terms of (\ref{eqn:intro_rt_formula}) to varying degrees of completeness. The section culminates in the example sketched in Figure~\ref{fig:full_example}.

\begin{figure}
\begin{center}
\includegraphics[width=0.8\textwidth]{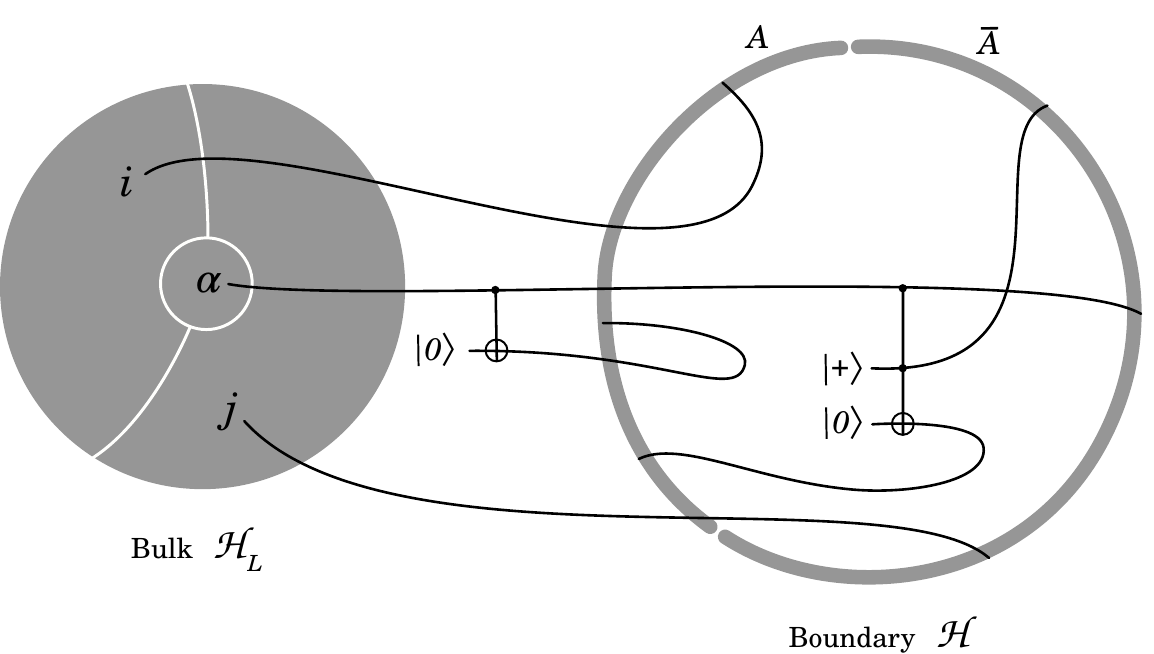}
\end{center}
\caption{\label{fig:full_example} A sketch of a quantum circuit with an RT formula. On the left side, $i, \alpha, j$ label three qubits in the bulk. The central degree of freedom $\alpha$ looks classical since it is visible from both $A$ and $\bar A$ via the CNOT. $\alpha$ also conditionally creates some of the entanglement across the bipartition via a Toffoli gate, which determines the area operator's eigenvalues. The bulk degrees of freedom $i,j$ are not encoded at all and are plainly visible from the boundary. A full technical explanation of the encoding can be found in Sections~4 and 5. See also Example~\ref{ex:cql} in particular. }
\end{figure}

Here we give a brief non-technical summary of how the code in Figure~\ref{fig:full_example} generates an RT formula. Since the full explanation is fairly involved, we focus only on the intuition behind the features and leave the technical explanation for Sections~4 and~5.  We will review in Section~3 that the algebraic entropy $S(\mathcal{M},\rho_L)$ naturally divides into two terms: a classical term $S_c$ and a quantum term $S_q$. 

 \begin{align}
S(\rho_A) = S_c + S_q + \text{Tr}(\rho_L L_A). 
\end{align}

The degree of freedom $\alpha$ is copied via CNOT onto an extra qubit, so that $\alpha$ can be seen from both $A$ and $\bar A$. This essentially `measures' the $\alpha$ degree of freedom in the computational basis, and makes it look entirely classical. This is the classical part of the entropy $S_c$.

On the other hand, the $i$ and $j$ degrees of freedom are not measured, and thus retain their coherence. Subregion $A$ cannot see the $j$ qubit, but it can see the $i$ qubit. Thus, the von Neumann entropy of the $i$ qubit forms the quantum part of the entropy $S_q$.

Finally, some of the entanglement across $A$-$\bar A$ stems from the subcircuit involving the Toffoli gate, which connects $\alpha$ and the two boundary regions. This entropy forms the area term $\text{Tr}(\rho_L L_A) $ since it is completely independent of entanglement between bulk degrees of freedom. However, the generation of entanglement is conditional on the value of $\alpha$, so the area operator $L_A$ is actually $\ket{1}\bra{1}$ on the $\alpha$ subsystem.  In this sense $\alpha$ is a bulk degree of freedom that indexes which geometry we are in.

The examples in section~\ref{sec:examples} are in some sense the simplest possible quantum error-correcting codes with non-trivial holographic properties. Their purpose is primarily to serve as pedagogical examples for understanding holographic quantum error correction. However, there are many possible future directions, some of which we elaborate in the Discussion below. An obvious direction is to try to use these examples as building blocks for a tensor network that supports superpositions of geometries. This is an idea that tensor network models seem to struggle to capture. Other possibilities include trying to find an example so that $\alpha$ is not a separate degree of freedom, or to add dynamics.

\section{The holographic principle}
\label{sec:holography_background}

What characterizes a (quantum) gravitational/holographic theory? Which of these characteristics can be captured by low-dimensional discrete toy models? The goal of this section is to provide the interested reader with enough intuition about gravity and holography to answer these questions. In particular, reading this section is meant to establish the concepts and intuition necessary to understand the following key claim:

\begin{itemize}
    \item The entanglement structure of holographic states is special. The entanglement entropy of a mixed state is not in general an observable. However, the reduced density matrix in the spatial subregion $A$ of a \emph{classical} holographic state $\rho$ obeys an area law known as the Ryu-Takayanagi formula:
\begin{align}\label{eq:classical_RT}
    S_A(\rho) = \frac{1}{4G_N} \min_{\gamma_{A}} Area(\gamma_{A}).
\end{align}
That is, entropies of subregions are proportional to the value of a geometric observable, the minimal area operator, which is obtainable by following a well-defined procedure (detailed later in this section) to construct the classical geometry corresponding to $\rho$, and the constant of proportionality involves the gravitational constant $G_N$. Furthermore, holographic states in a much larger class, where we allow entanglement of bulk degrees of freedom as well as superpositions of different geometries, nevertheless have entropies described by a more general RT formula: 
\begin{align}\label{eq:full_RT}
    S_A(\rho) = S_{\text{bulk},A}(\rho) + \text{Tr}(L \rho),
\end{align}
where $L$ is again a bulk observable whose eigenvalues are the minimal areas of the different geometries in the superposition. An RT formula also holds for the reduced state in the complementary subregion $\bar A$, with the \emph{same} operator $L$.
\end{itemize}

The following sections of the paper will be devoted to understanding the features of, and building examples of, quantum error-correcting codes which obey Eq.\ \eqref{eq:full_RT}. In particular, we will show that an operator $L$ exists for codes which obey a \emph{complementary recovery} property relating errors correctable in a region $A$ to errors correctable in the complement of the region. These codes are themselves \emph{holographic}: even though the codes are not gravitational, encoded states have the same special entanglement structure.

However, before moving to the error-correction setting, in this section we give more details on what the RT formula \emph{means}. We discuss how to describe a spacetime geometry which obeys the equations of general relativity: first using a metric, then using the more invariant data of geodesic distances and extremal areas. We then sketch how to think of a \emph{quantum state} which describes (a spatial slice of) a given geometry, and then the larger Hilbert space which allows for entangled degrees of freedom living on these geometries, as well as states describing superpositions of geometries. Finally, we move from this abstract discussion to the more concrete setting of holographic theories which give dual descriptions of some of these quantum-gravitational Hilbert spaces, allowing us to measure bulk observables via a holographic operator dictionary and relating the bulk geometry to the boundary entropic structure.

\subsection{Some general relativity}

We begin by explaining the objects that appear on the right-hand side of \eqref{eq:classical_RT}; in later subsections we'll discuss the meaning of the left-hand side and of the generalization \eqref{eq:full_RT}. These are geometric quantities, so we'll first need to explain what we mean by a geometry, by introducing the notion of a metric to define the distance between points and along curves, and then the special curves called geodesics which define the causal structure of a spacetime. We'll then pass from mathematics to physics by discussing how the Einstein field equations relate the geometry to the energy and matter living on top of it. Finally, we'll give a (reasonably) careful discussion of the symmetries of spacetime and of the Einstein equations. Many metrics can describe the same spacetime, so if we want to work with physical quantities we need objects which don't change when we alter the metric but leave the spacetime unchanged. We'll see that, when a spacetime has a boundary, one of these quantities is precisely the minimal area operator appearing in \eqref{eq:classical_RT}.

We recommend that the interested reader looking for a more complete but still concise introduction to GR consult \cite{carroll2001no}.

\subsubsection{Metrics and distances}\label{sub:metrics}

The full machinery of quantum gravity won't be necessary for this review, but it will be useful to establish some intuitions and terminology. We begin with classical general relativity. The space in this theory is a particular non-flat D-dimensional geometry. Formally, what we mean by a ``geometry" is some smooth (differentiable) manifold, which we can describe by some set of coordinates $\{x^\mu\}_{\mu=0}^{D-1}$, where the index $\mu$ ranges over the $D$ dimensions of the manifolds. What we mean by ``curved'' is that distances between points in this manifold aren't given by the Euclidean distance. Instead, we use a more general notion of distance, a \emph{pseudo-Riemannian metric} $g_{\mu \nu}$.
Using the metric, we can define the \emph{line element}
\begin{equation}
    ds^2 = g_{\mu \nu} dx^\mu dx^\nu,
\end{equation}
where we are adopting the convention that repeated indices are summed over. That is, at every point of space the metric is a \emph{matrix} (or, more formally, a two-index tensor): if we specify a particular point $x$ and pick two coordinate directions $\mu,\nu$ we can find the matrix element $g_{\mu \nu}$. 
When $g_{\mu \nu}=\delta_{\mu \nu}$ at all points in space, we recover the special case of the Euclidean metric in $D$ dimensions: the line element is
\begin{equation}
    ds^2_\text{Euc} = (dx^0)^2 + \ldots + (dx^{D-1})^2 \equiv dx^\mu dx^\mu.
\end{equation} 
However, we more often have cases in which some of the coordinates are \emph{timelike}, $g_{aa}<0$. In particular, when exactly one of the coordinates (by convention, $x^0$) is timelike (or, more precisely, when the metric has one negative eigenvalue everywhere on the manifold), we say that the metric is \emph{Lorentzian}. A simple example is given by taking the Euclidean metric and putting a minus sign in front of the {00} (``time-time'') component:
\begin{equation}
    ds^2_\text{SR} = -(dx^0)^2 + \ldots + (dx^{D-1})^2 = dx^\mu dx^\mu.
\end{equation} 
This is a metric which describes the situation of \emph{special relativity}: we can see that, for any fixed value of $x_0$, the \emph{spatial} part of the metric is still flat.

The line element in turn allows us to compute the length of a curve $\gamma$, which we can parametrize as a choice of coordinates at each point on the curve: $\gamma(\lambda)=x^\mu(\lambda)$. The length of the curve is given by adding up the infinitesimal displacements along the curve, i.e. the arc length integral
\begin{equation}\label{eq:arc_length}
    \left|\gamma\right| \equiv \int_0^1 d\lambda \sqrt{ds^2} = \int_0^1 d\lambda \sqrt{g_{\mu \nu} \frac{dx^\mu(\lambda)}{d\lambda} \frac{dx^\nu(\lambda)}{d\lambda}}.
\end{equation}
In a given geometry, we can construct the set of \emph{all possible} (smooth) curves which connect two points (in the equation above, the points are $x^\mu(0)$ and $x^\mu(1)$). Individual curves in this set depend on some choice of coordinates, but, crucially, the entire set depends only on the geometry and the choice of points\footnote{Admittedly, so far we've labelled these points in a particular coordinate system, but we could just call them $A$ and $B$, or alternatively consider any coordinate system that preserves the locations of the two points but allow the coordinates to vary on the rest of the geometry. Below we're going to consider the whole \emph{space} of geodesics, and that will remove even this dependence.}. So any quantity we can compute given access to the entire set is coordinate-independent. In particular, we'd like to use the set to come up with a coordinate-independent distance between two points.

In Euclidean (or more generally Riemannian) metrics, one such quantity is \emph{the length of the shortest curve connecting two points}. If there's a timelike direction, this isn't true anymore: we can take a curve and add zig-zags in the timelike direction which will make the curve steadily shorter and shorter. So we can't simply take this as our distance measure. The right generalization turns out to be to consider the lengths not of \emph{minimal} curves, but \emph{extremal} ones. To find these curves, we take the arc length integral \eqref{eq:arc_length}, consider it as a functional of the curve $\gamma$, vary the curve $\gamma\rightarrow\gamma + \delta\gamma$, and look for stationary points of the variation. This is the fundamental problem of the calculus of variations, and its solution is given by solving the Euler-Lagrange equations with the action taken to be the line element $ds=\sqrt{ds^2}$. We won't write this down explicitly, but the equation to be solved is known as the \emph{geodesic equation} and the extremal curves are called \emph{geodesics}; in GR, they're the curves traced out by non-accelerating (``freely-falling'') observers. If all of the geodesics connecting two points have the same length, we call this length the \emph{geodesic distance} between two points; if there are multiple geodesics with different lengths, we take the geodesic distance to be the shortest such length. 

In geometries with one timelike direction, the geodesic distance between pairs of points gives a \emph{causal structure} for the geometry: for all pairs of points, we can tell whether they are spacelike, timelike, or null separated by checking whether the geodesic distance between them is, respectively, positive, negative, or zero. When two points are timelike separated, we often call the negative of the geodesic distance the \emph{proper time}; with appropriate units, it measures the time elapsed by a clock carried by an observer freely falling between the two points. Crucially, as we'll discuss below, the causal structure is really a property \emph{of the geometry itself}: the set of all geodesics on a manifold is independent of how the metric is parameterized, and two metrics describe the same geometry precisely when they produce the same causal structure.

It should be clear that we can generalize this entire dicussion by passing from curves to higher-dimensional objects (surfaces, volumes, etc.). Instead of the arc length integral \eqref{eq:arc_length} we have some higher-dimensional integral, which we vary to find stationary solutions: extremal surfaces, volumes, etc. Like the geodesics, these are similarly coordinate-invariant objects. For ease of drawing figures, we'll typically work in two space and one time dimension. This hopefully explains our choice of notation in \eqref{eq:classical_RT}: the $A$ is a subregion of a spatial slice of the boundary, that is, a dimension $D-2$ (``codimension 2'') object, and the minimization is over the areas of extremal dimension $D-2$ objects in the bulk of the spacetime that touch the boundary at the edge of $A$ (actually a subclass of these objects, as we'll discuss towards the end of this section). If, like our universe, $D=4$, $A$ would be two-dimensional. But in three spacetime dimensions $D-2=1$, so the boundary is equivalent to a circle, $A$ is some portion of that circle, and the relevant extremal objects are curves which we've accordingly labeled as $\gamma_A$. We nevertheless call the operator which computes their length an \emph{area} operator because in general spacetime dimension the invariant objects have \emph{codimension} 2.

\subsubsection{The Einstein field equations}

So far we have just done mathematics (differential geometry, to be precise). General relativity is a physical theory which relates the geometry of a manifold to the matter distribution living on it. More precisely, the Einstein field equations read
\begin{equation}\label{eq:einstein}
    G_{\mu \nu} = 8 \pi G_N T_{\mu \nu}.
\end{equation}
It won't be necessary for us to precisely define the objects in this equation, but we mention several features of it:
\begin{itemize}
    \item $G_{\mu \nu}$ is the ``Einstein curvature tensor'', a geometric object which is a function of the metric and its derivatives.
    \item $T_{\mu \nu}$ is the ``stress-energy tensor.'' In a particular coordinate frame in a Lorentzian spacetime, we can identify $T_{00}$ as the energy density, $T_{0k}$ as momentum density in the $x^k$ direction, and the mixed components as pressures and stresses.
    \item \eqref{eq:einstein} is $D^2$ equations given by different choices of $\mu,\nu$, but both the left- and right-hand sides of the equation are symmetric under exchange of $\mu$ and $\nu$, e.g. $G_{\mu \nu}=G_{\nu \mu}$, so there are only $D(D-1)/2$ independent equations, 10 in 4 spacetime dimensions. For a fixed choice of stress-energy tensor, this is a set of (second-order) nonlinear coupled differential equations which determine the metric.
    \item Although it isn't manifest in this form, we can often rewrite the Einstein equations as an \emph{initial value problem}: if we know the metric and its derivatives and the stress-energy, \emph{at one particular moment in time}, i.e. everywhere in space for one particular value of a timelike coordinate, we can use the field equations to tell us what the metric will be at some later time\footnote{When we're trying to solve Einstein's equations on a manifold with a boundary, we need to give (spatial) boundary conditions in addition to initial data. This is the case, in particular, for the asymptotically-AdS spacetimes that are of interest in holography.}. Now we can study, for example, backreaction---given some particular matter configuration, how does the geometry evolve?
\end{itemize}

\subsubsection{Diffeomorphism invariance}\label{sub:diff}

We said above that a particular geometry is described by a manifold specified by a metric. In general, however, there is not a one-to-one correspondence between metrics and geometries---there are many metrics which describe the same geometry. We're already used to being able to use different coordinate choices to describe the same physics: in Newtonian physics we're free to choose different origins and choices of axis for our coordinate system, or, with a little more work, to make one coordinate system move and rotate with respect to another. We can tell that two seemingly different situations are actually the same thing in different coordinate systems when the laws of physics are the same in both situations. In Newtonian physics, for example, acceleration is the same in all inertial frames, and Newton's laws of motion depend only on the acceleration: they have the property of ``Galilean invariance." Similarly, in special relativity, the laws of physics are invariant under Lorentz transformations; two observers may differ, for example, on what the strength of an electric or magnetic field is but they will agree on the appropriate invariant combinations.

So, to answer the question of whether two metrics describe the same geometry, we should compute invariant quantities and check whether they are the same in both situations. In GR, the invariant quantities are related to the causal structure: they are the proper distances \eqref{eq:arc_length} along geodesics. Two metrics which share the same causal structure are related by a \emph{diffeomorphism}. That is, in general a given spacetime can be described by multiple distinct metrics. We emphasize that the \emph{metrics} really are distinct; it's only by computing ``diff-invariant" quantities that we can check that they describe the same spacetime. 

Another way to phrase this is that the metric doesn't only contain physical information about a spacetime, it \emph{also} contains extra redundant information that doesn't matter to an observer. (Think of the choice of the origin and axes for a flat metric, for example.)  In high-energy physics one often refers to this information as ``gauge'' degrees of freedom. When we go from a metric to the physical quantities, i.e.\ the geodesics, we ``gauge out" these degrees of freedom so that only the physical ones remain. We say that two metrics are ``gauge-equivalent" if they're related by a gauge transformation, i.e.\ there is a diffeomorphism which takes one metric to the other. These two metrics are members of a ``gauge orbit", the equivalence class of all gauge-equivalent metrics. We ``gauge-fix" by specifying information to go to a subset of metrics in the equivalence class; if we specify so much that only one metric remains, we've totally fixed the gauge. A ``gauge'' is just another word for a measuring device; think of gauge-fixing as specifying the properties of this measuring device, i.e. giving enough information that two observers can agree on how to perform a measurement. In Newtonian physics, for example, we'd gauge-fix by fixing the direction of each coordinate axis, and a position and velocity for the origin of the coordinate system. None of that affects the physics, but if you want to check someone else's measurements you'll need that information.

However, it's important to point out that not all gauge transformations preserve all of the information we might call physical. The issue arises when we consider metrics for manifolds with \emph{boundaries}. It's useful to gain some intuition by first thinking about the equivalent case in electrodynamics. Recall that the behavior of charged particles and electromagnetic fields is governed by Maxwell's equations. However, just like the metric, electric and magnetic fields are not gauge-invariant, only certain combinations of them (like $E^2 + B^2$ are). 
Just like Einstein's equations, the Maxwell equations are also differential equations, and hence their solution in a region with a boundary depends on a choice of boundary conditions. We could solve for the behavior of the electromagnetic fields inside a conducting sphere, for example, or with a charged surface. The point is that these boundary conditions represent an additional set of physical information. In general, then, the set of gauge transformations which preserve all physical quantities will be \emph{smaller} than if we didn't worry about boundary conditions at all.

This same issue arises even when we're not placing boundary conditions at some particular region of space, but instead placing them ``at infinity." In this case there is a precise language used to talk about these types of boundary conditions. We ask whether the gauge transformation has any effect at infinity, or, equivalently, if it can be distinguished from the identity transformation in the limit that we go very far away from the origin of our coordinates. If it can't, we call the gauge transformation a ``small gauge transformation". If it can, we call the gauge transformation a ``large gauge transformation." And we have in mind that large gauge transformations are \emph{physical} while small gauge transformations are not. A large gauge transformation in electrodynamics can, for example, change the total charge a distant observer measures enclosed within some radius. A similar story holds in general relativity, but now the gauge transformations are diffeomorphisms applied to metrics. A small gauge transformation of the metric is one that takes $g_{\mu \nu}\rightarrow g_{\mu \nu}+\delta g_{\mu \nu}$, with $\lim_{x^\mu\rightarrow0} \delta g_{\mu \nu}(x^\mu)=0.$ A large gauge transformation can, for example, change the total invariant mass enclosed within some region.

One convenient way to gauge-fix in general relativity is to fix a direction of time everywhere on the spacetime, or equivalently identify points which are on ``the same spatial slice" at a given time. In four spacetime dimensions, this is referred to as a \emph{3 + 1 decomposition}. Geometrically, we can think of this as a foliation\footnote{It turns out that there are some manifolds where it isn't actually possible to do such a foliation, but none of these exotic spacetimes will be relevant for our purposes. For a review of this formalism, which is most important when solving Einstein's equations numerically, see \ref{gourgoulhon20123+}.} of the spacetime into spatial slices. 

Again, when the spacetime has a boundary (in a spacelike direction), some foliations will coincide on the boundary and others will not. It's only the foliations which look the same at the boundary which we think of as describing the same physics. As we'll discuss below, the RT formula applies to spacetimes that have a (spatial) boundary. Invariant quantities are those which are left unchanged by ``small diffeomorphism'', i.e.\ diffeomorphisms which leave the boundary unchanged. In particular, the invariant quantities of interest for the RT formula are the areas of extremal surfaces which end on the boundary. In $2+1$ dimensions, these are geodesics which extend between points on the boundary; in higher dimensions there are also extremal surfaces, volumes, etc. which touch the boundary.

In subsequent sections we will talk about not spacetimes, but \emph{Hilbert spaces} with gauge symmetries and redundancies. Although this type of gauging can be described independently of anything having to do with gravity, we will always have in mind that holographic error-correcting codes should indeed exhibit some version of the gauge symmetries we see in gravity. In particular, our codes will manifest a particular version of the observation that large gauge transformations describe physically distinct spacetimes: we will see that changing the gauge in the holographic code yields a different result when measuring with an ``area operator." See Appendix \ref{app:2x2bacon-shor} for a discussion of the Bacon-Shor code which is phrased in the language of gauges.

\subsection{Towards quantum gravity}

General relativity is a \emph{classical} theory. Just like Newton's laws govern the behavior of massive particles and extended objects moving, accelerating, and applying forces to each other, and Maxwell's electrodynamics govern the coupled behavior of charged objects and electromagnetic fields, Einstein's equations \eqref{eq:einstein} govern the coupled behavior of energy distributions and geometry. By ``govern the behavior", we really mean that given enough data to describe things at an initial time (the position and velocity of particles, the electric and magnetic fields everywhere, the stress-energy tensor and metric), we can use the theory to find a description at a later (or earlier) time.

Quantum mechanics is also a theory in this sense: given a Hamiltonian and an initial wave function, evolution is governed by the Schr\"odinger equation.
If we arrived at the theory by quantizing an initial theory, we can get back the classical quantities by applying the appropriate observables (i.e., Hermitian operators) to the wave function: for example, for the quantum mechanics of a point particle in a potential, position or momentum operators. For reasonable choices of Hamiltonian, the \emph{expectation value} of these observables will evolve smoothly--but when we measure the observable we project onto one of its eigenstates according to the Born rule, and only by repeatedly resetting the system, evolving, and measuring can we actually get access to the expectation value.

Finding a fully consistent theory of quantum gravity is beyond the scope of this review, to put it mildly, but we \emph{can} say some things confidently. In the quantum theory of a point particle, the classical observables (i.e., the observables which reduce to classical quantities in the classical limit)  are the operators which measure position, velocity, etc. The classical observables of a field theory are, similarly, the operators which measure field value and its derivatives. The classical observables of a gravitational theory, then, when applied to states corresponding to classical geometries, measure the metric, stress tensor, etc. So, at minimum, we expect the Einstein equations to hold in some classical limit. That is, the Einstein equations suggest a schematic operator equation
\begin{equation}
    \hat G_{\mu \nu}\: ``="\:8 \pi G_N \hat T_{\mu \nu},
\end{equation}
where the hat indicates that this is an operator expression, and we've put quote marks around the equality to emphasize that it's not really precise. What we really mean by this expression is that, for classical states $\ket{\Psi}$ which are simultaneous eigenstates of both operators,
\begin{equation}
    \hat G_{\mu \nu} \ket{\Psi} = G_{\mu \nu} \ket{\Psi} = 8 \pi G \hat T_{\mu \nu} \ket{\Psi} = 8 \pi G T_{\mu \nu} \ket{\Psi},
\end{equation}
which automatically implies as well that
\begin{equation}
    \left\langle\hat G_{\mu \nu}\right\rangle_\Psi = 8 \pi G \left\langle\hat T_{\mu \nu}\right\rangle_\Psi.
\end{equation}

Again, we emphasize that making these expressions precise is complicated in full quantum gravity. The abundant gauge freedoms in GR which we discussed in the previous subsection mean that the notion of a local operator is itself subtle, for example. Keeping this in mind, we can proceed to move gingerly away from exactly classical states (i.e., exact eigenstates of these operators), in two ways:

\begin{itemize}
    \item We can use perturbation theory to understand the result of measuring operators in states close to classical states. For example, if we have a massive system in some superposition of locations, we can see that the expectation value of the metric is that sourced by the average position of the mass, but that measuring quantities sensitive to the metric (for example, the motion of a test mass passing near the system) will project the wave function onto a state of definite metric (and so the particle will be seen (experimentally! \cite{PhysRevLett.47.979}) to follow a geodesic of this metric, emit gravitational waves quantized as gravitons, etc.). For a given classical geometry, we can use these sorts of techniques to work our way all the way up to the full machinery of quantum field theory in curved space. At the linear, perturbative level, the graviton enters as just another type of field. (To be clear, though, perturbation theory has its limits! We can't use this machinery to fully quantize gravity, which is famously impossible using just the machinery of quantum field theory.)
    \item We can use the linearity of quantum mechanics to discuss not only states close to particular geometries but \emph{superpositions} of distinct geometries.
\end{itemize}

It's important to emphasize that there's a major caveat with this second point: the linearity of quantum mechanics applies to states \emph{in a fixed Hilbert space}. Let's return to the basic example of the quantum-mechanical theory of a single particle in a potential. There's a position operator on this Hilbert space, and we understand how it acts both on eigenstates of position and on general states (because we can write a general state in the position basis). But this doesn't tell us how to act on states of a single particle in a different potential.

Actually, there are ways to give a sensible answer to this question. We could arrange for the particle to move in a given potential by coupling it to another set of degrees of freedom, an ``external field'', so that the potential is recovered for a fixed state of the fields. Then, in this new, larger Hilbert space, we now have a way to talk about a superposition of a particle in one potential with a particle in a different one. But doing so requires us to understand how to embed the original Hilbert space into the larger one. This (finally) is where holography comes in. You might sensibly worry that we could only ever measure the metric, stress-energy, geodesic distance, etc., on a Hilbert space describing states perturbatively close to a single geometry. But holography, as we'll describe next, does much better than this. 

And, to be clear, we have very good reasons to expect that quantum gravity does in fact require us to deal with superpositions of different geometries! As we discussed in the first bullet point above, we can imagine coupling the metric to a quantum degree of freedom--for example, arranging to move a test mass into one of two locations depending on the result of a projective measurement. Even without explicitly arranging for this ourselves, though, there are (at least) two places where nature as we understand it naturally creates superpositions. One is cosmology: in the early universe, quantum variance of an inflating field \cite{Guth:1980zm,Linde:1981mu,Albrecht:1982wi} could have been converted \cite{Polarski:1995jg,Lombardo:2005iz} during the Big Bang into superpositions of different classical configurations of matter, radiation, etc. which ultimately seeded the large-scale structure of the universe. Another is black hole evaporation: according to Hawking's famous calculation \cite{Hawking:1975vcx}, a geometry with a black hole in it can ultimately evolve into a superposition of many possible states which each contain no black hole but rather some collection of matter and radiation, which in turn can source distinct spatial geometries. So, if we want our theory of quantum gravity to describe any of these scenarios, we're going to have to be able to work in a Hilbert space that allows for superpositions of geometries.

\subsection{Holographic theories}

Holographic theories are ones in which a gravitational theory can be described using a different non-gravitational theory ``at the boundary.'' These theories implement the desired feature of the last subsection: we can use them to describe superpositions of states which describe distinct geometries. Unfortunately, in the best-understood examples of holography none of these geometries look anything like our universe: in particular, they have \emph{negative spacetime curvature}, meaning that at large distances the spacetime metric becomes hyperbolic (``asymptotically anti-de Sitter''). Our universe, as best as we can tell, looks like it has \emph{positive} spacetime curvature. So we can't just immediately interpret our universe as a particular state in a holographic Hilbert space. However, holographic theories are nevertheless worth studying not only because they have a nice, well-understood Hilbert space, but also because states in these theories that describe geometries, or superpositions of geometries, have a number of nice properties, not least the RT formula itself.

Like in the previous subsections, but perhaps even more so, our discussion here will only scratch the surface of what is by now a vast literature. Our goal will be to reach the RT formula and its interpretation, in particular, and along the way we will sometimes be heuristic (and we note in a few footnotes places where the main discussion has been imprecise). We recommend that interested readers looking for more comprehensive reviews on the aspects of holography most closely related to quantum information consult, for example, \cite{harlow2016jerusalem,harlow2018tasi} and references therein.

First, we'll say a bit more about how the known examples of holographic theories actually work. Then, finally, we'll be in a position to present the RT formula once and for all. After we do that, we'll take the time to introduce a few last concepts to which it will be useful to refer later in the paper: the geometric notions of the causal and entanglement wedge, and the properties of complementary recovery and radial commutativity.

\subsubsection{The holographic dictionary}

Let's be a little more precise by what it means to describe one theory using another. The fundamental objects in any quantum theory are states and observables. In the last subsection we described how to think about states describing quantum fields on top of a spacetime geometry, or a superposition of spacetime geometries. In particular, when a state describes (a spatial slice of) a classical spacetime geometry satisfying the Einstein equations, it is an eigenstate of certain operators, with eigenvalues given by diffeomorphism-invariant quantities like the length of a geodesic, area of an extremal surface, etc. Then we can compute the expectation values of these operators on states close to these classical ones, and the linearity of quantum mechanics then allows us to compute the expectations on superpositions of near-geometric states.

It was realized in the 1990s by string theorists \cite{maldacena1999large,Gubser:1998bc,Witten:1998qj,Aharony:1999ti} that, for states describing asymptotically anti-de Sitter geometries in $D$ dimensions, the expectation values of \emph{all} of these operators could instead be computed using operators in a non-gravitational $(D-1)$-dimensional theory. In particular, one major result of the known holographic correspondences is that there is a precise dictionary for matching operators in the gravitational theory inserted at points near the AdS boundary to operators in the ``boundary theory'', and a precise prescription \cite{Hamilton:2005ju,Hamilton:2006az,Skenderis:2008dg,Christodoulou:2016nej} for integrating over points on the boundary to reconstruct operators deeper into the bulk of the spacetime. For the purposes of this paper, we won't really need to know about the details of the boundary theory: just the entropies of reduced density matrices constructed from (some of) the states in the theory. However, it's worth mentioning two of their properties. 

First, this type of correspondence could only make sense if the boundary theory at least had the same symmetries as the symmetries of the gravitational theory \emph{at its boundary}. In the language of Subsection \ref{sub:diff} above, these are the small gauge symmetries of diffeorphisms that leave the boundary at spatial infinity unchanged. With a little bit of work, you could stare at a metric that describes hyperbolic space and figure this out---it turns out that the group of transformations that do this is the \emph{conformal group} of angle-preserving tranformations. And, accordingly, the boundary theories are \emph{conformal} field theories. Now you know the reason why another name for the holographic correspondence is ``the AdS/CFT correspondence''!

Second, hyperbolic (and spherical) metrics have a length scale, the ``anti-de Sitter length''. Einstein's equations \eqref{eq:einstein} \emph{also} have a length scale, the Planck scale, which can be derived (in a dimension-dependent way) from the Newton gravitational constant $G_N$. The ratio of these two length scales is a dimensionless number, which also appears in the conformal field theory on the boundary. In spacetimes that look classical, this ratio needs to be very large: it's the ratio between ``cosmological'' scales and ``quantum gravitational'' scales. Accordingly, boundary CFTs that can describe classical-looking geometries have a very large number of fields--they're often referred to as ``large-$N$ CFTs.'' And it's a fact about (non-free) conformal theories that the larger the number of fields, the more strongly the fields couple to each other. So, when used to describe classical geometries, the holographic correspondence relates gravity in asymptotally AdS spacetimes to the behavior of strongly-coupled conformal field theories.

For the purposes of this article we'll only care about fundamentally \emph{gravitational} observables like the lengths of geodesics, etc., whose expectation values on classical states we can compute knowing only the metric. But it's important to understand that this dictionary doesn't apply only to these, it also applies to any other diff-invariant observable built from the \emph{fields} in the theory--for example, the stress tensor on the right-hand side of the Einstein equations, or just the expectation value of a field at some point. 

So far, this might just seem like an interesting coincidence, but no more than that. After all, as discussed in the previous subsection, we already in principle know how to compute these observables for nearly-classical states: we write down a metric describing the geometry on the state, solve the Einstein equations to get the field configurations on top of the geometry, then perturb these field configurations slightly and see how this backreacts on the geometry using the operator form of the Einstein equations.

However, there are a few reasons the existence of a holographic description is exciting:

\begin{itemize}
    \item Sometimes we can compute quantities easily in the gravitational theory but not the non-gravitational theory, or vice-versa. The RT formula itself is an example of this: it's a straightforward mechanical task to compute the areas of extremal curves given the metric, but computing the entropy of a CFT state requires first writing down what the state is, already a nontrivial task, and then doing all the integrals to trace out the state of the fields outside the region of interest.
    \item As we discussed in the last subsection, we're free to compute the expectation values of states that \emph{aren't} nearly-classical. These states need not come anywhere near solving the classical Einstein equations, i.e. they might not be easily described as geometric at all! Yet they live in the same Hilbert space as all the nearly-classical states. We can think of these as \emph{bona fide} quantum-gravitational states! (In fact, historically the logic worked almost the other way around: the holographic dualities were used within string theory to exhibit examples of states that weren't ``stringy'' but had nearly-classical descriptions.)
\end{itemize}

\subsubsection{The RT formula}\label{sub:RT}

Now, at last, let's return to the versions of the RT formula we presented at the start of this section. First, the version \eqref{eq:classical_RT} that applies to a holographic state dual to a classical geometry:
\begin{align}
    S_A(\rho) = \frac{1}{4G_N} \min_{\gamma_{A}} Area(\gamma_{A}).
\end{align}
We've already understood the meaning of the right-hand side from previous sections. If we have a metric already, we can find the geodesics (or extremal surfaces) which hit the boundary at precisely the edge of the boundary region $A$. If there are multiple such geodesics, we choose the one with smallest area\footnote{\label{fn:homology}We mention one caveat for experts: if $A$ is the entire boundary, and $\rho$ is mixed, then it might seem like the RT formula leads us astray, because the boundary of $A$ is the empty set and so any geodesic which is a closed circle hits the boundary at the empty set, i.e. doesn't touch it at all. This puzzle was resolved by realizing that the spacetime dual to a thermal state is \emph{AdS-Schwarzschild}, i.e. hyperbolic spacetime with a black hole of the appropriate temperature sitting in the black hole. Then we get the correct result if we take the minimal surface to be the one which wraps around the horizon of the black hole. To get this, we need to impose a ``homology constraint'': the only geodesics which we consider in the minimization in \eqref{eq:classical_RT} are those which not only meet the boundary at the appropriate place, but those which can be continuously extended through the bulk to touch $A$ without crossing any holes or horizons in spacetime, i.e. those that are ``homologous to A.''}. On the left-hand side, the state $\rho$ lives in the Hilbert space of a large-N CFT, with $N$ related to the Newton constant $G_N$ as discussed above. In principle, we can perform the field-theoretic equivalent of tracing out a subregion, which involves integrating out the values of fields outside the region with appropriately chosen boundary conditions, then take a logarithm to find the entropy. In practice, the computation is usually done by holographers using slicker mathematical techniques to compute the entropy, for example computing $\text{Tr}_A(\rho^n)$ and then obtaining the entropy by taking a limit in $n$. Even so, it's very hard to carry out this procedure except in states close to certain highly symmetric states like the vacuum or a thermal state.

Now let's consider the more general RT formula \eqref{eq:full_RT} which applies to states with entanglement in the bulk and superpositions of geometries:
\begin{align}
    S_A(\rho) = S_{\text{bulk},A}(\rho) + \text{Tr}(L \rho).
\end{align}
$L$ is an operator in the quantum-gravitational Hilbert space. Its eigenstates are the states dual to classical geometries (with appropriate AdS length to be described by the particular CFT under consideration), and its eigenvalues are the areas of the minimal surfaces that meet the boundary subregion. However, roughly speaking\footnote{For a discussion of the limitations of this approach, and the circumstances under which the RT formula breaks down (essentially, when there are very many, exponential in $N^2$, terms in the superposition), see \cite{Almheiri:2016blp}.} there are many different ``field-theoretic'' states that live on the same curved spacetime. The bulk entanglement term identifies which of these states (or rather, which equivalence class of states with the same bulk entanglement inside the extermal surface) is described by $\rho$. Hence neither the left-hand nor the right-hand side of \eqref{eq:full_RT} is the expectation value of an observable, but the difference between the boundary and bulk entropies \emph{is}. Moreover, we can see that, as we will discuss below, $L$ is an operator which can be obtained given access either to the reduced state either in $A$ or to its complement $\bar A$.

\subsubsection{The causal and entanglement wedges}\label{sub:wedges}

Which bulk operators can we reconstruct given access to a particular boundary region? As discussed in Subsection \ref{sub:metrics} above, Lorentzian metrics have \emph{causal structure}, so we know that only those parts of the boundary that can send or receive a signal from the point or region can affect the value of an operator there. This is formalized by the notion of the \emph{causal wedge}, depicted in Figure \ref{fig:subregion}:

\begin{figure}[htb!] 
\centering
\includegraphics[height=8cm]{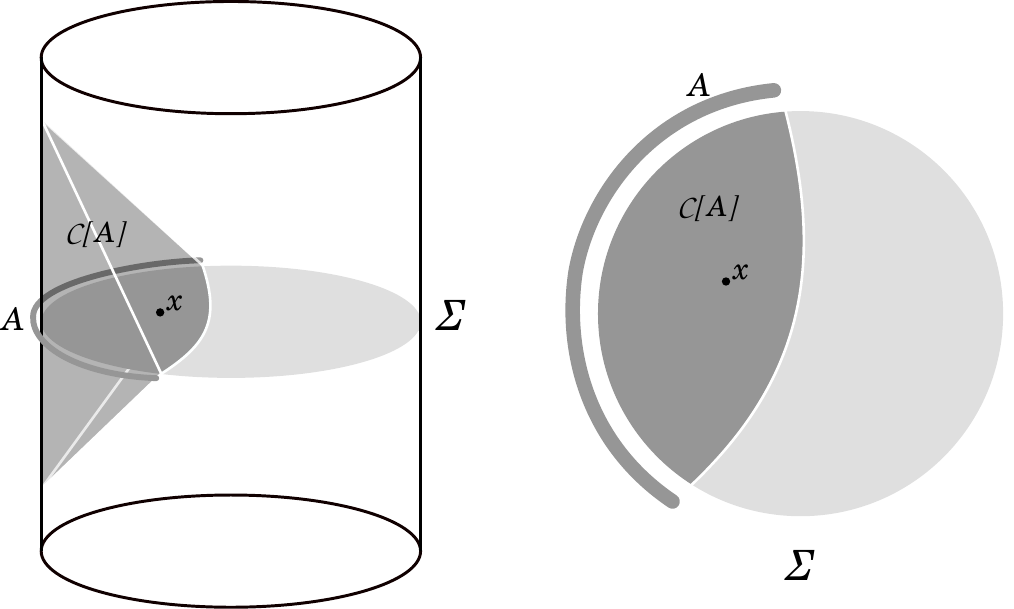}
\caption{(This is a slightly modified version of Figure 11 in~\cite{pastawski2015holographic} and its caption.) Bulk field reconstruction in the causal wedge.  On the left is a spacetime diagram, showing the full spacetime extent of the causal wedge $\mathcal{C}[A]$ associated with a boundary subregion $A$ that lies within a boundary time slice $\Sigma$. The point $x$ lies within $\mathcal{C}[A]$ and thus any operator at $x$ can be reconstructed on $A$.  On the right is a bulk time slice containing $x$ and $\Sigma$, which has a geometry similar to that of our tensor networks. The point $x$ can simultaneously lie in distinct causal wedges, so $\phi(x)$ has multiple representations in the CFT.}\label{fig:subregion}
\end{figure}

On a time-slice of the boundary theory, choose a spatial sub-region $A$.
The \emph{causal wedge} $\mathcal{C}[A]$ of $A$ is the bulk region bounded by (1) the boundary domain of dependence of $A$ (dark grey curve in Fig.~\ref{fig:subregion}) and (2) the set of bulk geodesics which start and end on (1). The casual wedge is determined by the domain of dependence of $A$, hence ``causal.''

Often, especially in static spacetimes, it's convenient not to work with the full causal wedge, but instead some particular spatial slice within it. If the RT surface is spacelike, then every spatial slice in the causal wedge ends on the RT surface itself, but they hit the boundary at different times. It's usually most convenient to choose a spatial slice which intersects the boundary at the region $A$ itself. In nice situations, for example if the spacetime is static, we can pick a spatial slice that extends between $A$ and its RT surface, which by causality contains all of the information necessary to reproduce the entire causal wedge. In this case, as shown on the right diagram in Figure \ref{fig:subregion}, we're free to draw diagrams which suppress the time direction entirely: compare to the figures in the Introduction, which similarly depict the situation at one particular time.

\begin{figure}[htb!] 
\centering
\includegraphics[height=7cm]{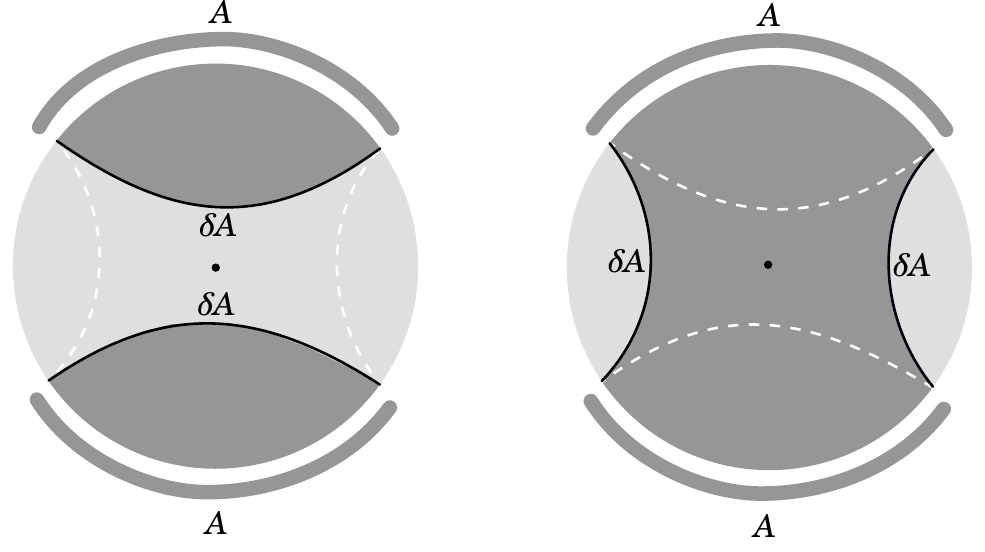}
\caption{(This is a slightly modified version of Figure 14 in~\cite{pastawski2015holographic} and its caption.) The intersection of the entanglement wedge $\mathcal{E}[A]$ with a bulk time-slice, in the case where $A$ has two connected components. Minimal geodesics in the bulk are solid lines.  
When $A$ is smaller than $A^c$, we have the situation on the left and the causal wedge agrees with the entanglement wedge.  When $A$ is bigger, however, the minimal geodesics switch and the entanglement wedge becomes larger. 
In particular the point in the center lies in the $\mathcal{E}[A]$ but not $\mathcal{C}[A]$.}\label{fig:wedges}
\end{figure}

However, we also know that given access to the entropy of a CFT subregion we can use the RT formula to compute the area of the relevant extremal surface. In fact, we expect that if we know not just the entropy but the full reduced density matrix, we can construct, e.g., the RT surface itself. And we can also use this information to construct the RT surfaces of smaller parts of the subregion, so we should be able to read off the metric everywhere in the portion of the bulk between the boundary region and the RT surface. This is formalized by the notion of the \emph{entanglement wedge}:

The \emph{entanglement wedge} $\mathcal{E}[A]$ is the domain of dependence of the bulk region bounded by (1) $A$ (dark grey curve in Fig.~\ref{fig:wedges}) and (2) the minimal extremal bulk surface homologous to $\partial A$ (i.e. the RT surface of $A$) (black curve in Fig.~\ref{fig:wedges}). 
The entanglement wedge is determined by the RT surface, which has area equal to the von Neumann entropy of the part of the boundary theory contained in $A$, hence ``entanglement.''

One can show, under the assumption the bulk theory describes a sensible gravitational spacetime, that the causal wedge is contained in the entanglement wedge \cite{Wall:2012uf,Headrick:2014cta}. Figure \ref{fig:wedges} depicts a situation in which the two wedges do not coincide. Note that, for a pure boundary state, the RT surface of $A$ can clearly seen to be the same as the RT surface of $\bar A$. That is, the operator $L$ which gives its area is both in the set of operators acting only on region $A$ of the boundary theory, and the set of operators acting only on region $\bar A$. But causality dictates that all operators in one of these two sets commute with all operators in the other of these sets. So, $L$ commutes with every operator acting on $A$: we say it's in the \emph{center} of the operators acting on $A$. One such operator is the identity. But, in general, when there is some gauge symmetry in the bulk theory, there will be elements in the center which are \emph{not} trivial, and the area operator will be one of these. So the nontriviality of the area operator tells us about the fact that the bulk is gravitational, and thus has diffeomorphism invariance! Thinking about the algebraic properties of bulk and boundary operators will be key to our approach in the rest of the paper; we'll review the concepts of operator algebras, centers, etc.\ in the next section.

\subsubsection{Complementary recovery and radial commutativity}\label{sub:radial}

Recall that the causal wedge tells us which bulk operators can be reconstructed given access to a boundary region $A$. As Figure \ref{fig:subregion} makes clear, if we divide the boundary into two regions, an operator acting at a particular bulk point must lie in the causal wedge of at least one of the regions, and it only lies in the causal wedge of both when the point is part of the RT surface. This is \emph{complementary recovery}: given a subregion we can reconstruct all operators inside its causal wedge but none of the operators outside it.

Now consider, instead of fixing the region $A$, what happens when we allow it to vary but still keep the bulk operator fixed. In general, we see that an operator lies inside the causal wedge of \emph{many} regions. So, if we had access to a boundary region large enough that many of its subregions could alone reconstruct the operator, knowledge of the operator is \emph{redundantly} encoded in the state: we don't need the full state on $A$ to reconstruct it, and there are many possible ways to reconstruct it. We say the state exhibits \emph{subregion duality}, in which many subregions can be used to reconstruct the same operator. Furthermore, if we erase a piece of the boundary that is much smaller than $A$, almost every bulk operator can still be exactly reconstructed. Historically, it was this sort of code-like redundancy which led to the consideration of holographic error correction.

The flip side of subregion duality is radial commutativity. We can see that, at least for non-pathological spacetimes, the RT surfaces of \emph{small} subregions don't extend deep into the bulk: we need large subregions to penetrate deep into the interior. If a bulk operator is outside of the causal wedge of a subregion, that means, by causality, that it commutes with every operator acting in the casual wedge subregion. In particular, it commutes with the operators that act on the boundary region itself. But every boundary operator acting on a point in the boundary lives inside the causal wedge of any boundary subregion containing that point; in particular, it lives inside the causal wedge of an arbitrarily small subregion around the point. Hence any bulk operator which lives away from the boundary must commute with \emph{all} boundary operators acting on single points in the boundary: this is the property of \emph{radial commutativity}. 

This might not seem to be a problem yet: an arbitrary operator in the boundary \emph{doesn't} act at a single point in the boundary, but at many points. However, field theories, and conformal field theories in particular, have an \emph{operator product expansion}: the product of operators acting at multiple points can be written as a sum of local operators acting only at a single point. And each of these local operators, by the argument above, commutes with the bulk operator! If we take this argument seriously, then, a bulk operator commutes with \emph{every} operator in the boundary theory. This seeming paradox was another motivation behind the introduction of error correction in holography---we only reached this conclusion because we treated bulk operators and boundary operators as acting on the same Hilbert space, but in fact the bulk Hilbert space of a given geometry, as we have seen, is much smaller and redundantly encoded into the CFT Hilbert space. So, in the language of this paper, the resolution can be stated simply: the bulk doesn't live ``inside the boundary,'' i.e. in the same space. Rather, as depicted in Figure \ref{fig:notation}, we must map the bulk into the boundary using an isometry $V$.
\section{Finite-dimensional von Neumann algebras}
\label{sec:vonNeumann}

In Section~\ref{sec:overview_HQEC}, we discussed how the appropriate language to analyse the entropy contributions arising from the bulk degrees of freedom is the one of von Neumann algebras.
In this section, we review some basic notions from the theory of von Neumann algebras (a special case of the more general $C^*$-algebras). Although it is common to study von Neumann algebras over infinite-dimensional Hilbert spaces (and it is in this case that they have proved most useful) we only consider the finite-dimensional case, which is the most relevant for our purposes. Unless otherwise specified, when we use the term von Neumann algebra we always refer to a von Neumann algebra over a finite-dimensional Hilbert space. The content of this section is mostly based on the presentation given in~\cite{harlow2017ryu} which in turn draws from the lecture notes of Jones~\cite{jones2003neumann}. 

An \emph{algebra} over a field is a set which is closed under scalar multiplication, addition and multiplication, and for which there exists a unit element. Von Neumann algebras are algebras of linear operators acting over a complex Hilbert space with the additional property of closure under complex conjugation. More specifically, we have that
\begin{definition}[von Neumann algebra]
Let $\mathcal{L}(\mathcal{H})$ be the set of linear operators over a finite-dimensional complex Hilbert space $\mathcal{H}$. A von Neumann algebra is a subset $\mathcal{M} \subseteq \mathcal{L}(\mathcal{H})$ which is closed under:
\begin{itemize}
    \item (addition) if $A, \, B \in \mathcal{M}$ then $A + B \in \mathcal{M}$;
    \item (multiplication) if $A, \, B \in \mathcal{M}$ then $AB \in \mathcal{M}$; 
    \item (scalar multiplication) if $A \in \mathcal{M}$  and $c \in \mathbb{C}$ then $cA \in \mathcal{M}$; 
    \item (complex conjugation) if $A \in \mathcal{M}$ then $A^\dagger \in \mathcal{M}$;
\end{itemize}
and for which there exists an element $I \in \mathcal{M}$ such that for every $A \in \mathcal{M}$ we have $IA = A$.
\end{definition}

From now on, whenever we use the term algebra we always assume that the algebra is a finite-dimensional von Neumann algebra (sometimes, when extra care is required, we still write the full name explicitly). 
We often define a von Neuman algebra through its generators using the following notation $\mathcal{M} = \langle A, B, \dots \rangle_{\mathrm{vN}}$, where the angle brackets denote the algebra generated by some operators $A, B, \dots$.
Note that the von Neumann algebra generated by a set of operators is different from the group generated by a set of operators as the latter does not have an addition and scalar multiplication operation.
Because in quantum error correction it is customary to use the angle bracket notation to define the group generated by a set of operators we chose to adopt the, bulkier, notation $\langle \dots \rangle_{\mathrm{vN}}$ for von Neumann algebras.
So, for example, $\langle Z \rangle_{vN}$ is the set of all 2 x 2 diagonal matrices ($X,Y,Z$ denote the Pauli matrices) and $\langle Z,X \rangle_{vN} = \mathcal{L}(\mathbb{C}^2)$. 

\begin{example}
The von Neumann algebra $\mathcal{M} = \langle ZII, IXI, IZI \rangle_{\mathrm{vN}}$ over $\mathcal{H}= \mathbb{C}^{8}$, where $X,Z$ are Pauli operators. 
\end{example}

There are three fundamental notions in the study of von Neumann algebras: commutant, center, and factor.

The commutant $\mathcal{M}^{\prime}$ is the set of operators which commute with every element of $\mathcal{M}$. The commutant itself forms a von Neumann algebra. 
\begin{definition}[commutant] Given a von Neumann algebra $\mathcal{M} \subseteq \mathcal{L}(\mathcal{H})$ the commutant is the set
\begin{equation}
    \mathcal{M}^{\prime} \equiv \left \{ B \in \mathcal{L}(\mathcal{H}) \mid \forall A \in \mathcal{M}: AB = BA \right \}.
\end{equation}
\end{definition}

Double commutation (or bicommutation) leaves a von Neumann algebra unvaried. This important property is known as the bicommutant theorem.
\begin{theorem}[bicommutant]
For every von Neumann algebra $\mathcal{M}\subseteq \mathcal{L}(\mathcal{H})$ we have that
\begin{equation}
    \mathcal{M}^{\prime \prime} \equiv (\mathcal{M}^{\prime})^\prime = \mathcal{M}.
\end{equation}
\end{theorem}

The center is the set of commuting elements of an algebra.
\begin{definition}[center] Given a von Neumann algebra $\mathcal{M} \subseteq \mathcal{L}(\mathcal{H})$ the center $Z_{\mathcal{M}}$ is the set
\begin{equation}
    Z_{\mathcal{M}}\equiv \mathcal{M} \cap \mathcal{M}^{\prime}.
\end{equation}
\end{definition}

Trivial centers (i.e. centers that are multiples of the identity) are known as factors.
\begin{definition}[factor] Let $c\in \mathbb{C}$. A von Neumann algebra $\mathcal{M} \subseteq \mathcal{L}(\mathcal{H})$ is a factor if its center satisfies
\begin{equation}
    Z_{\mathcal{M}} = \langle I\rangle_\text{vN} \equiv \{z I\hspace{1mm}|\hspace{1mm} z \in \mathbb{C}\}.
\end{equation}
\end{definition}

\subsection{Classification of von Neumann algebras}
\label{thm:Wedderburn}
The classification theorem shows that any von Neumann algebra can be decomposed as a direct sum of factors.
\begin{theorem}[classification theorem]
For every von Neumann algebra $\mathcal{M}$ on a finite-dimensional Hilbert space $\mathcal{H}$ there exists a block decomposition of the Hilbert space
\begin{equation}
    \mathcal{H} = \left[ \oplus_{\alpha}  \left( \mathcal{H}_{A_\alpha} \otimes \mathcal{H}_{\bar{A}_\alpha}  \right) \right] \oplus \mathcal{H}_0
\end{equation}
such that
\begin{align}
    \label{eq:Wedderburn} 
    \mathcal{M} = \left[ \oplus_{\alpha}  \left( \mathcal{L}(\mathcal{H}_{A_\alpha}) \otimes I_{\bar{A}_\alpha}  \right) \right] \oplus 0, \\
    \mathcal{M}^{\prime} = \left[  \oplus_{\alpha}  \left( I_{A_\alpha} \otimes \mathcal{L}(\mathcal{H}_{\bar{A}_\alpha} ) \right) \right] \oplus 0, \\
    Z_\mathcal{M} = \oplus_{\alpha}  \left( c_{\alpha} I_{A_\alpha} \otimes I_{\bar{A}_\alpha} \right),
\end{align}
where $\mathcal{H}_0$ is the null space and $0$ is the zero operator on $\mathcal{H}_0$. For simplicity, whenever we write a decomposition of an algebra (Hilbert space), we no longer write the direct sum with the null space (zero operator).
The decomposition in \eqref{eq:Wedderburn} is known as the Wedderburn decomposition.
\end{theorem}
Note that, in order to denote the different blocks in the sum, we adopt the heavy notation $\H_{A_\alpha}$---and not the simpler $\H_{\alpha}$---to ensure consistency with the notation of Lemma~\ref{lemma:factorization}. In that case, the letter $A$ is used to denote a partition of the Hilbert space. In this section the letter $A$ has no other meaning but to denote one of the two factors of a block.

We now proceed to give a series of examples of increasing generality of the classification theorem. We begin with the special case of a factor algebra over $\mathcal{H}$ that is equivalent to $\mathcal{L}(\mathcal{H})$.
\begin{example}
The von Neumann algebra over $\mathcal{H} =  \mathbb{C}^2$ with Wedderburn decomposition
\begin{equation}
    \mathcal{M} = \mathcal{L}(\mathbb{C}^2) \otimes 1 = \mathcal{L}(\mathbb{C}^2) =  \begin{bmatrix}
a & b \\ 
c & d
\end{bmatrix},
\end{equation}
where $a, \dots, d \in \mathbb{C}$, is a factor.

The commutant of the algebra is $\mathcal{M}^\prime = \langle I\rangle_\text{vN}$.
\end{example}

The following is an example of a factor algebra over $\mathcal{H}$ that is strictly contained in $\mathcal{L}(\mathcal{H})$.
\begin{example}
The von Neumann algebra over $\mathcal{H} =  \mathbb{C}^4$ with Wedderburn decomposition
\begin{equation}
    \mathcal{M} = \mathcal{L}(\mathbb{C}^2) \otimes I = \begin{bmatrix}
a & 0 & b & 0 \\ 
0 & a & 0 & b \\ 
c & 0 & d & 0 \\ 
0 & c & 0 & d
\end{bmatrix},
\end{equation}
where $a, \dots, d \in \mathbb{C}$,is a factor.

The commutant of the algebra is $\mathcal{M}^\prime = I \otimes \mathcal{L}(\mathbb{C}^2)$.
\end{example}

Finally, we give two examples of algebras that are not a factor. The first is a fully diagonal algebra (and, in the language of quantum mechanics, can be thought of describing a classical algebra of observables) while the second has a block diagonal structure (thus describing a quantum algebra of observables). For many more examples of von Neumann algebras, Weddernburn decompositions, and their relationship to coarse-graining and decoherence the interested reader can consult \cite{Kabernik:2019jko}.
\begin{example}
The von Neumann algebra over $\mathcal{H} =  \mathbb{C}^2$ generated by the Pauli $Z$ operator has the following Wedderburn decomposition
\begin{equation}
    \mathcal{M} = \langle Z \rangle_{vN} = \begin{bmatrix}
a & 0 \\ 
0 & b
\end{bmatrix}
=
\left[ \mathcal{L}(\mathbb{C})\otimes 1 \right] \oplus \left[ \mathcal{L}(\mathbb{C})\otimes 1 \right],
\end{equation}
where $a,b \in \mathbb{C}$.
\end{example}

\begin{example}
\label{example:Wedderburn}
The von Neumann algebra $\mathcal{M} = \langle ZII, IXI, IZI \rangle_{\mathrm{vN}}$ over $\mathcal{H}= \mathbb{C}^{8}$ has the following Wedderburn decomposition
\begin{equation}
\mathcal{M} = \bigoplus_{\alpha=0} ^1 \left (\mathcal{L}(\mathbb{C}^2) \otimes I \right) =
\begin{pmatrix}
  \begin{matrix}
a & 0 & b & 0 \\ 
0 & a & 0 & b \\ 
c & 0 & d & 0 \\ 
0 & c & 0 & d 
  \end{matrix}
  & \rvline & \bigzero \\
\hline
  \bigzero & \rvline &
  \begin{matrix}
e & 0 & f & 0\\ 
0 & e & 0 & f \\ 
g & 0 & h & 0\\ 
 0 & g & 0 & h
  \end{matrix}
\end{pmatrix},
\end{equation}
where $a, \dots, h \in \mathbb{C}$.
\end{example}

\subsection{Algebraic states and entropies}
\label{sec:algebraic_states}

A quantum state on a Hilbert space $\mathcal{H}$ is a Hermitian positive semi-definite operator $\rho \in \mathcal{L}(\mathcal{H})$ with $\operatorname{Tr}(\rho) =1$. 
Given a state $\rho$ and a Hermitian operator $O$ we can define the expectation value of the operator $O$ on $\rho$ as 
\begin{equation}
    \mathbb{E}_\rho (O) = \operatorname{Tr} (O \rho). 
\end{equation}
It is often the case that one is interested in computing expectation values of operators that form an algebra $\M$. A generic state $\rho$ is not necessarily an element of $\M$ and could contain more information than is needed to compute expectation of values of operators in $\M$. It is therefore useful to define the notion of an algebraic state---that is, the state that is ``visible'' from an algebra $\M$.
For an algebra $\M$ and quantum state $\rho$ we denote the respective algebraic state by $\rho_\M$. 
The following theorem shows that algebraic states are unique and that, for the purpose of computing expectation values of operators in $\M$, we can always replace $\rho$ by $\rho_\M$. That is, the algebraic state is a generalization of the reduced density matrix for an algebra which need not be a factor.
\begin{theorem}
\label{thm:expectations_vN}
Let $\mathcal{M}$ be a von Neumann algebra on $\mathcal{H}$ and let $\rho \in \mathcal{H}$ be a quantum state. Then, there exists a unique state $\rho_\mathcal{M} \in \mathcal{M}$ such that 
\begin{equation}
    \operatorname{Tr}(O \rho_\mathcal{M}) = \operatorname{Tr}(O \rho)
\end{equation}
for all $O \in \mathcal{M}$.
\end{theorem}

For an algebra $\M$ and state $\rho$ it is possible to write an explicit formula for the algebraic state $\rho_\M$.
Recall that by Theorem~\ref{thm:Wedderburn} there exists a decomposition of the Hilbert space
\begin{equation}
\label{eq:Hdec}
    \mathcal{H} = \oplus_{\alpha}  \left( \mathcal{H}_{A_\alpha} \otimes \mathcal{H}_{\bar{A}_\alpha}  \right),
\end{equation}
in terms of which we can write the Wedderburn decomposition of the algebra 
\begin{equation}
    \mathcal{M} = \oplus_{\alpha}  \left( \mathcal{L}(\mathcal{H}_{A_\alpha}) \otimes I_{\bar{A}_\alpha}  \right).
\end{equation}
Let $\{\ket{\alpha,i,j}\}$ be an orthonormal basis for $\H_{A_\alpha} \otimes \H_{\bar{A}_\alpha} $ (a block in the decomposition) that is ``compatible with $\M$'', that is, the $\alpha$ enumerates the diagonal blocks and within each block we have $\ket{\alpha,i,j} = \ket{i_\alpha}_{A_{\alpha}} \otimes \ket{j_\alpha}_{{\bar A_\alpha}}$ where $\{\ket{i_\alpha}_{A_{\alpha}} \}$ and $\{ \ket{j_\alpha}_{\bar A_\alpha} \}$ are orthonormal bases for $\H_{A_\alpha}$ and $\H_{\bar A_\alpha}$ respectively.
Any state $\rho$ can be written in terms of the Hilbert space decomposition of \eqref{eq:Hdec} as
\begin{equation}
\rho = \sum_{\alpha,\alpha'}\sum_{i,j}\sum_{i',j'} \rho[\alpha,\alpha']_{i,j,i',j'} \ket{\alpha,i,j}\bra{\alpha',i',j'},
\end{equation}
where $\{\ket{\alpha,i,k} \}$ is a basis for the $\alpha$-block.
Because for the purpose of computing expectation values of elements of $\M$ only the blocks that are diagonal in $\alpha$ will give non-zero contributions we have that $\rho[\alpha,\alpha'] = 0$ for all $\alpha \neq \alpha'$.
For computational purposes, it is then  useful to define the blocks of $\rho$ that are diagonal in $\alpha$ as 
\begin{equation}
    \rho_{A_\alpha} \equiv \frac{1}{p_\alpha} \operatorname{Tr}_{\bar A_\alpha} (\rho[\alpha]),
\end{equation}
where $p_\alpha \equiv \sum_{i,j} \rho [\alpha, \alpha] _{i,j,i,j}$ is a positive normalisation constant such that $\operatorname{Tr}_{ A_\alpha} (\rho_{A_\alpha}) = 1$ and $\rho[\alpha] \equiv \rho[\alpha, \alpha]$ is the part of $\rho$ which is in the $\alpha$-block.
Using this notation we can write the algebraic state $\rho_\M$ as
\begin{equation}
\label{eq:algebraic_state}
    \rho_\M \equiv \oplus_\alpha \left( p_\alpha \rho_{A_\alpha} \otimes \frac{I_{\bar{A}_\alpha}}{|I_{\bar{A}_\alpha|}}\right).
\end{equation}

From \eqref{eq:algebraic_state} we can see that when $\M$ is a factor the von Neumann entropy of the algebraic state is equivalent to the von Neumann entropy of the reduced state $\rho_A = \operatorname{Tr}_{\bar A} (\rho)$.
This suggests the following generalisation of the von Neumann entropy for a general quantum state $\rho$ and an arbitrary algebra $\M$:
\begin{definition}\textbf{Algebraic entropy.}
Let $\rho$ be a state on an arbitrary von Neumann algebra $\mathcal{M}$. The algebraic entropy of $\rho$ with respect to $\mathcal{M}$ is
\begin{equation}
\label{eq:algebraic_entropy} 
S(\rho, \mathcal{M}) \equiv-\sum_{\alpha} \operatorname{Tr}_{A_{\alpha}}\left(p_{\alpha} \rho_{A_{\alpha}} \log \left(p_{\alpha} \rho_{A_{\alpha}}\right)\right)=-\sum_{\alpha} p_{\alpha} \log p_{\alpha}+\sum_{\alpha} p_{\alpha} S\left(\rho_{A_{\alpha}}\right), 
\end{equation}
where $S\left(\rho_{A_{\alpha}}\right) \equiv - \operatorname{Tr}_A (\rho_A \log \rho_{A_\alpha})$ is the von Neumann entropy of the reduced state $\rho_{A_\alpha}$. 
\end{definition}
Note that when $\mathcal{M}$ is a factor the algebraic entropy reduces to the standard von Neumann entropy (i.e. the classical term $-\sum_{\alpha} p_{\alpha} \log p_{\alpha}$ vanishes).

The definition in \eqref{eq:algebraic_entropy} has two terms: a \emph{classical} term arising from the uncertainty over which block of the Wedderburn decomposition the state is in, and a \emph{quantum} term associated to the standard von Neumann entropies over the blocks.

\begin{example}
Consider the von Neumann algebra of Example~\ref{example:Wedderburn}. 
The algebra has two diagonal blocks denoted by $\alpha = 0,1$.
Consider the $3$-qubit GHZ state $\ket{\Psi} = 2^{-1/2} (\ket{000} + \ket{111})$. We have that
\begin{equation}
    \rho_{A_0} = \begin{bmatrix}
1 & 0 \\ 
0 & 0
\end{bmatrix},
\quad
\rho_{A_1} = \begin{bmatrix}
0 & 0 \\ 
0 & 1
\end{bmatrix}
\end{equation}
and $p_0 = 1/2$, $p_1 = 1/2$.
The algebraic entropy of the state is
\begin{equation}
    S(\ket{\Psi}\bra{\Psi}, \mathcal{M}) = 1.
\end{equation}
\end{example}

\section{Complementary recovery and the RT formula}
\label{sec:complementary_recovery}

In this section we define several holographic properties of quantum error-correcting codes and establish some relationships between them. The main goal is the RT formula, which is a remarkable relationship between the entropy of a subregion of the boundary $A$, called $S_A$, as well as the entropy $S_{\text{bulk},A}$ of the bulk degrees of freedom visible from $A$:
\begin{align}
    S_A(\rho) = S_{\text{bulk},A}(\rho) + \text{Tr}(L \rho).
\end{align}
For a holographic quantum error-correcting code, the above holds for any encoded state $\rho$.

Such a relationship imposes a lot of structure on the family of states in the code: entropies are non-linear functions of $\rho$, whereas the rightmost term $\text{Tr}(L \rho)$ is a linear function. If $\rho$ is pure, we can intuitively think of $S_A$ as the entanglement entropy of the state encoded into physical qubits, whereas $S_{\text{bulk},A}$ is like the entanglement entropy of the underlying logical state. Rearranging the equation to $S_A(\rho) - S_\text{bulk}(\rho) = \text{Tr}(L \rho)$, we can see that this is essentially saying that the extra entropy added by the encoding process is linear in $\rho$. Furthermore, the amount of entropy is added is an observable: a hermitian `area operator' $L$.

While such a structured relationship might seem very rare, we find that there is actually a fairly simple and natural property that implies it: complementary recovery. This property demands a certain symmetry of the error-correcting code across a bipartition $A$-$\bar A$ of the physical Hilbert space. The errors correctable given only access to $A$ are exactly those that commute with the ones correctable only from subregion $\bar A$. This symmetry is present in many quantum error-correcting codes, such as stabilizer codes (See Lemma~\ref{lemma:stabm}). Surprisingly, it immediately implies an RT formula!

\subsection{Complementary recovery}

We begin with a discussion of complementary recovery and its relationship to quantum error correction. A quantum error-correcting code can be thought of as a subspace $\H_\text{code}$ of the physical Hilbert space $\H$. However, in this discussion as well as in the next section we will find it more convenient to work with a `logical space' $\H_L$ with the same dimension as $\H_\text{code}$, which is thought of as separate from $\H$. Then, an `encoding isometry' $V: \H_L \to \H$ takes logical states and encodes them in the physical Hilbert space. The image of $V$ is $\H_\text{code}$. Intuitively one can think of this as fixing a basis for the code space, since $\H_\text{code}$ is invariant under a basis change $V \to VU_L$ for some unitary $U_L$ on $\H_L$. While different from the approach of other literature, this view has two advantages. First, it makes the notion of a commutant of a von Neumann algebra in $H_L$ a little easier to understand. Second, we find that when giving explicit examples in the next section it is easier to write down $V$ rather than $\H_\text{code}$.

Above we have been speaking of holographic properties of a code, defined by its encoding isometry $V$. However, two other quantities are important for an RT formula: the subregion $A$ that determines the entropy $S_A$, and the visible bulk degrees of freedom that determine the entropy $S_{\text{bulk},A}$. What degrees of freedom are visible is denoted by a von Neumann algebra $\M$. Clearly, $(V,A,\M)$ are interrelated, so we establish the following vocabulary:

\begin{definition} Say $V : \H_L \to \H$ is an encoding isometry $V$ for some quantum error-correcting code, and $A$ is a subregion of $\H$ inducing the factorization $\H = \H_A \otimes \H_{\bar A}$. A von Neumann algebra $\M \subseteq \mathcal{L}(\H_L)$ is said to be:
    \begin{itemize}
        \item \textbf{correctable} from $A$ with respect to $V$ if $\M \subseteq V^\dagger (\mathcal{L}(\H_A)\otimes I_{\bar A})V$. That is: for every $O_L \in \M$ there exists an $O_A \in \mathcal{L}(\H_A)$ such that $O_L = V^\dagger (O_A \otimes I_{\bar A}) V$.
        \item \textbf{private} from $A$ with respect to $V$ if $V^\dagger (\mathcal{L}(\H_A)\otimes I_{\bar A})V \subseteq  \M'$. That is: for every $O_A \in \mathcal{L}(\H_A)$ it is the case that $ V^\dagger (O_A \otimes I_{\bar A}) V$ commutes with every operator in $\M$.
    \end{itemize}
\end{definition}

If $\M$ is correctable, then it is a set of logical operators that can be performed on the encoded state given access to only the subregion $A$. A hermitian element in $\M$ then corresponds to an observable on the logical Hilbert space that could be measured from $A$, so, intuitively, $\M$ tells us about what parts of the logical state are recoverable from $A$. Conversely, if $\M$ is private then the observables in $\M$ tell us what parts of $\rho$ are invisible from $A$.

The notion of correctability is central to complementary recovery: a von Neumann algebra $\M$ exhibits complementary recovery if it can be corrected from $A$, and its commutant $\M'$ can be corrected from $\bar A$.

    \begin{definition} A code with encoding isometry $V:\H_L\to \H$, a subregion of the physical Hilbert space $A$ and a von Neumann algebra $\M \subseteq \mathcal{L}(\H_L)$, together $(V,A,\M)$, exhibit \textbf{complementary recovery} if:
    \begin{itemize}
        \item $\M$ is correctable from $A$ with respect to $V$: $\M \subseteq V^\dagger (\mathcal{L}(\H_A)\otimes I_{\bar A})V$,
        \item $\M'$ is correctable from $\bar A$ with respect to $V$: $\M' \subseteq V^\dagger (I_A \otimes \mathcal{L}(\H_{\bar A}))V$.
    \end{itemize}
\end{definition}

So far, there do not appear to be very many restrictions on the von Neumann algebra $\M$. In particular if $\mathcal{N}$ is a subalgebra of $\M$, and $\M$ is correctable, then $\mathcal{N}$ is correctable as well. It would thus seem plausible that if $(V,A,\M)$ has complementary recovery, then so does $(V,A,\mathcal{N})$, so there are multiple von Neumann algebras with complementary recovery. However, we will find that complementary recovery is actually so restrictive on $\M$ that it determines it uniquely, and subalgebras of $\M$ do not have complementary recovery.

This is important because the von Neumann algebra plays a key role in the RT formula: it tells us how to concretely define `the entropy of bulk degrees of freedom visible from $A$' via an algebraic entropy $S(\M,A)$. For this to make sense, $\M$ must be completely determined by the isometry $V$ and the subregion $A$.

To prove this result we cite a helpful lemma from the quantum error correction literature.

\begin{lemma} \label{lemma:corrpriv} \textbf{Correctable from $A$ $\leftrightarrow$ private from $\bar A$.} A von Neumann algebra $\M$ is correctable from $A$ with respect to $V$ if and only if $\M$ is private from $\bar A$ with respect to $V$. 
\end{lemma}

This lemma establishes the complementarity of correctable and private algebras for the case of erasure errors (the correctability from $A$ implies that $\bar A$ has been erased).
Informally, a subsystem $B$ of a Hilbert space $\mathcal{H} = A \otimes B$ is private if it completely decoheres after the action of a channel.

The lemma follows from the more general Theorem~\ref{thm:correctable-private} (which we present in Appendix~\ref{app:privacy}) that applies to the case of general error channels (to get the lemma simply consider $\mathcal{E}$ to be the erasure channel for the subsystem $\bar A$ and $P$ a projection onto the code subspace).

The complementarity theorem was first proven for the case of factor algebras \cite{kretschmann2008complementarity} and then extended to general, infinite dimensional, von Neumann algebras~\cite{crann2016private}.

Now we are ready to demonstrate that $\M$ is unique, provided it exists at all. The theorem below also shows an easy way to calculate $\M$ as well as a simple criterion for its existence. The proof relies on the fact that privacy of $\M'$ is defined as a statement that is a bit like an `upper bound version' of correctability of $\M$: it demands that the set of correctable operators lies in $\M$, rather than that $\M$ is correctable. By sandwiching together correctability of $\M$ and privacy of $\M'$ we fix what $\M$ must be.

\begin{theorem}\label{thm:whatalgebra} \textbf{Uniqueness of the Neumann algebra.} Say $V$ is an encoding isometry and say $A$ is a subregion. Let $\M := V^\dagger (\mathcal{L}(\H_A)\otimes I_{\bar A})V$ be the image of operators on $\H_A$ projected onto $\H_L$. If $\M$ is a von Neumann algebra (that is, it is closed under multiplication), then it is the unique von Neumann algebra satisfying complementary recovery with $V$ and $A$. If it is not, then no von Neumann algebra satisfying complementary recovery exists.
\end{theorem}
\begin{proof} We split this proof into two conditions:
    \begin{description}
        \item[Existence] If $\M := V^\dagger (\mathcal{L}(\H_A)\otimes I_{\bar A})V$ is a von Neumann algebra, then $(V,A,\M)$ have complementary recovery.
        \item[Uniqueness] If $\mathcal{N} \subsetneq  V^\dagger (\mathcal{L}(\H_A)\otimes I_{\bar A})V$ is a von Neumann algebra, then $(V,A,\mathcal{N})$ do not have complementary recovery.
    \end{description}

    We begin with existence: we assume that  $ \M := V^\dagger (\mathcal{L}(\H_A)\otimes I_{\bar A})V$  is a von Neumann algebra, so $\M'$ is well defined. The first condition of complementary recovery holds by definition of $\M$. Also by definition we have that:
    \begin{align}
        V^\dagger (\mathcal{L}(\H_A)\otimes I_{\bar A})V \subseteq \M = \M'',
    \end{align}
    where in the last part we used the bicommutant theorem. Thus, by definition of privacy, $\M'$ is private from $A$ with respect to $V$. By Lemma~\ref{lemma:corrpriv}, $\M'$ is thus correctable from $\bar A$ with respect to $V$, which is the second condition of complementary recovery.

    Next we show uniqueness. Let $\mathcal{N} \subsetneq  V^\dagger (\mathcal{L}(\H_A)\otimes I_{\bar A})V$ be any von Neumann algebra that is correctable from $A$, but not equal to the full set of correctable operators. We assume that $(V,A,\mathcal{N})$ has complementary recovery and derive a contradiction.
    
    By the second condition of complementary recovery, $\mathcal{N}$ is correctable from $\bar A$ with respect to $V$. By Lemma~\ref{lemma:corrpriv}, $\mathcal{N}'$ is thus private from $A$ with respect to $V$, that is:
    \begin{align}
V^\dagger (\mathcal{L}(\H_A)\otimes I_{\bar A})V \subseteq \mathcal{N}'' =  \mathcal{N},
    \end{align}
    where in the last part we used the bicommutant theorem. So we have
    \begin{align}
       \mathcal{N} \subsetneq  V^\dagger (\mathcal{L}(H_A)\otimes I_{\bar A})V \subseteq  \mathcal{N},
    \end{align}
    implying that $\mathcal{N}$ does not contain itself - a contradiction.
\end{proof}

While an RT formula seems like an extremely unlikely property, complementary recovery on the other hand seems like a property that is rather natural and that most, if perhaps all, quantum error-correcting codes should have. Thus, the fact that complementary recovery implies an RT formula is surprising.

However, the fact that a von Neumann algebra with complementary recovery can fail to exist implies that complementary recovery is actually less trivial than it might seem. While still exhibited by many quantum error-correcting codes, it is worth giving an explicit example of a code without complementary recovery.

\begin{example} \label{ex:badcode} \textbf{A code without complementary recovery.} \\ Let $\H = \text{span}(\ket{00},\ket{01},\ket{10},\ket{11})$ be two qubits and $\H_L = \text{span}(\ket{0},\ket{1},\ket{2})$ be a qutrit. Let $A$ be the first qubit of $\H$, and let:
    \begin{align}
        V = \ket{00}\bra{0} + \ket{01}\bra{1} + \ket{10}\bra{2}.
    \end{align}
    Then the set of correctable operators is:
\begin{align}
    V^\dagger (\mathcal{L}(\H_A)\otimes I_{\bar A})V  = \begin{bmatrix} a & b & 0 \\ c & d & 0 \\ 0 & 0 & a\end{bmatrix} \text{ for all }a,b,c,d \in \mathbb{C}.
\end{align}
    Notably, this set is not closed under multiplication and is not a von Neumann algebra. Let us pick $\M$ to be the largest von Neumann algebra in this set:
\begin{align}
     \M = \begin{bmatrix} a & 0 & 0 \\ 0 & d & 0 \\ 0 & 0 & a\end{bmatrix} \text{ for all }a,d \in \mathbb{C}.
\end{align}
    While $(V,A,\M)$ satisfy the first condition of complementary recovery, they do not satisfy the second:
    \begin{align}
        \begin{bmatrix} a & 0 & b \\ 0 & d & 0 \\ c & 0 & e\end{bmatrix} = \M' \not\subseteq V^\dagger (I_{A} \otimes \mathcal{L}(\H_{\bar A}))V = \begin{bmatrix} a & 0 & b \\ 0 & a & 0 \\ c & 0 & e\end{bmatrix}. 
    \end{align}
    Say we had chosen $\M$ to be some smaller subalgebra of $V^\dagger (\mathcal{L}(\H_A)\otimes I_{\bar A})V$. Then $\M'$ would only be larger, containing the $\M'$ in the line above. But since that $\M'$ is already not contained in $V^\dagger (I_{A} \otimes \mathcal{L}(\H_{\bar A}))V$, there does not exist \emph{any} von Neumann algebra with complementary recovery with $V$ and $A$.
\end{example}

The fact that complementary recovery can fail to exist should illustrate that it actually imposes a non-trivial constraint on the quantum error-correcting code. This constraint is strong enough to imply an RT formula, which we will now define carefully. Note that it is not obvious at all how to obtain an $\M$ that makes the RT formula work from the definition of the formula itself - that is where complementary recovery comes in.

\subsection{The RT formula and its properties}

\begin{definition} Say $V$ is an encoding isometry, say $A$ is a subregion, and say $\M$ is a von Neumann algebra on $\H_L$. Then we say $(V,A,\M)$ have an \textbf{RT formula} if there exists an \textbf{area operator} $L \in \mathcal{L}(\H_L)$ such that for any state $\rho$ on $\H_L$:
    \begin{align}
        S(\text{Tr}_{\bar A}(V\rho V^\dagger)) = S( \M, \rho) + \text{Tr}(\rho  L).
    \end{align}
    If $L \propto I$ then we say $(V,A, \M)$ have a trivial RT formula.
\end{definition}

Now we show the connection between complementary recovery and the existence of the RT formula. This is a highly non-trivial claim that makes use of an enormous amount of structure implied by complementary recovery. Recall from the previous section that a von Neumann algebra implies a Wedderburn decomposition on the Hilbert space that it acts on. We find that when a von Neumann algebra is correctable from $A$ with respect to $V:\H_L \to \H$, then not only does $\H_L$ decompose, but the Hilbert space associated with the subregion $A$ also decomposes. Furthermore, these decompositions are directly related.

The following lemma formalizes this structure, even when complementary recovery is not present. Recall that complementary recovery really implies the correctability of both $\M$ and $\M'$, which allows us to invoke the lemma below not once but twice. We then exploit this to prove that an RT formula exists.

\begin{lemma} \label{lemma:factorization} \textbf{Factorization of encoded states.} Say $V : \H_L \to \H$ is an encoding isometry, $A$ is a subregion inducing $\H = \H_A \otimes \H_{\bar A}$, and say $\M$ is a von Neumann algebra on $\H_L$ that is correctable from $A$ with respect to $V$.

    Say $\M$ induces the decomposition $\H_L = \bigoplus_{\alpha} \left( \H_{L_\alpha} \otimes \H_{\bar L_\alpha}\right) $ so that \begin{equation}
    \label{eq:algebra_4_lemma}
        \M = \bigoplus_{\alpha} \left(  \mathcal{L}(\H_{L_\alpha}) \otimes I_{\bar L_\alpha}  \right),
    \end{equation} and that $\{\ket{\alpha,i,j}\}$ is an orthonormal basis for $\H_L$ that is ``compatible with $\M$'', that is, the $\alpha$ enumerates the diagonal blocks and within each block we have $\ket{\alpha,i,j} = \ket{i_\alpha}_{L_{\alpha}} \otimes \ket{j_\alpha}_{{\bar L_\alpha}}$ where $\{\ket{i_\alpha}_{L_{\alpha}} \}$ and $\{ \ket{j_\alpha}_{\bar L_\alpha} \}$ are orthonormal bases for $\H_{L_\alpha}$ and $\H_{\bar L_\alpha}$ respectively.

    Then there exists a factorization $\H_A = \bigoplus_\alpha\left( \H_{A^\alpha_1} \otimes \H_{A^\alpha_2}  \right) \oplus \H_{A_3}$ and a unitary $U_A$ on $\H_A$ such that the state $(U_A \otimes I_{\bar A}) V \ket{\alpha,i,j}$ factors as follows:
    \begin{align}
    \label{eq:factorisation_encoded_state}
        (U_A \otimes I_{\bar A}) V \ket{\alpha,i,j} = \ket{\psi_{\alpha,i}}_{A^\alpha_1} \otimes \ket{\chi_{\alpha,j}}_{A^\alpha_2 \bar A},
    \end{align}
    where the state $\ket{\psi_{\alpha,i}}$ is independent of $j$, and $\ket{\chi_{\alpha,j}}$ is independent of $i$.
\end{lemma}
\begin{proof}

The full proof for the general case of the algebra in Eq.~\ref{eq:algebra_4_lemma} is given in~\cite[Section 5.1]{harlow2017ryu}.
We give a proof  for the simpler case of a factor algebra in Appendix~\ref{app:structure_lemma}. 
Both proofs follow a similar strategy---originally developed in~\cite{schumacher1996quantum}---that involves introducing a reference system $R$ which is maximally entangled with the region $A$.
By analysing the von Neumann entropies of the reduced density matrices of the $RA\bar A$ system one can obtain a necessary and sufficient condition for quantum error correction. 
This condition and standard properties of the Schmidt decomposition give the proof of the lemma.
We note that an alternative proof of this result can be obtained via a result---see~\cite[Section VI]{hayden2004structure}---that shows that states that saturate the strong subadditivity inequality for the von Neumann entropy can be decomposed as direct sums of tensor products.
\end{proof}

The above lemma already sets up an enormous amount of notation, and even more notation will be required to apply it to a complementary situation. Explicit expressions for these Hilbert space decompositions quickly become rather cumbersome, which is why much of the literature skips many steps in the derivations in order to focus on the intuitive interpretation. While intuition is key, an explicit calculation can also help make one's understanding more concrete. For this reason we give the following derivation with more detail. In the next section we will provide explicit examples of quantum error-correcting codes and analyse them in the same language established here. The reader may wish to skip the proof of the following theorem and read the examples in the next section first.

The following derivation is inspired by proofs in \cite{harlow2017ryu} and \cite{almheiri2015bulk}.

\begin{theorem} \label{thm:complementaritytoRT} \textbf{Complementary recovery implies a two-sided RT formula.} Consider an encoding isometry $V$, a subregion $A$ and a von Neumann algebra $\M$ so that $(V,A,\M)$ have complementary recovery. Then $(V,A,\M)$ and $(V,\bar A,\M')$ both have an RT formula with the same area operator $L$ (that is, the RT formula is `two-sided'). Furthermore, $L$ is in the center $Z_\M$.
\end{theorem}
\begin{proof} Say $\M$ induces the decomposition $\H_L = \bigoplus_{\alpha} \left( \H_{L_\alpha} \otimes \H_{\bar L_\alpha}\right) $. This way we can decompose $\M$ and $\M'$ together as:
    \begin{align}
        \M =  \bigoplus_{\alpha} \left(  \mathcal{L}(\H_{L_\alpha}) \otimes I_{\bar L_\alpha}  \right),  \hspace{1cm} \M' =  \bigoplus_{\alpha} \left( I_{L_\alpha}  \otimes \mathcal{L}(\H_{\bar L_\alpha})  \right) .
    \end{align}
    Let $\{\ket{\alpha,i,j}\}$ be a basis that is ``compatible with $\M$'' as in Lemma~\ref{lemma:factorization}. We observe that $\{\ket{\alpha,i,j}\}$ also `lines up with $\M'$' in the same sense, since really $\{\ket{\alpha,i,j}\}$ just lines up with the underlying decomposition of $\H_L$.

    Now, with two applications of Lemma~\ref{lemma:factorization} we know that there exist factorizations of $\H_{A}$ and $\H_{\bar A}$ of the form:
    \begin{align}
        \H_A = \bigoplus_\alpha\left( \H_{A^\alpha_1} \otimes \H_{A^\alpha_2}  \right) \oplus \H_{A_3}, \hspace{1cm} \H_{\bar A} = \bigoplus_\alpha\left( \H_{\bar A^\alpha_1} \otimes \H_{\bar A^\alpha_2}  \right) \oplus \H_{\bar A_3},
    \end{align}
so that there are unitaries $U_A$ and $U_{\bar A}$ such that:
    \begin{align}
        (U_A \otimes I_{\bar A}) V \ket{\alpha,i,j} = \ket{\psi_{\alpha,i}}_{A^\alpha_1} \otimes \ket{\chi_{\alpha,j}}_{A^\alpha_2 \bar A}\\
        (I_A \otimes U_{\bar A}) V \ket{\alpha,i,j} = \ket{\bar \chi_{\alpha,i}}_{A \bar A^\alpha_2} \otimes \ket{\bar \psi_{\alpha,j}}_{\bar A^\alpha_1}.
    \end{align}

    If we consider applying $(U_A \otimes I_{\bar A})$ followed by $(I_A \otimes U_{\bar A})$:
    \begin{align}
        (I_A \otimes U_{\bar A})(U_A \otimes I_{\bar A}) V \ket{\alpha,i,j} = \ket{\psi_{\alpha,i}}_{A^\alpha_1} \otimes (I_{A^\alpha_2} \otimes U_{\bar A}) \ket{\chi_{\alpha,j}}_{A^\alpha_2 \bar A},
    \end{align}
    we see that $U_{\bar A}$ actually just acts on the state $\ket{\chi_{\alpha,j}}_{A^\alpha_2 \bar A}$. Thus, we see in order for both decompositions to be true simultaneously, there must exist states $\ket{\bar \psi_{\alpha,j}}$ and $\ket{\chi_\alpha}$ such that $(I_{A^\alpha_2} \otimes U_{\bar A})\ket{\chi_{\alpha,j}}_{A^\alpha_2 \bar A} = \ket{\chi_{\alpha}}_{A^\alpha_2 \bar A^\alpha_2} \otimes \ket{\bar\psi_{\alpha,j}}_{\bar A^\alpha_1}$, implying:
    \begin{align}
        (U_A \otimes U_{\bar A})V \ket{\alpha,i,j} = \ket{\psi_{\alpha,i}}_{A^\alpha_1} \otimes \ket{\chi_{\alpha}}_{A^\alpha_2 \bar A^\alpha_2} \otimes \ket{\bar \psi_{\alpha,j}}_{\bar A^\alpha_1}. \label{eqn:twosideddecomp}
    \end{align}

    The above factorization will spell out the RT formula when a logical operator is considered in this basis. Say $\rho$ is a state on $\H_L$. To show that $(V,A,\M)$ have an RT formula, we will proceed to compute $S(\M,\rho)$ as well as $S(\text{Tr}_{\bar A}( V\rho V^\dagger ))$ and take the difference. We will observe that the difference will have the form $\text{Tr}(\rho L)$ for some $L$.


   To derive $S(\text{Tr}_{\bar A}( V\rho V^\dagger ))$ recall the discussion in Section~\ref{sec:algebraic_states} and observe that one might as well consider $S(\text{Tr}_{\bar A}( V\rho_\M V^\dagger ))$ instead: say $O_A \in \mathcal{L}(\H_A)$, and write:
    \begin{align}
        \text{Tr}( O_A \cdot \text{Tr}_{\bar A}( V\rho V^\dagger )  ) =  \text{Tr}( (O_A \otimes I_{\bar A}) \cdot V\rho V^\dagger ) = \text{Tr}( V^\dagger (O_A \otimes I_{\bar A}) V \cdot \rho ).
    \end{align}
    But $V^\dagger (O_A \otimes I_{\bar A}) V$ is in $\M$. Since for any $O \in \M$ we have $\text{Tr}(O \rho) = \text{Tr}(O \rho_\M)$ (see Theorem~\ref{thm:expectations_vN}) we can just replace $\rho$ with $\rho_\M$ in the above. The states $\text{Tr}_{\bar A}( V\rho V^\dagger )$ and  $\text{Tr}_{\bar A}( V\rho_\M V^\dagger )$ give the same expectations for all observables, so they must be the same state and have the same entropy. Furthermore, since acting with a unitary on $\H_{A}$ and $\H_{\bar A}$ separately does not change the entropy, we see:
    \begin{align}
        S(\text{Tr}_{\bar A}( V\rho V^\dagger )) = S(\text{Tr}_{\bar A}( (U_A \otimes U_{\bar A})V \rho_M V^\dagger (U_A \otimes U_{\bar A})^\dagger )). 
    \end{align}

    Next, we define isometries $\tilde V_\alpha: (\H_{L_\alpha} \otimes \H_{\bar L_\alpha}) \to (\H_{A^\alpha_1} \otimes \H_{\bar A^\alpha_1})$ using the states $\ket{\psi_{\alpha,i}}_{A^\alpha_1}$ and $\ket{\bar\psi_{\alpha,j}}_{\bar A^\alpha_1}$ from (\ref{eqn:twosideddecomp}):
\begin{align}
        \tilde V_\alpha \ket{\alpha,i,j} :=  \ket{\psi_{\alpha,i}}_{A^\alpha_1} \otimes \ket{\psi_{\alpha,j}}_{\bar A^\alpha_1}.
    \end{align}
    We know that $\tilde V_\alpha$ is indeed an isometry because the states $\ket{\psi_{\alpha,i}}_{A^\alpha_1}$ and $\ket{\psi_{\alpha,j}}_{\bar A^\alpha_1}$ are actually bases for $\H_{A^\alpha_1}$ and $\H_{\bar A^\alpha_1}$ respectively. This follows from (\ref{eqn:twosideddecomp}) and the fact that the $\{\ket{\alpha,i,j}\}$ for fixed $\alpha$ are a basis for $\H_{L_\alpha} \otimes \H_{\bar L_\alpha}$. 

    The purpose of $\tilde V_\alpha$ is that it lets us simplify (\ref{eqn:twosideddecomp}) to:
    \begin{align}
        (U_A \otimes U_{\bar A})V \ket{\alpha,i,j} = \tilde V_\alpha \ket{\alpha,i,j} \otimes \ket{\chi_\alpha}.
    \end{align}

    This lets us bring the cumbersome expression $(U_A \otimes U_{\bar A})V \rho_\M V^\dagger (U_A \otimes U_{\bar A})^\dagger $ into a much neater form:
\begin{align}
    & \hspace{4mm} (U_A \otimes U_{\bar A})V \rho_\M V^\dagger (U_A \otimes U_{\bar A})^\dagger \\
    &=   \sum_\alpha p_\alpha \cdot (U_A \otimes U_{\bar A})V \rho_\alpha V^\dagger (U_A \otimes U_{\bar A})^\dagger \\
    &=  \sum_\alpha p_\alpha \cdot \frac{1}{p_\alpha} \sum_{i,j}\sum_{i',j'} \rho[\alpha]_{i,j,i',j'}  (U_A \otimes U_{\bar A})V\ket{\alpha,i,j}\bra{\alpha,i',j'} V^\dagger (U_A \otimes U_{\bar A})^\dagger  \\
    &=  \sum_\alpha p_\alpha \cdot \frac{1}{p_\alpha} \sum_{i,j}\sum_{i',j'} \rho[\alpha]_{i,j,i',j'}  \tilde V_\alpha \ket{\alpha,i,j}\bra{\alpha,i',j'} \tilde V_\alpha^\dagger  \otimes \ket{\chi_\alpha}\bra{\chi_\alpha} \\
    &=  \sum_\alpha p_\alpha \cdot \tilde V_\alpha \rho_\alpha \tilde V_\alpha^\dagger  \otimes \ket{\chi_\alpha}\bra{\chi_\alpha}. 
\end{align}

Since each of the states $\tilde V_\alpha \rho_\alpha \tilde V_\alpha^\dagger  \otimes \ket{\chi_\alpha}\bra{\chi_\alpha}$ are normalized and disjoint (act on different blocks), the entropy takes the form:

\begin{align}
    S(\text{Tr}_{\bar A}( V\rho V^\dagger ))  &= \sum_\alpha p_\alpha \log(p^{-1}_\alpha) + \sum_{\alpha} p_\alpha S( \text{Tr}_{\bar A}( \tilde V_\alpha \rho_\alpha \tilde V_\alpha^\dagger  \otimes \ket{\chi_\alpha}\bra{\chi_\alpha}  ) )\\
    &= \sum_\alpha p_\alpha \log(p^{-1}_\alpha) + \sum_{\alpha} p_\alpha S(\text{Tr}_{\bar A}( \tilde V_\alpha \rho_\alpha \tilde V_\alpha^\dagger )) + \sum_\alpha p_\alpha  S( \text{Tr}_{\bar A}(\ket{\chi_\alpha}\bra{\chi_\alpha}  )). \label{eqn:entanglemententropy}
\end{align}

Finally, we observe that since $\ket{\psi_{\alpha,j}}$ is independent of $i$, we have that:
\begin{align}
    S(\text{Tr}_{\bar A}( \tilde V_\alpha \rho_\alpha \tilde V_\alpha^\dagger ) = S(\text{Tr}_{\bar A^\alpha_1}( \tilde V_\alpha \rho_\alpha \tilde V_\alpha^\dagger )) = S( \text{Tr}_{\bar L_\alpha}(\rho_\alpha)).
\end{align}

We observe that the first two terms of (\ref{eqn:entanglemententropy}) are the exact same as the two terms of (\ref{eq:algebraic_entropy}), so their difference is just:
\begin{align}
    S(\text{Tr}_{\bar A}( V\rho V^\dagger ))  - S(\M,\rho) = \sum_\alpha p_\alpha  S( \text{Tr}_{\bar A}(\ket{\chi_\alpha}\bra{\chi_\alpha}  )). 
\end{align}

The right-hand side is linear in the $p_\alpha$, so it is linear in $\rho$, so there exists an area operator $L$ such that the right-hand side is $\text{Tr}(\rho L)$. We construct it explicitly below.

\begin{align}
    I_\alpha &:=  \sum_{i,j} \ket{\alpha,i,j}\bra{\alpha,i,j}\\
    L &:= \sum_\alpha   S( \text{Tr}_{\bar A}(\ket{\chi_\alpha}\bra{\chi_\alpha} ))  \cdot I_\alpha. 
\end{align}
Observe that $L \in \M$, so therefore $\text{Tr}(\rho L) = \text{Tr}(\rho_\M L)$. Then we write:
\begin{align}
    \text{Tr}(\rho_\M L) &= \text{Tr}\left( \sum_\alpha p_\alpha \rho_\alpha \cdot \sum_\alpha   S( \text{Tr}_{\bar A}(\ket{\chi_\alpha}\bra{\chi_\alpha} ))  \cdot I_\alpha  \right) \\
    &=   \sum_\alpha p_\alpha  S( \text{Tr}_{\bar A}(\ket{\chi_\alpha}\bra{\chi_\alpha} ))  \cdot \text{Tr}(\rho_\alpha I_\alpha) =  S(\text{Tr}_{\bar A}( V\rho V^\dagger ))  - S(\M,\rho). 
\end{align}
We have derived that $(V,A,\M)$ satisfy an RT formula with operator $L$ and furthermore that $L \in \M$. The derivation for $(V,\bar A,\M')$ is exactly the same just with $i$ and $j$ swapped, and since $L \in \M'$ we have that $L$ is in the center $Z_\M$.
\end{proof}

According to \cite{harlow2018tasi} the reverse direction also holds: if $(V, A,  \M )$ and $(V,\bar A, \M')$ both satisfy an RT formula with the same $L$, then $(V,A,\M)$ must have complementary recovery. So actually, complementary recovery is equivalent to the existence of a `two-sided RT formula' for both $(V, A,  \M )$ and $(V,\bar A, \M')$

This suggests the possibility that complementary recovery is actually stronger than the existence of a one-sided RT formula. Is it possible for $(V,A,\M)$ to exhibit an RT formula, but not $(V,\bar A,\M')$?

\subsection{A recipe for analysing codes}
\label{sec:recipe}

The derivation in this section not only defines the holographic properties of an error-correcting code, but also gives a recipe for computing the area operator of the RT formula:
\begin{enumerate}
    \item Follow Theorem~\ref{thm:whatalgebra} and compute $\M := V^\dagger (\mathcal{L}(H_A)\otimes I_{\bar A})V$, and verify that it is indeed a von Neumann algebra. If so, we have complementary recovery.
    \begin{description}
    \item \emph{Shortcut}: If $\M$ has a trivial center ($Z_\M = \langle I\rangle_\text{vN}$), then since $L \in Z_\M$ we already know that the code must have a trivial RT formula.  
    \end{description}
    \item Compute the Wedderburn decomposition on $\H_L$ that follows from $\M$. Follow Lemma~\ref{lemma:factorization} and define a basis $\ket{\alpha,i,j}$ that `lines up with $\M$'.
    \item Apply Lemma~\ref{lemma:factorization} twice to obtain unitaries $U_A$ and $U_{\bar A}$ such that:  $(U_A \otimes U_{\bar A})V \ket{\alpha,i,j} = \ket{\psi_{\alpha,i}} \otimes \ket{\chi_{\alpha}}\otimes \ket{\bar \psi_{\alpha,j}}$.
    \item Obtain the states $\ket{\chi_\alpha}$ and compute their entanglement entropies. These are the eigenvalues of the area operator.
\end{enumerate}

This is already a very complicated series of steps. While computing $\M$ is not so difficult for small codes, the later steps where we explicitly construct the $\ket{\chi_\alpha}$ states can be cumbersome. For this reason we recall that Theorem~\ref{thm:complementaritytoRT}  showed that $L \in Z_\M$. So a trivial center implies a trivial RT Formula. The intuition is that an interesting holographic code features a variety of superselection sectors, each representing a different geometry with a different area. The center $Z_\M$ is the set of operators acting proportionally to the identity on each sector. Thus, if the center is trivial, there is only one superselection sector, so there can only be one area.  This provides a convenient shortcut for analyzing the holographic properties of codes.

In the next section we will practice this recipe on various examples.

\section{Atomic examples}
\label{sec:examples}

In the previous section we discussed holographic properties of an isometry $V: \H_L \to \H$, a subregion $A$ and a von Neumann algebra $\M \subseteq \mathcal{L}(\H_L)$. Together $(V,A,\M)$ can exhibit `complementary recovery' if $\M$ can be recovered from $A$ and its commutant $\M'$  can be recovered from $\bar A$. Furthermore, $(V,A,\M)$ are said to exhibit an `RT formula' if the following equation holds for all states $\rho$ on $\H_L$:
\begin{align}
    S(\text{Tr}_{\bar A}(V \rho V^\dagger)) = S(\M,\rho) + \text{Tr}(\rho L).
\end{align}

We established two results: First, we showed that the isometry $V$ and the subregion $A$ together uniquely determine an $\M$ so that $(V,A,\M)$ have complementary recovery, and gave a simple method for calculating $\M$ if it exists. Second, we showed that complementary recovery implies that an RT Formula holds for both $(V,A,\M)$ and $(V,\bar A,\M')$.

In this section we give some examples of quantum error-correcting codes that exhibit an RT formula. These examples aim to be non-trivial while using as few qubits as possible, motivating the name `atomic'.  We begin with simple examples where the equation above holds in a trivial way, but then build our way up to an example that features an RT formula where every single term in the equation is nonvanishing. These toy models are a useful stepping stone toward an intuitive understanding of holography and its connection to quantum error correction. In particular, the statement of Lemma~\ref{lemma:factorization} and proof of Theorem~\ref{thm:complementaritytoRT} made heavy use of abstract decompositions of the Hilbert spaces as well as various intermediate states. These arguments are significantly easier to understand when keeping the examples in mind. 

In Theorem~\ref{thm:whatalgebra} we showed that $V,A$ together determine the algebra $\M$. In our examples however we only specify the encoding isometry $V$. This is because these examples actually exhibit RT formulae for all `contiguous subregions' $A$. That is, the physical Hilbert space is to be thought of as a ring of qubits, and $A$ can only contain adjacent sets of qubits. Moreover, the isometries $V$ are sufficiently symmetrical that the RT formulae for all these different subregions $A$ are identical provided they are large enough. Combined with the fact that $(V,A,\M)$ and $(V,\bar A, \M')$ have the same area operator, the analysis is thus greatly simplified.

Recall from the proof of Theorem~\ref{thm:complementaritytoRT} that for any state $\rho$ we can derive $\rho_\alpha$ and $p_\alpha$ so that the algebraic entropy can be written as:
\begin{align}
    S(\M,\rho) = \sum_{\alpha} p_\alpha \log(p_\alpha^{-1}) + \sum_\alpha p_\alpha S(\text{Tr}_{\bar L_\alpha}( \rho_\alpha)),
\end{align}
which intuitively splits the entropy into a `classical term' and a `quantum term'. The classical term is indeed just the classical entropy corresponding to the probabilities $p_\alpha$, while the quantum term is a probabilistic mixture of various von Neumann entropies. Substituting this expansion into the RT Formula, we obtain an equation with four terms. We name the first three $S_A$ after the subregion $A$, $S_c$ for `classical' and $S_q$ for `quantum':
\begin{align}
    \underbrace{S(\text{Tr}_{\bar A}(V \rho V^\dagger))}_{S_A} = \underbrace{\sum_{\alpha} p_\alpha \log(p_\alpha^{-1})}_{S_c} + \underbrace{\sum_\alpha p_\alpha S(\text{Tr}_{\bar L_\alpha}( \rho_\alpha))}_{S_q} + \text{Tr}(\rho L). \label{eqn:expandedrt}
\end{align}

The structure of this section is as follows: we begin with three examples where only one of the terms $S_c, S_q$ and $\text{Tr}(\rho L)$ is nonzero. Then we give three examples where exactly two terms are nonvanishing. Then, finally, we give one example where all three terms appear. The definitions of all the isometries $V$ are summarized in Figure~\ref{fig:circuits}.

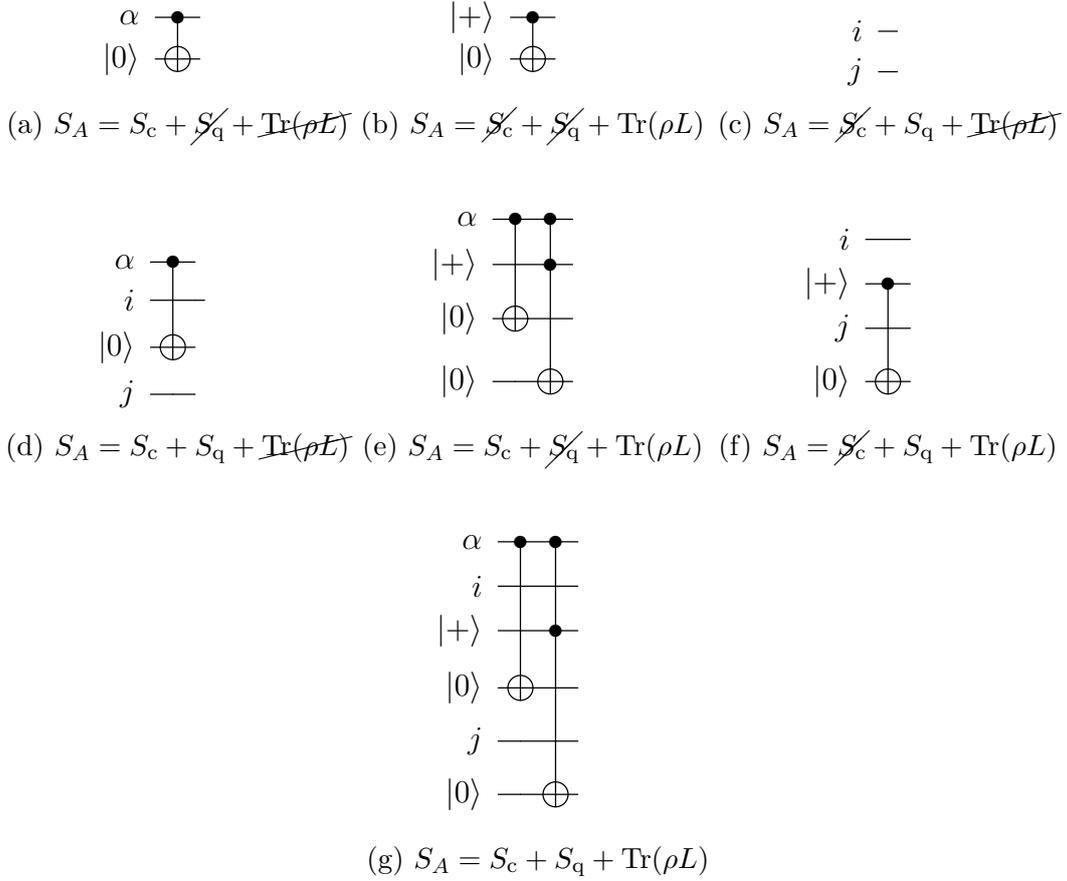
\begin{figure}[t]
\centering

\begin{subfigure}[b]{0.3\textwidth}
$$  \Qcircuit @C=0.3em @R=.8em {
    \lstick{\alpha}  & \ctrl{1} & \qw\\
    \lstick{\ket{0}}& \targ & \qw
}$$
    \caption{$S_A = S_\text{c} + \cancel{S_\text{q}} + \cancel{\text{Tr}(\rho L)} $}
\end{subfigure}
\begin{subfigure}[b]{0.3\textwidth}
$$  \Qcircuit @C=0.3em @R=.8em {
    \lstick{\ket{+}} & \ctrl{1} & \qw\\
    \lstick{\ket{0}}& \targ & \qw
}$$
    \caption{$S_A = \cancel{S_\text{c}} + \cancel{S_\text{q}} + \text{Tr}(\rho L) $}
\end{subfigure}
\begin{subfigure}[b]{0.3\textwidth}
$$  \Qcircuit @C=0.6em @R=1.3em {
    \lstick{i} & \qw\\
   \lstick{j}  & \qw
}$$
    \caption{$S_A = \cancel{S_\text{c}} + S_\text{q} + \cancel{\text{Tr}(\rho L)} $}
\end{subfigure}
\vspace{5mm}

\begin{subfigure}[b]{0.3\textwidth}
$$  \Qcircuit @C=0.3em @R=1.1em {
    \lstick{\alpha} & \ctrl{2} & \qw \\
    \lstick{i} & \qw & \qw & \qw \\
    \lstick{\ket{0}} & \targ & \qw \\
   \lstick{j}  & \qw & \qw 
}$$
    \caption{$S_A = S_\text{c} + S_\text{q} + \cancel{\text{Tr}(\rho L)} $}
\end{subfigure}
\begin{subfigure}[b]{0.3\textwidth}
$$  \Qcircuit @C=0.3em @R=1.2em {
    \lstick{\alpha} & \ctrl{2} & \ctrl{1} & \qw \\
    \lstick{\ket{+}} & \qw &  \ctrl{2} & \qw \\
    \lstick{\ket{0}} & \targ & \qw  & \qw \\
    \lstick{\ket{0}}  & \qw & \targ & \qw
}$$
    \caption{$S_A = S_\text{c} + \cancel{S_\text{q}} + \text{Tr}(\rho L) $}
\end{subfigure}
\begin{subfigure}[b]{0.3\textwidth}
$$  \Qcircuit @C=0.3em @R=1.3em {
    \lstick{i}  & \qw & \qw \\
    \lstick{\ket{+}} &  \ctrl{2} & \qw \\
    \lstick{j} & \qw  & \qw \\
    \lstick{\ket{0}} & \targ & \qw
}$$
    \caption{$S_A = \cancel{S_\text{c}} + S_\text{q} + \text{Tr}(\rho L) $}
\end{subfigure}
\vspace{5mm}

\begin{subfigure}[b]{0.3\textwidth}
$$  \Qcircuit @C=0.3em @R=1.3em {
    \lstick{\alpha}  & \ctrl{3} & \ctrl{2} & \qw \\
    \lstick{i}  & \qw & \qw & \qw \\
    \lstick{\ket{+}} & \qw & \ctrl{3} & \qw \\
    \lstick{\ket{0}} & \targ & \qw & \qw\\
    \lstick{j} & \qw  & \qw  & \qw \\
    \lstick{\ket{0}} & \qw &  \targ & \qw
}$$
    \caption{$S_A = S_\text{c} + S_\text{q} + \text{Tr}(\rho L) $}
\end{subfigure}

\caption{\label{fig:circuits}Examples of encoding isometries $V$ considered in this section. All of these exhibit an RT formula as in (\ref{eqn:expandedrt}), but various terms vanish as shown. The logical Hilbert space $\H_L$ always factors into $\H_\alpha \otimes \H_i \otimes \H_{j}$, with the input qubits marked as such.}

\end{figure}

 \subsection{Codes with one term}
 
 We specify all the isometries in terms of quantum circuits, which makes many of the non-trivial Hilbert space decompositions in Lemma~\ref{lemma:factorization} and Theorem~\ref{thm:complementaritytoRT} much simpler to understand. In particular, recall from Lemma~\ref{lemma:factorization} that $\M$ induces a decomposition on $\H_L$ of the form:
\begin{align}
    \H_L = \bigoplus_\alpha \left( \H_{L_\alpha} \otimes \H_{\bar L_\alpha}\right).
\end{align}
This very general form of a decomposition accounts for the fact that the dimensions of $\H_{L_\alpha}$ and $\H_{\bar L_\alpha}$ may vary depending on $\alpha$. This will not be the case for these examples, so we can simply remove the $\alpha$ dependence, relabeling $\H_{L_\alpha} \to \H_i$ and $\H_{\bar L_\alpha} \to \H_j$, and write:
\begin{align}
    \H_L = \H_\alpha \otimes \H_i \otimes \H_{j}.
\end{align}
Each of the degrees of freedom $\alpha,i$ and $j$ is then simply encoded by the corresponding qubit, which is labeled as such on the left side of the circuit. Here, which block of the decomposition we are in is associated with its own Hilbert space $\H_\alpha$.

\begin{atomicexample} \label{ex:c} We begin with an example where the RT Formula is simply $S_A = S_c$:
    \begin{align}
        V_{(a)} :=   \hspace{8mm}\begin{array}{c}\Qcircuit @C=0.3em @R=.8em {
    \lstick{\alpha}  & \ctrl{1} & \qw\\
        \lstick{\ket{0}}& \targ & \qw}
        \end{array}
    \end{align}
    Without loss of generality we pick $\H_A$ to be the first qubit and $\H_{\bar A}$ to be the second qubit. Intuitively, when $\H_{\bar A}$ is traced out, then the qubit $\H_A$ acts like it has been measured in the computational basis. The probabilities of the two outcomes $p_0$ and $p_1$ are a classical probability distribution.

    Following Theorem~\ref{thm:whatalgebra}, we compute $V^\dagger (\mathcal{L}(\H_A) \otimes I_{\bar A}) V$ to obtain $\M$.  A general element in $O \in \mathcal{L}(\H_A) \otimes I_{\bar A}$ can be expanded into Pauli matrices:
    \begin{align}
        O &= \alpha (I\otimes I) + \beta (X \otimes I) + \gamma (Y \otimes I) + \delta( Z \otimes I)\\
        V^\dagger OV &= \alpha I + \delta Z.
    \end{align}

    So $\M$ is indeed a von Neumann algebra: the set of diagonal operators on $\H_\alpha$. This means that observables in $\M$ cannot really distinguish superpositions over different $\alpha$ from classical probability distributions over $\alpha$, since $\rho_\M$ is also diagonal. So the algebraic entropy $S(\M,\rho)$ is also entirely classical. Notice however that $\M$ is its own center, and is not trivial! So we see that a von Neumann algebra with a non-trivial center can still have a trivial area operator $L = 0$.
\end{atomicexample}

\begin{atomicexample} \label{ex:l}Next, we consider an isometry where $S_A = \text{Tr}(\rho L)$. In this case the logical Hilbert space $\H_L$ is one dimensional: there are no logical qubits. We can still define a density matrix, though: the 1 x 1 matrix $\rho = 1$.
    \begin{align}
        V_{(b)} :=   \hspace{8mm}\begin{array}{c}\Qcircuit @C=0.3em @R=.8em {
    \lstick{\ket{+}} & \ctrl{1} & \qw\\
    \lstick{\ket{0}}& \targ & \qw
}\end{array}
    \end{align}

    $V_{(b)}$ simply prepares a Bell state, so $S_A$ is simply the constant 1. Furthermore, $\M := V^\dagger (\mathcal{L}(H_A) \otimes I_{\bar A})  V $ is just the set of scalars, so $S(\M,\rho)$ vanishes. We can thus achieve  $S_A  = 1= \text{Tr}(\rho L)$ by making the area operator the 1x1 matrix $L = 1$. This is consistent with the fact that $\M$, being the set of scalars, has a trivial center.
\end{atomicexample}

\begin{atomicexample} \label{ex:q} Third, we consider an isometry with only a quantum part: $S_A = S_q$. In this case $\H_L $ and $\H$ are both two qubits, and $V$ is the identity.
    \begin{align}
        V_{(c)} :=   \hspace{3mm}\begin{array}{c}\Qcircuit @C=0.6em @R=1.3em {
    \lstick{i} & \qw\\
   \lstick{j}  & \qw
}\end{array}
    \end{align}
    We see that $\H_i = \H_A$ and $\H_j = \H_{\bar A}$, so $S_A = S(\text{Tr}_j(\rho))$. We also have that $\M = \mathcal{L}(\H_i) \otimes I_j$, which is a factor, so the associated Hilbert space decompositions has only one big block with $\alpha = 0$ and no other values of $\alpha$. This makes the distribution over blocks trivial with $p_0 = 1$, so the classical part of $S(\M,\rho)$ vanishes and only the quantum component remains. $\M$ is a factor, so it has a trivial center, consistent with $L = 0$.
\end{atomicexample}

Indeed the above examples are extremely trivial, since they each only feature one term in the RT formula. However, they are the fundamental building blocks for codes with more complicated RT formulae.  

 \subsection{Codes with two terms}

Now we move on to RT formulae with two non-trivial terms. These allow us to make some of the steps in the proof of Theorem~\ref{thm:complementaritytoRT} more explicit. In particular, the proof involved further decomposition of $\H_{A}$ and $\H_{\bar A}$ into:

\begin{align}
    \H_{A} = \bigoplus_\alpha \left( \H_{A^\alpha_1} \otimes \H_{A^\alpha_2}   \right) \oplus \H_{A_3} \hspace{10mm} \H_{\bar A} = \bigoplus_\alpha \left( \H_{\bar A^\alpha_1} \otimes \H_{\bar A^\alpha_2}   \right) \oplus \H_{\bar A_3}. 
\end{align}

As with $ \H_L = \bigoplus_\alpha \left( \H_{L_\alpha} \otimes H_{\bar{L}_\alpha}\right)$, the $\alpha$ dependence allows the different blocks enumerated by $\alpha$ to have varying dimension. This will not be the case for our examples. Furthermore, the extra $\H_{A_3}$ allows this factorization to only factorize the image of $V$ in $\H_A$. In our case it is actually easier to just factor all of $\H_A$ and $\H_{\bar A}$ directly.
\begin{align}
    \H_{A} =  \H_{A_\alpha} \otimes \H_{A_1} \otimes \H_{A_2}   \hspace{10mm}  \H_{\bar A} =  \H_{\bar A_\alpha} \otimes \H_{\bar A_1} \otimes \H_{\bar A_2}. 
\end{align}
As with $\H_L$, which block $\alpha$ of the decomposition we are in actually factors out onto its own qubit $\H_{A_\alpha}$ or $\H_{\bar A_\alpha}$. The fact that $\alpha$ is visible from both sides of the bipartition is what lends it its classical behavior. In our circuits we now label the right side with the associated decomposition of $\H_{A}$ and $\H_{\bar A}$ as well.

In Theorem~\ref{thm:complementaritytoRT}, the purpose of the decomposition of $\H_A$ and $\H_{\bar A}$ was to show that there exist unitaries $U_A$ and $U_{\bar A}$ that bring the states $V\ket{\alpha,i,j}$ into a particular form, specifically that of equation (\ref{eqn:twosideddecomp}):
\begin{align}
    (U_A \otimes U_{\bar A})V \ket{\alpha,i,j} = \ket{\psi_{\alpha,i}}_{A_\alpha A_1} \otimes \ket{\chi_{\alpha}}_{A_2 \bar A_2} \otimes \ket{\bar\psi_{\alpha,j}}_{\bar A_\alpha\bar A_1}. \label{eqn:twosideddecomp_again}
\end{align}

Then, the entanglement entropies of the $\ket{\chi_\alpha}$ states across $A_2\bar A_2$ yield the eigenvalues of the area operator. In our examples $U_A$ and $U_{\bar A}$ will just be the identity.

\begin{atomicexample} \label{ex:cq} The following code has the RT formula $S_A = S_c + S_q$. This is the first example where all three components of $\H_L = \H_\alpha \otimes \H_i \otimes \H_j$ are two-dimensional. Below we have selected $\H_A$ as the first two qubits and $\H_{\bar A}$ as the last two qubits. However, other choices of $A$ will have the same formula provided $A$ is a pair of adjacent qubits.
    \begin{align} 
        V_{(d)} :=   \hspace{8mm}\begin{array}{c}\Qcircuit @C=0.3em @R=1.1em {
                \lstick{\alpha} & \ctrl{2} & \qw & \rstick{A_\alpha} \\
            \lstick{i} & \qw & \qw & \qw & \rstick{A_1} \\
            \lstick{\ket{0}} & \targ & \qw & \rstick{\bar A_\alpha} \\
            \lstick{j}  & \qw & \qw & \rstick{\bar A_1}
}\end{array}
    \end{align}
    We begin by computing $\M$: we see how, just as in Example~\ref{ex:q}, $\H_A$ has full access to $\H_i$, and for the same calculation as in Example~\ref{ex:c}, $\H_A$ has access to diagonal operators on $\H_\alpha$. On the other hand, it must act like the identity on $\H_j$. $Z_\M$ acts like the identity on $\H_i,\H_j$, but can act non-trivially on $\H_\alpha$. So we cannot rule out a trivial RT formula yet. 

    At this point we see that a basis $\{\ket{\alpha,i,j}\}$ for $\H_L$ that `lines up with $\M$' as in Lemma~\ref{lemma:factorization} is actually just the computational basis on $\H_L$. We can just write $\ket{\alpha,i,j} = \ket{\alpha}\ket{i}\ket{j}$. Considering such a state we see that:
    \begin{align}
        V_{(d)} \ket{\alpha,i,j}  =  \ket{\alpha}_{A_\alpha}\ket{i}_{A_1}\ket{\alpha}_{\bar A_\alpha} \ket{j}_{\bar A_1}.
    \end{align}
    Since this state already splits so cleanly into states $\ket{\psi_{\alpha,i}}_{A_\alpha A_1} = \ket{\alpha}_{A_\alpha}\ket{i}_{A_1}$ and $\ket{\bar \psi_{\alpha,j}}_{\bar A_\alpha \bar A_1} = \ket{\alpha}_{\bar A_\alpha}\ket{j}_{\bar A_1}$, we actually can just select $U_A$ and $U_{\bar A}$ to be the identity.

    The only thing missing from equation (\ref{eqn:twosideddecomp_again}) is the $\ket{\chi_\alpha}$ on $A_2,\bar A_2$. However, we have both other contributions. The $\alpha$ degree of freedom is visible from $\bar A$, so therefore acts like it has been measured from $A$'s perspective. The $\H_i \otimes \H_j$ register might be entangled with $\H_\alpha$, so after the measurement it will collapse to one of the $\rho_\alpha$ states from the decomposition $\rho_\M = \sum_\alpha p_\alpha \rho_\alpha$. The quantum term of the entropy is then the associated probabilistic mixture of the von Neumann entropy of $\rho_\alpha$ reduced to $\H_i$. Writing out the full formula:
\begin{align}
    \underbrace{S(\text{Tr}_{\bar A}(V \rho V^\dagger))}_{S_A} = \underbrace{\sum_{\alpha} p_\alpha \log(p_\alpha^{-1})}_{S_c} + \underbrace{\sum_\alpha p_\alpha S(\text{Tr}_{j}( \rho_\alpha))}_{S_q}.
\end{align}
    Just as with Example~\ref{ex:c}, $Z_\M$ had a non-trivial center, but we still have $L = 0$.
\end{atomicexample}

\begin{atomicexample} \label{ex:cl} Next we consider a code with a classical term and an area term, but no quantum term: $S_A = S_c + \text{Tr}(\rho L)$. This is the first code where $L$ is not proportional to the identity.
    \begin{align} 
        V_{(e)} :=   \hspace{8mm}\begin{array}{c}\Qcircuit @C=0.3em @R=1.2em {
                \lstick{\alpha} & \ctrl{2} & \ctrl{1} & \qw & \rstick{A_\alpha} \\
            \lstick{\ket{+}} & \qw &  \ctrl{2} & \qw & \rstick{A_2} \\
            \lstick{\ket{0}} & \targ & \qw  & \qw & \rstick{\bar A_\alpha} \\
            \lstick{\ket{0}}  & \qw & \targ & \qw & \rstick{\bar A_2}
}\end{array}
    \end{align}
    The von Neumann algebra $\M$, due to a similar calculation as in Example~\ref{ex:c}, is again just the set of diagonal operators on $\H_L = \H_\alpha$. The algebraic entropy $S(\M,\rho)$ is then again the classical entropy of the probability distribution $\{p_\alpha\}$.  Since there are multiple superselection sectors corresponding to different $\alpha$, we do not have a trivial center.

    For this example, the entropy $S_A$ can actually be computed explicitly for some logical pure state $\beta_0\ket{0} + \beta_1\ket{1}$ where $p_\alpha = |\beta_\alpha|^2$. The circuit conditionally prepares a Bell state depending on the value of $\alpha$:
\begin{align}
    V_{(e)} ( \beta_0\ket{0} + \beta_1\ket{1}) &= \beta_0 \ket{0\text{+}00} + \beta_1 \frac{ \ket{1010} +\ket{1111}}{\sqrt{2}}\\
    \text{Tr}_{\bar A}( V_{(e)}\rho V_{(e)}^\dagger ) &=  |\beta_0|^2 \ket{0}\bra{0}_{A_\alpha} \otimes \ket{+}\bra{+}_{A_2} + |\beta_1|^2 \ket{1}\bra{1}_{A_\alpha} \otimes \frac{I_{A_2}}{2}   \\
    S_A = S(\text{Tr}_{\bar A}( V_{(e)}\rho V_{(e)}^\dagger  )) &= \left[|\beta_0|^2\log(|\beta_0)|^{-2}) + |\beta_1|^2\log(|\beta_1)|^{-2})\right]\\
    &+ \left[ |\beta_0|^2 S( \ket{+}\bra{+}  ) + |\beta_1|^2 S( I/2) \right]\\
    &=  \sum_\alpha p_\alpha \log(p_\alpha^{-1}) +  \text{Tr}\left( \rho \begin{bmatrix} 0 & 0 \\ 0 & 1\end{bmatrix} \right).
\end{align}
    So we have explicitly derived an area operator $L = \ket{1}\bra{1}$.  Also worth noting is that equation (\ref{eqn:twosideddecomp_again}) is now almost fully rendered out: while $A_1$ and $\bar A_1$ are missing, we now have:
    \begin{align}
        V \ket{\alpha,i,j} = \ket{\psi_{\alpha,i}}_{A_\alpha} \otimes \ket{\chi_{\alpha}}_{A_2 \bar A_2} \otimes \ket{\bar\psi_{\alpha,j}}_{\bar A_\alpha},
\end{align}
    where $\ket{\psi_{\alpha,i}} = \ket{\bar\psi_{\alpha,j}} = \ket{\alpha}$ and $\ket{\chi_0} = \ket{+}\ket{0}$ and $\ket{\chi_1}$ is a Bell state. We see that $L = \sum_\alpha S(\text{Tr}_{\bar A}(\ket{\chi_\alpha}\bra{\chi_\alpha})) \cdot I_\alpha$ matches what we derived above.
\end{atomicexample}

\begin{atomicexample} \label{ex:ql} Now we consider a code with a quantum term and an area term, but no classical term: $S_A = S_q + \text{Tr}(\rho L)$. This code actually features an area operator proportional to the identity again: the $\alpha$ degree of freedom determines $\ket{\chi_\alpha}$, whose entanglement in turn determines the area. But since there is only one $\alpha$, we have a trivial center and there can be no superposition over areas.
    \begin{align} 
        V_{(f)} :=   \hspace{8mm}\begin{array}{c}\Qcircuit @C=0.3em @R=1.2em {
                \lstick{i}  & \qw & \qw & \rstick{A_1} \\
            \lstick{\ket{+}}  &  \ctrl{2} & \qw & \rstick{A_2} \\
            \lstick{j}  & \qw  & \qw & \rstick{\bar A_1} \\
            \lstick{\ket{0}}  & \targ & \qw & \rstick{\bar A_2}
}\end{array}
    \end{align}
    Similarly to Example~\ref{ex:q}, we have $\H_i = \H_{A_1}$ and $\H_j = \H_{\bar A_1}$, and $\M = \mathcal{L}(\H_i) \otimes I_j$. Since $\M$ is a factor, the only contribution to $S(\M,\rho)$ is the entropy of the reduced state on $\H_i$, that is, $S(\text{Tr}_j{\rho})$. There is only one value of $\alpha$.

    However, the entropy $S_A$ now features two contributions: the entropy of the state on $\H_i$ visible from $\H_{A_1}$, and the entropy of the Bell state across $\H_{A_2}\otimes \H_{\bar A_2}$. We can see this from the form of $V_{(f)}\ket{\alpha,i,j} = \ket{\psi_{\alpha,i}}_{A_1} \otimes \ket{\chi_{\alpha}}_{A_2 \bar A_2} \otimes \ket{\bar\psi_{\alpha,j}}_{\bar A_1}$ where $\ket{\psi_{\alpha,i}} = \ket{i}$, $\ket{\bar\psi_{\alpha,j}} = \ket{j}$ and $\ket{\chi_\alpha}$ is a Bell state. 

    As a result, $S_A - S_q = 1$, so we achieve $S_A = S_q + \text{Tr}(\rho L)$ by setting $L = I$. 
\end{atomicexample}

As we have seen, combining two of the primitive circuits from Examples~\ref{ex:c}~\ref{ex:l}, and~\ref{ex:q} already produces non-trivial results, including states of the form $\ket{\alpha,i,j}$ in Example~\ref{ex:cq} and area operators not proportional to the identity in Example~\ref{ex:cl}. Of particular importance in Example~\ref{ex:cl} was the conditional preparation of a Bell state based on $\alpha$. This caused the different $\ket{\chi_\alpha}$ states to exhibit varying amounts of entanglement, each becoming a different eigenvalue of $L$.

 \subsection{A complete example}

To finish the section, we give a final example that features all three terms of the RT formula, and makes both of the decompositions $\H_L = \H_\alpha \otimes \H_i \otimes \H_j$ and $\H_A = \H_{A_\alpha} \otimes \H_{A_1} \otimes \H_{A_2}$ completely non-trivial.

\begin{atomicexample} \label{ex:cql} This six-qubit code's RT formula has all three terms on the right-hand side: $S_A = S_c + S_q + \text{Tr}(\rho L)$. We consider the subregion $A$ to be the first three qubits, but the same RT formula holds for any choice of $A$ that is three adjacent qubits. Smaller or larger $A$ will exhibit a simpler RT formula, similar to those from the previous examples.
\begin{align}
V_{(g)} :=   \hspace{8mm}\begin{array}{c}\Qcircuit @C=0.3em @R=1.3em {
        \lstick{\alpha}  & \ctrl{3} & \ctrl{2} & \qw & \rstick{A_\alpha}\\
    \lstick{i}  & \qw & \qw & \qw & \rstick{A_1} \\
    \lstick{\ket{+}} & \qw & \ctrl{3} & \qw & \rstick{A_2} \\
    \lstick{\ket{0}} & \targ & \qw & \qw & \rstick{\bar A_\alpha} \\
    \lstick{j} & \qw  & \qw  & \qw & \rstick{\bar A_1} \\
    \lstick{\ket{0}} & \qw &  \targ & \qw & \rstick{\bar A_2}
}\end{array}
\end{align}

    Similarly to Example~\ref{ex:cq}, $\H_A$ has full access to $\H_i$ via $\H_{A_1}$, as well as access to the diagonal operators on $\H_{\alpha}$ via $\H_{A_\alpha}$ from the calculation in Example~\ref{ex:c}, and no access to $\H_j$. Therefore, the basis $\ket{\alpha,i,j}$ is just the computational basis on the three logical qubits with $\ket{\alpha,i,j} = \ket{\alpha}\ket{i}\ket{j}$.

    If we apply the isometry $V_{(g)}$ to such a basis state we get the full equation (\ref{eqn:twosideddecomp_again}):
\begin{align}
    V_{(g)} \ket{\alpha,i,j} &= \ket{\psi_{\alpha,i}}_{A_\alpha A_1} \otimes \ket{\chi_{\alpha}}_{A_2 \bar A_2} \otimes \ket{\bar\psi_{\alpha,j}}_{\bar A_\alpha \bar A_1}, \\[2mm]
    \ket{\psi_{\alpha,i}}_{A_\alpha A_1} &= \ket{\alpha}_{A_\alpha}\ket{i}_{A_1}, \hspace{15mm} \ket{\bar \psi_{\alpha,j}}_{\bar A_\alpha \bar A_1} = \ket{\alpha}_{\bar A_\alpha}\ket{j}_{\bar A_1}, \\ 
    \ket{\chi_0}_{A_2 \bar A_2} &= \ket{+}_{A_2}\ket{0}_{\bar A_2}, \hspace{14mm} \ket{\chi_0}_{\bar A_\alpha \bar A_1} = \frac{\ket{00}_{A_2\bar A_2}+\ket{11}_{A_2\bar A_2} }{\sqrt{2}}.
\end{align}

    As in Example~\ref{ex:cl}, we conditionally prepare a Bell state on $\H_{A_2}\otimes \H_{\bar A_2}$, so following the same calculation we see that the area operator is $L = \ket{1}\bra{1}$. But additionally this example also features the $\H_{A_1}$ and $\H_{\bar A_1}$ spaces corresponding to $\H_i$ an $\H_j$, contributing a quantum term to the RT formula as in Example~\ref{ex:cq}.
\end{atomicexample}

These circuits seem to be the smallest examples of qubit quantum error-correcting codes to exhibit interesting holographic properties. However, there are some ideas that these circuits still oversimplify.

First, the factorizations $\H_L = \H_\alpha \otimes \H_i \otimes \H_j$ and $\H_A = \H_{A_\alpha} \otimes \H_{A_1} \otimes \H_{A_2}$ are a significant simplification of the decompositions $\H_L = \bigoplus_\alpha\left( \H_{L_\alpha} \otimes \H_{\bar L_\alpha} \right)$ and $\H_A = \bigoplus \left( \H_{A^\alpha_1} \otimes \H_{A^\alpha_2} \right) \oplus \H_{A_3}$ respectively. Not only are all the dimensions of $\H_{L_\alpha},\H_{\bar L_\alpha}, \H_{A^\alpha_1}, \H_{A^\alpha_2}$ independent of $\alpha$, but consequently the $\alpha$ degree of freedom neatly factors out onto a separate qubit. This is a highly non-generic feature for von Neumann algebras: while $\alpha$ can be measured via a projective measurement, it usually does not factor to its own degree of freedom like this.

Second, none of these examples exhibit the `radial commutativity' discussed in Subsection \ref{sub:radial}. In holography, operators acting on a single point at the boundary do not have access to any bulk degrees of freedom and must therefore commute with all bulk operators. In a finite-dimensional analogy, \cite{harlow2017ryu} constructed a three-qutrit code where operators acting on any single qutrit must commute with the logical operators of the code. However, the codes presented here do not have this property. In example~\ref{ex:cql}, access to the physical qubit labeled $A_1$ already gives full access to the $\H_i$ factor of the logical Hilbert space. One method for remedying this could be to encode each of the physical qubits of example~\ref{ex:cql} into another quantum error-correcting code that protects against single qubit erasures.

\section{Discussion}
\label{sec:discussion}
Toy models for holographic quantum error correction serve as a microcosm for understanding AdS/CFT. In this work we have reformulated and extended the framework of \cite{harlow2017ryu} with a uniqueness result and several examples. These together serve to make holographic quantum error correction `more concrete' in the sense that they pave the road to more complex examples. In this discussion we briefly summarize the ways in which our construction differs from previous work, and also list some future directions.

The construction of \cite{harlow2017ryu} is, of course, central to our work and discussions of holographic quantum error correction in general. However, we made several changes to the formalism to facilitate our particular viewpoint. Here is a brief summary of these changes:

\begin{description}
    \item[Code subspace vs encoding isometry.] In holography, we can think of the bulk Hilbert space as `emanating from' the boundary space and physically place the bulk into the boundary. In this sense, we could consider the space of allowed bulk states a subspace $\H_\text{code} $ of the physical space $\H$, which could be defined via some set of constraints on the boundary qubits. This perspective might be effective for stabilizer codes. However, the codes we discuss in Section~\ref{sec:examples} are more naturally described via a quantum circuit, which is an active transformation. For this reason we explicitly think of the bulk space as a separate space $\H_L$, which is not emanating from the boundary in the same way, and is then mapped to the boundary space via an encoding isometry $V :\H_L \to \H$. Of course, we could still switch to the old picture by defining $\H_\text{code}$ to be the image of $V$.

    \item[Step by step vs general case.] A key result of \cite{harlow2017ryu} is that many seemingly disparate ideas are actually equivalent: subregion duality, the existence of an RT formula, and entropic properties of the holographic states. This is an illuminating observation about the general properties of holographic quantum error correction codes. However, in our work we are interested in the analysis of particular codes: we want to consider a particular encoding isometry $V$ and obtain its RT formula. To that end, we `unroll' the sequence of equivalences given in Theorem~5.1 of \cite{harlow2017ryu} and focus on the direction that yields a method for computing the area operator. Our derivation of this result in Theorem~\ref{thm:complementaritytoRT} goes into significantly more detail, and the resulting recipe from Section~\ref{sec:recipe} makes the analysis of codes more straightforward.

    \item[Uniqueness of the algebra.] Following \cite{harlow2017ryu}, we still consider holography to be a property that a code, a bipartition, and a von Neumann algebra can have together. But to some extent this is no longer really necessary: we can say that holography is merely a property of a code and a bipartition, because when these are fixed then the von Neumann algebra is unique if it exists. Ideally, we would like to go even further and say that it is a property of a quantum error correction code alone, asserting that every (reasonable) bipartition obeys an RT formula. 
\end{description}

In particular, making `holography' a property of a code alone leaves a couple open questions. Furthermore, there are several directions in which this framework could be expanded.

\begin{description}
    \item[A `one-sided' RT formula without complementary recovery?] In our Theorem~\ref{thm:complementaritytoRT}, we demonstrate that complementary recovery of $(V,A,\M)$ implies a `two-sided' RT formula, an RT formula for both $(V,A,\M)$ and $(V,\bar A,\M')$. Indeed, Theorem~5.1 of \cite{harlow2017ryu} shows that the existence of this `two-sided' is actually equivalent to complementary recovery. So why do we not simply remove $\M$ since it is uniquely determined by complementary recovery? We have not closed the possibility of a `one-sided' RT formula exhibited by just $(V,A,\M)$ but not by $(V,\bar A,\M')$. Is this mathematically possible? Does a code with such an RT formula possess a sensible physical interpretation?
    \item[A non-trivial RT formula for all subregions?] The qubits in Example~\ref{ex:cql} are arranged such that every contiguous subregion $A$ of three qubits has a non-trivial RT formula. But when we consider three qubits that are non-adjacent, then the RT formula becomes trivial. We would hope that larger holographic error correction codes have sensible and interesting area operators even when $A$ is not contiguous. But is it possible for \emph{every} subregion to have a non-trivial RT formula? We attempted to construct such a code without success. It is possible that this difficulty is related to the difficulty of obtaining power-law correlations between generic subregions in holographic tensor networks, observed in e.g.\ \cite{Gesteau:2020hoz,Jahn:2020ukq,Cao:2021wrb}. Since the number of possible subregions grows very quickly, this requirement places many constraints on the code. Thus, we conjecture that this is not possible. Is it possible if we restrict the pieces of the subregions to be at least a certain size?
    \item[A tensor network with superposition of geometries?]
    Seminal work by [Happy] showed that holographic tensor networks can be constructed from a tessellation of hyperbolic space with a fundamental tensor, in their case a perfect tensor. There are very many extensions of this construction, for instance \cite{cao2020approximate} consider replacing the fundamental tensor with skewed Bacon-Shor codes, and \cite{taylor2021holography} consider higher-dimensional tessellations. What happens when we replace the fundamental tensor with one of our circuits? What does operator pushing look like in this scenario? Does the network possess a non-trivial area operator?

    \item[A holographic stabilizer code?] All the atomic examples that have a non-trivial area operator possess a Toffoli gate, so they are not stabilizer codes. Furthermore, the skewed codes considered by \cite{cao2020approximate}, although they are superpositions of stabilizer codes, are themselves also not stabilizer codes. It appears that the stabilizer formalism places strong limitations on the entanglement properties of the resulting codes, making the design of a stabilizer code with a non-trivial area operator challenging. Is it even possible?

    \item[Consequences of the uniqueness of $\M$ in quantum gravity?] It is often natural to consider only a subalgebra of the operators in the entanglement wedge of a particular boundary $A$. For example, we might only be interested in local operators. But an implication of Theorem~\ref{thm:whatalgebra} is that such a von Neumann algebra cannot exhibit complementary recovery. This is clear from the perspective of error correction, but can it be proved from the AdS perspective as well? It is also natural from the AdS perspective to consider sets of operators which are not subalgebras (such as low-point correlators of local bulk operators)---can anything be said about such cases?

    \item[Extensions of holographic quantum error correction?] The toy models considered in this work, just like the constructions of \cite{harlow2017ryu} and \cite{cao2020approximate}, are restricted to a single time slice. Can they be extended to exhibit dynamics (similarly to what has been proposed for tensor networks~\cite{kohler2019toy})? What about dynamics with decoherence and black hole formation/evaporation? Since the purpose of the toy models is to illuminate and provide more mathematically tractable examples of AdS/CFT, extending them towards the full capabilities of AdS/CFT is a very natural direction. For a fixed geometry, one expects bulk time evolution (for example, on a Rindler wedge) to be implemented approximately as a local operator in the code subspace, the (modular) Hamiltonian--but evolving with the full Hamiltonian of the boundary system should give corrections to this picture. See \cite{jahn2021holographic} for a recent review.

\end{description}

\section*{Acknowledgements}
We thank Scott Aaronson, Elena Caceres, Charles Cao, William Kretschmer, Kunal Marwaha, Alex May, Frank Schindler, Haoyu Sun, and Yuxuan Zhang for detailed comments on the manuscript.
We thank Mario Martone for his comments and for his participation in an initial phase of this project. AR and JP are supported by the Simons Foundation through It from Qubit: Simons Collaboration on Quantum Fields, Gravity, and Information. PR is supported by Scott Aaronson's Vannevar Bush Faculty Fellowship.

\appendix
\section{Complementarity of private and correctable algebras}
\label{app:privacy}

Here we give a brief overview of the main result of~\cite{crann2016private} closely following the simpler, finite-dimensional, presentation given in~\cite{kribs2018quantum}. 

A quantum channel $\Phi: \mathcal{L}(\mathcal{H}_A) \rightarrow \mathcal{L}(\mathcal{H}_A)$ is a completely-positive, trace preserving map between two spaces of linear operators. 
The dual map $\Phi^{\dagger}$ of a quantum channel $\Phi$ is defined via the trace inner product $\operatorname{Tr}(\Phi(\rho) X)=\operatorname{Tr}\left(\rho \Phi^{\dagger}(X)\right)$.

Using the Stinespring dilation theorem we can express any quantum channel $\Phi$ in terms of its action on an auxiliary Hilbert space $\mathcal{H}_C$ (with $|\mathcal{H}_C|\leq |\mathcal{H}_C|^2$ ). 
In particular, given a state $\left|\psi_{C}\right\rangle \in \mathcal{H}_{C}$ and a unitary $U$
on $\mathcal{H}_{A} \otimes \mathcal{H}_{C}$ such that for all $\rho \in \mathcal{L}\left(\mathcal{H}_{A}\right)$ we have that,
\begin{equation}
\Phi(\rho)=\operatorname{Tr}_{C} \circ \, \mathcal{U}\left(\rho \otimes\left|\psi_{C}\right\rangle\left\langle\psi_{C}\right|\right)=\operatorname{Tr}_{C} \circ \mathcal{V}(\rho),
\end{equation}
where $\operatorname{Tr}_{C}$ denotes the partial trace map from $\mathcal{L}\left(\mathcal{H}_{A} \otimes \mathcal{H}_{C}\right)$ to $\mathcal{L}\left(\mathcal{H}_{A}\right)$, the map $\mathcal{U}(\cdot)=U(\cdot) U^{*}$, and $\mathcal{V}(\cdot)=V(\cdot) V^{*}$ is the map implemented by the isometry $V: \mathcal{H}_{A} \rightarrow \mathcal{H}_{A} \otimes \mathcal{H}_{C}$ defined by $V|\psi\rangle=U\left(|\psi\rangle \otimes\left|\psi_{C}\right\rangle\right)$.

The Stinespring dilation theorem allows us to define a notion of complementarity for quantum channels.
\begin{definition}[complementary map]
Given a quantum channel $\Phi$ the complementary map from $\mathcal{L}\left(\mathcal{H}_{A}\right)$ to $\mathcal{L}\left(\mathcal{H}_{C}\right)$ is
\begin{equation}
    \Phi^{C}(\rho)=\operatorname{Tr}_{A} \circ \mathcal{V}(\rho).
\end{equation}
\end{definition}

Equipped with these notions we proceed to define correctable and private algebras.
\begin{definition}[correctable algebra, Definition 2.1~\cite{kribs2018quantum}]
Let $\mathcal{H}$ be a finite-dimensional Hilbert space (the physical space). A quantum error-correcting code is defined by a projection $P$ on $\mathcal{H}$ such that $\mathcal{H}_{L} = P\mathcal{H} $. Given an error channel $\mathcal{E}: \mathcal{L}(\mathcal{H}) \rightarrow \mathcal{L}(\mathcal{H})$, we say that a von Neumann algebra $\M \subseteq \mathcal{L}(P \mathcal{H})$ is correctable for $\mathcal{E}$ with respect to $P$ if there exists a channel $\mathcal{R}: \mathcal{L}\left(\mathcal{H}_{C}\right) \rightarrow \mathcal{M}$ such that
\begin{equation}
    \Phi_{P} \circ \mathcal{E}^{\dagger} \circ \mathcal{R}^{\dagger}=i d_{\mathcal{A}},
\end{equation}
where $\Phi_{P}$ is the channel associated with the projection into the code subspace $\Phi_{P}(\cdot)=P(\cdot) P$.
\end{definition}

\begin{definition}[private algebra, Definition 2.2~\cite{kribs2018quantum}]
Let $\mathcal{H}$ be a (finite-dimensional) Hilbert space and let $P$ be a projection on $\mathcal{H}$. Given a channel $\mathcal{E}: \mathcal{L}(\mathcal{H}) \rightarrow \mathcal{L}(\mathcal{H})$, a von Neumann algebra $\mathcal{M} \subseteq \mathcal{L}(P \mathcal{H})$ is private for $\mathcal{E}$ with respect to $P$ if
\begin{equation}
    \quad \Phi_{P} \circ \mathcal{E}^{\dagger}(\mathcal{L}(\mathcal{H})) \subseteq \mathcal{M}^{\prime}=\{X \in \mathcal{L}(P \mathcal{H}) \mid[X, O]=0 \, \forall O \in \mathcal{<}\}.
\end{equation}
\end{definition}

Correctable and private algebras are related by the following theorem
\begin{theorem}[Proposition 2.4~\cite{kribs2018quantum}]
\label{thm:correctable-private}
Let $\mathcal{M}$ be a subalgebra of $\mathcal{L}(P \mathcal{H})$ for some Hilbert space $\mathcal{H}$ and projection $P$. Let $\mathcal{E}$ be a channel on $\mathcal{H}$ with complementary channel $\mathcal{E}^{C}$. Then $\mathcal{M}$ is correctable for $\mathcal{E}$ with respect to $P$ if and only if $\mathcal{M}$ is private for $\mathcal{E}^{C}$ with respect to $P$.
\end{theorem}
\section{A proof of a special case of the factorization lemma}

\label{app:structure_lemma}

We give a proof of Lemma~\ref{lemma:factorization} for the factor algebra $\mathcal{M}=\mathcal{L}(\mathcal{H}_L)$.
Our proof is similar to the ones given in~\cite[Section 3.2]{almheiri2015bulk} and~\cite[Section 3.1]{harlow2017ryu} and, as the proofs given in these works, is based on a technique developed in~\cite{schumacher1996quantum} to prove that the presence of entanglement in a code is a necessary and sufficient condition for perfect quantum error correction. 
More specifically, \cite{schumacher1996quantum} considers a setting where the system $A$ is entangled to a reference system $R$ and shows that perfect quantum error correction (i.e. the ability to recover the logical information after the erasure of $\bar A$) is possible if and only if
\begin{equation}
\label{eq:entropic_correction_condition}
    I_{R\bar A} = S_R + S_{\bar A} - S_{R\bar A} = 0,
\end{equation}
where $I_{R\bar A}$ is the mutual information of the composite system ${R\bar A}$ and $S$ denotes the von Neumann entropy. 

Let $\ket{\phi}$ be a pure state on the $R A \bar A$ system. Then \eqref{eq:entropic_correction_condition} implies that
\begin{equation}
\label{eq:reduced_factorization}
\rho_{R \bar A}[\phi]=\rho_{R}[\phi] \otimes \rho_{\bar A}[\phi],
\end{equation}
where $\rho_{R \bar A}[\phi] = \operatorname{Tr}_{R \bar A} (\ket{\phi} \bra{\phi})$, $\rho_{R}[\phi]=\operatorname{Tr}_{R } (\ket{\phi} \bra{\phi})$, and $\rho_{ \bar A}[\phi]=\operatorname{Tr}_{ \bar A} (\ket{\phi} \bra{\phi})$ denote the reduced density matrices for the state $\ket{\phi}$.


For the proof of Lemma~\ref{lemma:factorization} when $\mathcal{M}$ is a factor algebra $\mathcal{M} = \mathcal{L}(\mathcal{H}_L)$, consider a reference system $R$ maximally entangled with $A$ (this is the Choi state for $V$)
\begin{equation}
\label{eq:Choi}
   \ket{\phi}=2^{-k / 2} \sum_{i}\ket{i}_{R} (V\ket{i})_{A\bar A},
\end{equation}
where $|R| = |\mathcal{H}_L| = 2^k$. Observe that $\ket{\phi}$ is a purification of $\rho_{R \bar A}[\phi]$ on $A$. Because $\ket{\phi}$ is maximally entangled we have that $\rho_{R}[\phi] = \frac{I}{2^{k / 2}}$ is the maximally mixed state.
Therefore \eqref{eq:reduced_factorization} becomes
\begin{equation}
\label{eq:product_maxmixed}
    \rho_{R \bar A}[\phi]=\frac{I}{2^{k / 2}} \otimes \rho_{\bar A}[\phi].
\end{equation}

Say $k$ is the largest integer such that $|A| = k |R| + r$. Then there exists a factorization $\H_A = (\H_{A_1} \otimes \H_{A_2}) \oplus \H_{A_3}$ such that $|A_1| = |R|$, $|A_2| = k$ and $|A_3| = r$. Now define the following states:
\begin{equation}
    \ket{\Psi}_{R A_1} = \frac{1}{2^{k / 2}} \sum_{i}\ket{i}_{R} \ket{i}_{A_1}, \quad \ket{\chi}_{A_2 \bar A} = \sum \sqrt{p_j} \ket{j}_{A_2} \ket{j}_{\bar A}.
\end{equation}
and observe that the state
\begin{equation}
    \ket{\phi ^\prime} = \ket{\Psi}_{R A_1} \otimes \ket{\chi}_{A_2 \bar A},
\end{equation}
is a purification of $\rho_{R\bar A}[\phi]$
\begin{align}
    \operatorname{Tr}_{A_1 A_2} \left( \ket{\Psi}\bra{\Psi}_{R A_1} \otimes \ket{\chi} \bra{\chi}_{A_2 \bar A} \right) &= \operatorname{Tr}_{A_1} \left( \ket{\Psi} \bra{\Psi}_{R A_1} \right) \operatorname{Tr}_{A_2} \left( \ket{\chi} \bra{\chi}_{A_2 \bar A} \right) \\
    &= \rho_{R}[\phi] \otimes \rho_{\bar A}[\phi].
\end{align}
where $\ket{\Psi}_{R A_1}$  purifies $\rho_R[\phi]$ in $A_1$ and $\ket{\chi}_{A_2 \bar A}$ purifies $\rho_{\bar A} [\phi]$ in $A_2$. Note that such a factorisation exists because the $R$ and $\bar A$ registers are unentangled in \eqref{eq:product_maxmixed}. In a purification the dimension of the purifying system needs to be at least as big as the rank of the state to be purified and therefore we have that $ |A_1| = |R| = 2^{k}$ (because $\rho_R [\phi]$ is maximally mixed) and $\operatorname{rank}(\rho_{\bar A}[\phi]) \leq |A_2|$.

Because all purifications are equivalent up to unitaries performed on the purifying system ($A$, in our case) we know that there exists a unitary $U_A$ acting solely on the subsystem $A$ that maps $\ket{\phi^\prime}$ to $\ket{\phi}$. Therefore we have that
\begin{equation}
    (U_A \otimes I_{\bar A} )V\ket{i}_{A\bar A} = \ket{i}_{A_1} \ket{\chi}_{A_2 \bar A}.
\end{equation}

\section{The 2x2 Bacon-Shor code}
\label{app:2x2bacon-shor}

\cite{cao2020approximate} presents a construction of holographic tensor networks using the 2 x 2 Bacon-Shor code. This four-qubit stabilizer subsystem code can be shown to have simple holographic properties. \cite{cao2020approximate} find that, via a notion of `skewing' which involves taking linear combinations of several encoding isometries, they can construct quantum error-correcting codes with a non-trivial RT formula.

For ease of comparison to their work, we review some of their calculations in our language. We rederive that, while the 2x2 Bacon-Shor code is holographic, its area operator is proportional to the identity and its RT formula is trivial. This demonstrates that some notion of `skewing' is necessary to obtain non-trivial RT formulas from this code. 

Before we talk about the Bacon-Shor code, we derive a result about stabilizer codes in general.
\begin{lemma} \label{lemma:stabm} Say $G$ is an abelian subgroup of the $n$-qubit Pauli group, defining a stabilizer code. Say $V_G : \H_L \to \H$ is an encoding isometry of this code. Say $A$ is any subset of the $n$ qubits, decomposing $\H = \H_A \otimes \H_{\bar A}$. Then $\mathcal{M} := V_G^\dagger (\mathcal{L}(\H_A) \otimes I_{\bar A})V_G $ is a von Neumann algebra, and $(V_G,A,\M)$ satisfy complementary recovery.
\end{lemma}
\begin{proof} All we need to show is that $V_G^\dagger (\mathcal{L}(\H_A) \otimes I_{\bar A})V_G $ is closed under multiplication, since the other properties of von Neumann algebras are always guaranteed. Then Theorem~\ref{thm:whatalgebra} implies complementary recovery. To do so, we observe that $\mathcal{L}(\mathbb{C}^2) = \langle X,Z \rangle_\text{vN}$, which lets us write:
    \begin{align}
        \mathcal{L}(\H_A) \otimes I_{\bar A} = \langle  X_i, Z_i \text{ for } i \in A  \rangle_\text{vN}.
    \end{align}
Conjugating by $V_G^\dagger$ projects each Pauli matrix in the above set into the code space, and then decodes it. Note that a Pauli matrix is in the code space if and only if it commutes with the stabilizer $G$. We write:
    \begin{align}
      \M = V_G^\dagger (\mathcal{L}(\H_A) \otimes I_{\bar A})V_G = \left\{ V_G^\dagger P V_G \text{ for } P \in \langle  X_i, Z_i \text{ for } i \in A  \rangle_\text{vN} \text{ if } PG = GP \right\}.
    \end{align}
    Now all that remains to show is that the above set is closed under multiplication. Write the code space projector as $\Pi_G = V_G^\dagger V_G \in \mathcal{L}(\H)$, and consider two elements $V_G^\dagger P V_G$ and $V_G^\dagger Q V_G$ in the above set. Then, since $P$ and $\Pi_G$ commute:
    \begin{align}
    V_G^\dagger P V_G V_G^\dagger Q V_G = V_G^\dagger P \Pi_G Q V_G  = V_G^\dagger\Pi_G P  Q V_G = V_G^\dagger P  Q V_G.
\end{align}
 $PQ$ is in $\langle  X_i, Z_i \text{ for } i \in A  \rangle_\text{vN}$ since it is a von Neumann algebra by definition, and furthermore since $P$ and $Q$ both commute with $G$, so must $PQ$. So $ V_G^\dagger P  Q V_G$ is also in $\M$.
\end{proof}

This establishes the fact that situations like Example~\ref{ex:badcode} cannot happen with stabilizer codes, but also gives  a simple method for computing $\M$.

Now we discuss the 2x2 Bacon-Shor code. Recall from subsystem quantum error correction that a subsystem code is generated by a non-abelian group of Pauli matrices $\mathcal{G}$. In this case:
\begin{align}
    \mathcal{G} := \langle X_1X_2, X_3X_4, Z_1Z_3, Z_2Z_4   \rangle.
\end{align}

Non-abelian Pauli groups do not have a simultaneous $+1$ eigenspace. However, we can construct several abelian Pauli groups from $\mathcal{G}$. One of these is its center:
\begin{align}
    Z_\mathcal{G} = \langle X_1X_2X_3X_4, Z_1Z_2Z_3Z_4   \rangle.
\end{align}

This is by definition abelian, and has an encoding isometry:
\begin{align}
    V_{Z_\mathcal{G}} := \hspace{8mm}   \begin{array}{c}\Qcircuit @C=.8em @R=.8em {
             & \ctrl{2} & \qw & \gate{H} & \ctrl{1} & \qw \\
             & \qw & \ctrl{2} & \qw      & \targ    & \qw \\
\lstick{\ket{0}} & \targ    & \qw & \gate{H} & \ctrl{1} & \qw \\
\lstick{\ket{0}} & \qw & \targ    & \qw      & \targ    & \qw
    }\end{array}
\end{align}

Visibly, $Z_\mathcal{G}$ defines a quantum error-correcting code with two logical qubits. This code has a von Neumann algebra $\mathcal{M}^{Z_\mathcal{G}}$ corresponding to the subregion $A$. We can restrict this code further by selecting a `gauge': an operator $P \in \mathcal{G}$ that is not in the center $P \not\in Z_\mathcal{G}$. Then, the abelian group $\langle P , Z_\mathcal{G} \rangle$ defines a stabilizer code with just one logical qubit. We can construct an encoding isometry for $\langle P , Z_\mathcal{G} \rangle$ by computing the two-qubit Pauli matrix, $P_L := V_\mathcal{G}^\dagger P V_\mathcal{G}$ and then constructing a two-qubit isometry $V_{P_L}$ such that $V_{P_L}^\dagger P_L V_{P_L} = Z_2$. Then the map $V_\mathcal{G}V_{P_L}$ is an encoding isometry for $\langle P , Z_\mathcal{G} \rangle$. This code has a von Neumann algebra $\mathcal{M}^{P}$ corresponding the bipatition $A$.

\begin{example} \textbf{3-1 bipartitions of the 2x2 Bacon Shor Code.} Here we derive that for \emph{any} choice of $P$, if $A$ contains just one qubit then $L = I_L$. We begin with analyzing the code defined by $Z_\mathcal{G}$. Since $Z_\mathcal{G}$ is symmetric to permutations of qubits, we assume without loss of generality that $A = \{1\}$, and use the method in Lemma~\ref{lemma:stabm} to compute $\M^{Z_\mathcal{G}}$:
    \begin{align}
        \mathcal{L}(\H_A) \otimes I_{\bar A} = \langle X_1 , Z_1\rangle_\text{vN} = \{I, X_1,Y_1,Z_1\}.
    \end{align} 
    We find that none of these commute with $Z_\mathcal{G}$, so $\M^{Z_\mathcal{G}} = \langle I \rangle_\text{vN}$. Therefore $S(\M^{Z_\mathcal{G}},\rho) = 0$ for all $\rho$. 

    Now all that remains to be done is to compute $S(\text{Tr}_{\bar A}(V_{Z_\mathcal{G}}^\dagger \rho V_{Z_\mathcal{G}}))$.    We define the following states for $a,b \in \{0,1\}$:
    \begin{align} 
        \ket{X^a Z^b} := \left( Z^b \otimes X^a  \right)  \frac{\ket{00} + \ket{11}}{\sqrt{2}}.
    \end{align}
Now we can inspect the action of $V_{Z_\mathcal{G}}$ on the computational basis for $\H_L$:
    \begin{align} 
        V_{Z_\mathcal{G}} \ket{ab}_L =  \ket{X^a Z^b}_{1,2} \ket{X^a Z^b}_{3,4}.
    \end{align}
Tracing out qubits $3,4$ measures $a,b$ so now qubits $1,2$ are in some probabilistic mixture of the $\ket{X^a Z^b}$. But these states are all maximally entangled, so the reduced state on qubit $1$ is $I/2$. So  $S(\text{Tr}_{\bar A}(V_{Z_\mathcal{G}}^\dagger \rho V_{Z_\mathcal{G}})) = 1$. We find that:
    \begin{align}
        \text{Tr}(\rho L) = S(\text{Tr}_{\bar A}(V_{Z_\mathcal{G}}^\dagger \rho V_{Z_\mathcal{G}})) - S(\M^{Z_\mathcal{G}},\rho) = 1,
    \end{align}
    which is achieved by $L = I_L$.

    Now we consider any gauge $P$, defining an isometry $V_{P_L}$. Observe that $\M^P = V_{P_L}^\dagger  \langle I \rangle_\text{vN}V_{P_L} =   \langle I \rangle_\text{vN}$ remains unchanged. Furthermore, we say that the entanglement entropy on qubit $1$ is $I/2$. So  $S(\text{Tr}_{\bar A}(V_{Z_\mathcal{G}}^\dagger \rho V_{Z_\mathcal{G}})) = 1$ is independent of $a,b$ so it also remains unchanged if we select a subspace of the code space. Thus we also have $L = I_L$ in this situation.
\end{example}

We saw that we could perform an analysis of all gauges $P$ in a unified manner by instead analyzing the code stabilized by $Z_\mathcal{G}$. While for 3-1 bipartitions the RT formula was independent of $P$, it is actually dependent on $P$ for 2-2 bipartitions. Nonetheless it is helpful to consider the code stabilized by $Z_\mathcal{G}$. For the discussion below, we label the von Neumann algebras with their corresponding subregions in the subscript: for example, $\mathcal{M}^{Z_\mathcal{G}}_{1,2}$ is the algebra defined by $V_{Z_\mathcal{G}}$ and $A = \{1,2\}$.

\begin{example} \textbf{2-2 bipartitions of the 2x2 Bacon Shor Code with no gauge.} We begin with $A = \{1,2\}$: Write:
    \begin{align}
        \mathcal{L}(\H_{1,2}) \otimes I_{3,4} = \langle X_1 , Z_1, X_2,Z_2\rangle_\text{vN}.
    \end{align} 
    Of these, only the subalgebra $\langle X_1X_2, Z_1Z_2 \rangle_\text{vN}$ commutes with $Z_G$, so:
    \begin{align}
        \M_{1,2}^{Z_\mathcal{G}} = \langle  V_{Z_\mathcal{G}}^\dagger X_1X_2 V_{Z_\mathcal{G}}, V_{Z_\mathcal{G}}^\dagger Z_1Z_2  V_{Z_\mathcal{G}} \rangle_\text{vN} = \langle Z_1 , Z_2  \rangle_\text{vN},
    \end{align}
    which is just the set of diagonal operators in the computational basis of $\H_L$. Now, while the code $Z_\mathcal{G}$ is symmetric with respect to qubit permutations, the encoding isometry $V_{Z_\mathcal{G}}$ is not.

    Since $\M^{Z_\mathcal{G}}_{1,2}$ is just the set of diagonal operators in the computational basis, $S(\M^{Z_\mathcal{G}}_{1,2},\rho)$ is the classical entropy after measuring in the computational basis.  Recall the relation $ V_{Z_\mathcal{G}} \ket{ab}_L =  \ket{X^a Z^b} \ket{X^a Z^b}$ from above. We see that discarding qubits $3,4$ essentially measures the first two qubits in the $\ket{X^a Z^b}$ basis, so $ S(\text{Tr}_{\bar A}(V_{Z_\mathcal{G}}^\dagger \rho V_{Z_\mathcal{G}}))$ is actually the same as $S(\M^{Z_\mathcal{G}}_{1,2},\rho)$. So their difference vanishes, and $L = 0$.

    Now we consider $A = \{1,3\}$, which we can actually obtain from the above analysis by inserting the gate $S_{2,3}$ that swaps qubits 2 and 3. That is, $V_{Z_\mathcal{G}}$ under $A=\{2,3\}$ should have the von Neumann algebra as $V_{Z_\mathcal{G}}S_{2,3}$ under $A=\{1,2\}$. Observe that $S_{2,3}$ commutes with $Z_G$, so it must implement a logical operator. With some calculation we see that $ V_{Z_\mathcal{G}}^\dagger S_{2,3} V_{Z_\mathcal{G}} = H^{\otimes 2} S_L $ where $S_L$ swaps the two logical qubits (this is done most easily by propagating $X_1,Z_1,X_2,Z_2$ through the Clifford circuit $V_{Z_\mathcal{G}}^\dagger S_{2,3} V_{Z_\mathcal{G}}$, and observing that it implements the same transformation as $H^{\otimes 2}  S_L$).

    As a result, we see that $\M^{Z_\mathcal{G}}_{1,3} = H^{\otimes 2} S_L \M^{Z_\mathcal{G}}_{1,2} S_L H^{\otimes 2}  = \langle X_1, X_2 \rangle_\text{vN} $, which is the set of diagonal operators in the $H^{\otimes 2}\ket{ab}_L$ basis. We also see that:
    \begin{align}
        V_{Z_\mathcal{G}} H^{\otimes 2}\ket{ab}_L = S_{2,3} V_{Z_\mathcal{G}}\ket{ba}_L  = \ket{X^a Z^b}_{1,3} \ket{X^a Z^b}_{2,4}.
    \end{align}
    Tracing out qubits 2 and 4, just like before, measures the qubits in 1 and 3 in the $\ket{X^a Z^b}$ basis, and the resulting entropy is the same as $S(\M^{Z_\mathcal{G}}_{1,3},\rho)$, implying $L = 0$.

\end{example}

The code $Z_\mathcal{G}$ is symmetrical under $S_{2,3}$ so we expect the code to have the same entanglement properties for both $A = \{1,2\}$ and $A = \{2,3\}$. However, gauges will break this symmetry. To illustrate this, we analyze the gauges considered by \cite{cao2020approximate}.

\begin{example} \textbf{2-2 bipartitions of the 2x2 Bacon Shor Code with fixed gauges}. First we consider $P = Z_1Z_2$ with $A = \{1,2\}$. We calculate $P_L = V_{Z_\mathcal{G}}^\dagger PV_{Z_\mathcal{G}} = Z_2 $. This forces the second qubit in $V_{Z_\mathcal{G}}$ to be $\ket{0}$: we could write $V_P \ket{\psi} = \ket{\psi}\ket{0}$. The corresponding von Neumann algebra is $\mathcal{M}^{Z_1Z_2}_{1,2} = V_P^\dagger \M^{Z_\mathcal{G}}_{1,2}  V_P =  \langle Z_1 \rangle_\text{vN}$, which is again just the set of diagonal operators in the computational basis on the first qubit. 

    The corresponding encoded states are $V_{Z_\mathcal{G}}\ket{a0}_L = \ket{X^a Z^0}_{1,2} \ket{X^a Z^0}_{3,4}$. We see that discarding qubits 3 and 4 measures qubits 1 and 2, so for the exact same reasoning as above, we have $ S(\text{Tr}_{\bar A}(V^\dagger_P V_{Z_\mathcal{G}}^\dagger \rho V_{Z_\mathcal{G}} V_P )) = S( \mathcal{M}^{Z_1Z_2}, \rho)$ so $L = 0$.

    However, $A = \{1,3\}$ with $P = Z_1Z_2$ yields a different result. We find that $\mathcal{M}^{Z_1Z_2}_{1,3} = V_P^\dagger \M^{Z_\mathcal{G}}_{1,3} V_P =  \langle X_1 \rangle_\text{vN} $, which is diagonal in the $H_1 \ket{a0}_L$ basis. $S(\mathcal{M}^{Z_1Z_2}_{1,3},\rho)$ corresponds to the entropy of the $a$ degree of freedom as measured in the computational basis. We find that:
    \begin{align}
        V_{Z_\mathcal{G}} H_1 \ket{a0}_L &= S_{2,3} V_{Z_\mathcal{G}}  S_L H^{\otimes 2}  H_1 \ket{a0}_L  =  S_{2,3} V_{Z_\mathcal{G}}  H_1  S_L \ket{a0}_L =    S_{2,3} V_{Z_\mathcal{G}} \ket{+a}_L\\
        &= \frac{\ket{X^0 Z^a}_{1,3} \ket{X^0 Z^a}_{2,4} + \ket{X^1 Z^a}_{1,3} \ket{X^1 Z^a}_{2,4} }{\sqrt{2}}.
    \end{align}
    
    Now we see that tracing out qubits 2 and 4 yields two sources of entropy for the remaining qubits on $A$: one bit of entropy from the $X$ degree of freedom, and the other stemming from the measurement of $a$ in the computational basis. Thus, $S(\text{Tr}_{\bar A}(V_{Z_\mathcal{G}}^\dagger \rho V_{Z_\mathcal{G}})) - S(\mathcal{M}^{Z_1Z_2}_{1,3},\rho) = 1$. So $L = I_L$.

    We saw that for $P = Z_1 Z_2$, $A = \{1,2\}$ featured $S = 0$ and $S = \{1,3\}$ featured $S = I_L$. Now we consider $P = X_1 X_3$: we will find that the opposite is the case by just swapping $V_{Z_\mathcal{G}}$ with $S_{2,3}V_{Z_\mathcal{G}}$. We compute: 
    \begin{align}
        V_{Z_\mathcal{G}}^\dagger S_{2,3} P S_{2,3} V_{Z_\mathcal{G}} = S_L H^{\otimes 2} V_{Z_\mathcal{G}}^\dagger P  V_{Z_\mathcal{G}} H^{\otimes 2} S_L = S_L H^{\otimes 2}  X_1  H^{\otimes 2} S_L  = Z_2.
    \end{align}

    So we see that $P = X_1X_3$ behaves just like $Z_1 Z_2$ when qubits 2 and 3 are swapped. Thus, $A = \{1,2\}$ features $S = I_L$ and $S = \{1,3\}$ featured $S = 0$.
\end{example}

\clearpage


\printbibliography[heading=bibintoc] 

@article{almheiri2015bulk,
  title={Bulk locality and quantum error correction in AdS/CFT},
  author={Almheiri, Ahmed and Dong, Xi and Harlow, Daniel},
  journal={Journal of High Energy Physics},
  volume={2015},
  number={4},
  pages={1--34},
  year={2015},
  publisher={Springer},
  archivePrefix = "arXiv", eprint={1411.7041},
  doi={10.1007/JHEP04(2015)163}
}

@article{hayden2004structure,
  title={Structure of states which satisfy strong subadditivity of quantum entropy with equality},
  author={Hayden, Patrick and Jozsa, Richard and Petz, Denes and Winter, Andreas},
  journal={Communications in mathematical physics},
  volume={246},
  number={2},
  pages={359--374},
  year={2004},
  publisher={Springer},
  archivePrefix = "arXiv", eprint={quant-ph/0304007},
  doi={10.1007/s00220-004-1049-z}
}

@article{carroll2001no,
  title={A No-Nonsense Introduction to General Relativity},
  author={Carroll, Sean M},
  journal={Enrico Fermi Institute and Department of Physics, University of Chicago},
  year={2001}
}

@article{harlow2016jerusalem,
  title={Jerusalem lectures on black holes and quantum information},
  author={Harlow, Daniel},
  journal={Reviews of Modern Physics},
  volume={88},
  number={1},
  pages={015002},
  year={2016},
  publisher={APS},
  archivePrefix = "arXiv", eprint={1409.1231},
  doi={10.1103/RevModPhys.88.015002}
}

@book{rath2020aspects,
  title={Aspects of Holography And Quantum Error Correction},
  author={Rath, Pratik},
  year={2020},
  publisher={University of California, Berkeley},
  url={https://escholarship.org/uc/item/4pd8g3wz}
}

@article{schumacher1996quantum,
  title={Quantum data processing and error correction},
  author={Schumacher, Benjamin and Nielsen, Michael A},
  journal={Physical Review A},
  volume={54},
  number={4},
  pages={2629},
  year={1996},
  publisher={APS},
  archivePrefix = "arXiv", eprint={quant-ph/9604022},
  doi={10.1103/PhysRevA.54.2629}
}

@article{harlow2017ryu,
  title={The Ryu--Takayanagi formula from quantum error correction},
  author={Harlow, Daniel},
  journal={Communications in Mathematical Physics},
  volume={354},
  number={3},
  pages={865--912},
  year={2017},
  publisher={Springer},
  archivePrefix = "arXiv", eprint={1607.03901},
  doi={10.1007/s00220-017-2904-z}
}

@article{ryu2006holographic,
  title={Holographic derivation of entanglement entropy from the anti--de Sitter space/conformal field theory correspondence},
  author={Ryu, Shinsei and Takayanagi, Tadashi},
  journal={Physical review letters},
  volume={96},
  number={18},
  pages={181602},
  year={2006},
  publisher={APS},
  archivePrefix = "arXiv", eprint={hep-th/0603001},
  doi={10.1103/PhysRevLett.96.181602}
}

@article{harlow2013quantum,
  title={Quantum computation vs. firewalls},
  author={Harlow, Daniel and Hayden, Patrick},
  journal={Journal of High Energy Physics},
  volume={2013},
  number={6},
  pages={1--56},
  year={2013},
  publisher={Springer},
  archivePrefix = "arXiv", eprint={1301.4504},
  doi={10.1007/JHEP06(2013)085}
}

@article{verlinde2013black,
  title={Black hole entanglement and quantum error correction},
  author={Verlinde, Erik and Verlinde, Herman},
  journal={Journal of High Energy Physics},
  volume={2013},
  number={10},
  pages={1--34},
  year={2013},
  publisher={Springer},
  archivePrefix = "arXiv", eprint={1211.6913},
  doi={JHEP10(2013)107}
}

@article{kohler2019toy,
  title={Toy models of holographic duality between local Hamiltonians},
  author={Kohler, Tamara and Cubitt, Toby},
  journal={Journal of High Energy Physics},
  volume={2019},
  number={8},
  pages={1--64},
  year={2019},
  publisher={Springer},
  archivePrefix = "arXiv", eprint={1810.08992},
  doi={10.1007/JHEP08(2019)017}
}

@article{donnelly2017living,
  title={Living on the edge: a toy model for holographic reconstruction of algebras with centers},
  author={Donnelly, William and Marolf, Donald and Michel, Ben and Wien, Jason},
  journal={Journal of High Energy Physics},
  volume={2017},
  number={4},
  pages={1--18},
  year={2017},
  publisher={Springer},
  archivePrefix = "arXiv", eprint={1611.05841},
  doi={10.1007/JHEP04(2017)093}
}

@article{haferkamp2021linear,
  title={Linear growth of quantum circuit complexity},
  author={Haferkamp, Jonas and Faist, Philippe and Kothakonda, Naga BT and Eisert, Jens and Halpern, Nicole Yunger},
  journal={arXiv preprint arXiv:2106.05305},
  year={2021},
  archivePrefix = "arXiv", eprint={2106.05305}
}

@article{hayden2016holographic,
  title={Holographic duality from random tensor networks},
  author={Hayden, Patrick and Nezami, Sepehr and Qi, Xiao-Liang and Thomas, Nathaniel and Walter, Michael and Yang, Zhao},
  journal={Journal of High Energy Physics},
  volume={2016},
  number={11},
  pages={1--56},
  year={2016},
  publisher={Springer},
  archivePrefix = "arXiv", eprint={1601.01694},
  doi={10.1007/JHEP11(2016)009}
}

@article{mintun2015bulk,
  title={Bulk-boundary duality, gauge invariance, and quantum error corrections},
  author={Mintun, Eric and Polchinski, Joseph and Rosenhaus, Vladimir},
  journal={Physical review letters},
  volume={115},
  number={15},
  pages={151601},
  year={2015},
  publisher={APS},
  archivePrefix = "arXiv", eprint={1501.06577},
  doi={10.1103/PhysRevLett.115.151601}
}

@article{almheiri2013black,
  title={Black holes: complementarity or firewalls?},
  author={Almheiri, Ahmed and Marolf, Donald and Polchinski, Joseph and Sully, James},
  journal={Journal of High Energy Physics},
  volume={2013},
  number={2},
  pages={1--20},
  year={2013},
  publisher={Springer},
  archivePrefix = "arXiv", eprint={1207.3123},
  doi={10.1007/JHEP02(2013)062}
}

@article{maldacena1999large,
  title={The large-N limit of superconformal field theories and supergravity},
  author={Maldacena, Juan},
  journal={International journal of theoretical physics},
  volume={38},
  number={4},
  pages={1113--1133},
  year={1999},
  publisher={Springer},
  url={https://arxiv.org/hep-th/9711200},
  doi={10.1023/A:1026654312961}
}

@article{Witten:1998qj,
    author = "Witten, Edward",
    title = "{Anti-de Sitter space and holography}",
    eprint = "hep-th/9802150",
    archivePrefix = "arXiv",
    reportNumber = "IASSNS-HEP-98-15",
    doi = "10.4310/ATMP.1998.v2.n2.a2",
    journal = "Advances in Theoretical and Mathematical Physics",
    volume = "2",
    pages = "253--291",
    year = "1998"
}

@article{Aharony:1999ti,
    author = "Aharony, Ofer and Gubser, Steven S. and Maldacena, Juan Martin and Ooguri, Hirosi and Oz, Yaron",
    title = "{Large N field theories, string theory and gravity}",
    eprint = "hep-th/9905111",
    archivePrefix = "arXiv",
    reportNumber = "CERN-TH-99-122, HUTP-99-A027, LBNL-43113, RU-99-18, UCB-PTH-99-16, LBL-43113",
    doi = "10.1016/S0370-1573(99)00083-6",
    journal = "Physics Reports",
    volume = "323",
    pages = "183--386",
    year = "2000"
}

@article{maldacena2013cool,
  title={Cool horizons for entangled black holes},
  author={Maldacena, Juan and Susskind, Leonard},
  journal={Fortschritte der Physik},
  volume={61},
  number={9},
  pages={781--811},
  year={2013},
  publisher={Wiley Online Library},
  archivePrefix = "arXiv", eprint={1306.0533},
  doi={10.1002/prop.201300020}
}

@article{brown2016holographic,
  title={Holographic complexity equals bulk action?},
  author={Brown, Adam R and Roberts, Daniel A and Susskind, Leonard and Swingle, Brian and Zhao, Ying},
  journal={Physical review letters},
  volume={116},
  number={19},
  pages={191301},
  year={2016},
  publisher={APS},
  archivePrefix = "arXiv", eprint={1509.07876},
  doi={10.1103/PhysRevLett.116.191301}
}

@article{aaronson2016complexity,
  title={The complexity of quantum states and transformations: from quantum money to black holes},
  author={Aaronson, Scott},
  journal={arXiv preprint arXiv:1607.05256},
  year={2016},
  archivePrefix = "arXiv", eprint={1607.05256}
}

@article{susskind2016computational,
  title={Computational complexity and black hole horizons},
  author={Susskind, Leonard},
  journal={Fortschritte der Physik},
  volume={64},
  number={1},
  pages={24--43},
  year={2016},
  publisher={Wiley Online Library},
  archivePrefix = "arXiv", eprint={1402.5674},
  doi={10.1002/prop.201500092}
}

@article{almheiri2020page,
  title={The Page curve of Hawking radiation from semiclassical geometry},
  author={Almheiri, Ahmed and Mahajan, Raghu and Maldacena, Juan and Zhao, Ying},
  journal={Journal of High Energy Physics},
  volume={2020},
  number={3},
  pages={1--24},
  year={2020},
  publisher={Springer},
  archivePrefix = "arXiv", eprint={1908.10996},
  doi={10.1007/JHEP03(2020)149}
}

@article{penington2020entanglement,
  title={Entanglement wedge reconstruction and the information paradox},
  author={Penington, Geoffrey},
  journal={Journal of High Energy Physics},
  volume={2020},
  number={9},
  pages={1--84},
  year={2020},
  publisher={Springer},
  archivePrefix = "arXiv", eprint={1905.08255},
  doi={10.1007/JHEP09(2020)002}
}

@article{cao2020approximate,
  title={Approximate Bacon-Shor Code and Holography},
  author={Cao, ChunJun and Lackey, Brad},
  journal={arXiv preprint arXiv:2010.05960},
  year={2020},
  archivePrefix = "arXiv", eprint={2010.05960},
  doi={10.1007/JHEP05(2021)127}
}

@article{kribs2006operator,
  title={Operator quantum error correction},
  author={Kribs, David W and Laflamme, Raymond and Poulin, David and Lesosky, Maia},
  journal={Quantum Information \& Computation},
  volume={6},
  number={4},
  pages={382--399},
  year={2006},
  publisher={Rinton Press, Incorporated Paramus, NJ},
  archivePrefix = "arXiv", eprint={quant-ph/0504189}
}

@article{crann2016private,
  title={Private algebras in quantum information and infinite-dimensional complementarity},
  author={Crann, Jason and Kribs, David W and Levene, Rupert H and Todorov, Ivan G},
  journal={Journal of Mathematical Physics},
  volume={57},
  number={1},
  pages={015208},
  year={2016},
  publisher={AIP Publishing LLC},
  doi={10.1063/1.4935399},
  archivePrefix = "arXiv", eprint={1510.06672}
}

@article{kretschmann2008complementarity,
  title={Complementarity of private and correctable subsystems in quantum cryptography and error correction},
  author={Kretschmann, Dennis and Kribs, David W and Spekkens, Robert W},
  journal={Physical Review A},
  volume={78},
  number={3},
  pages={032330},
  year={2008},
  publisher={APS},
  url={http://arxiv.org/abs/0711.3438},
  doi={10.1103/PhysRevA.78.032330}
}

@article{kribs2018quantum,
  title={Quantum complementarity and operator structures},
  author={Kribs, David W and Levick, Jeremy and Nelson, Mike I and Pereira, Rajesh and Rahaman, Mizanur},
  journal={arXiv preprint arXiv:1811.10425},
  year={2018},
  archivePrefix = "arXiv", eprint={1811.10425}
}

@article{beny2007quantum,
  title={Quantum error correction of observables},
  author={B{\'e}ny, C{\'e}dric and Kempf, Achim and Kribs, David W},
  journal={Physical Review A},
  volume={76},
  number={4},
  pages={042303},
  year={2007},
  publisher={APS},
  archivePrefix = "arXiv", eprint={0705.1574},
  doi={10.1103/PhysRevA.76.042303}
}

@article{beny2007generalization,
  title={Generalization of quantum error correction via the Heisenberg picture},
  author={B{\'e}ny, C{\'e}dric and Kempf, Achim and Kribs, David W},
  journal={Physical review letters},
  volume={98},
  number={10},
  pages={100502},
  year={2007},
  publisher={APS},
  archivePrefix = "arXiv", eprint={quant-ph/0608071},
  doi={10.1103/PhysRevLett.98.100502}
}

@article{kribs2005unified,
  title={Unified and generalized approach to quantum error correction},
  author={Kribs, David and Laflamme, Raymond and Poulin, David},
  journal={Physical review letters},
  volume={94},
  number={18},
  pages={180501},
  year={2005},
  publisher={APS},
  archivePrefix = "arXiv", eprint={quant-ph/0412076},
  doi={10.1103/PhysRevLett.94.180501}
}

@inproceedings{gottesman2010introduction,
  title={An introduction to quantum error correction and fault-tolerant quantum computation},
  author={Gottesman, Daniel},
  booktitle={Quantum information science and its contributions to mathematics, Proceedings of Symposia in Applied Mathematics},
  volume={68},
  pages={13--58},
  year={2010},
  archivePrefix = "arXiv", eprint={0904.2557}
}

@article{gottesman1997stabilizer,
  title={Stabilizer codes and quantum error correction},
  author={Gottesman, Daniel},
  journal={arXiv preprint quant-ph/9705052},
  year={1997},
  archivePrefix = "arXiv", eprint={quant-ph/9705052}
}

@book{gourgoulhon20123+,
  title={3+ 1 formalism in general relativity: bases of numerical relativity},
  author={Gourgoulhon, Eric},
  volume={846},
  year={2012},
  publisher={Springer Science \& Business Media},
  archivePrefix = "arXiv", eprint={gr-qc/0703035}
}

@article{taylor2021holography,
  title={Holography, cellulations and error correcting codes},
  author={Taylor, Marika and Woodward, Charles},
  journal={arXiv preprint arXiv:2112.12468},
  year={2021},
  archivePrefix = "arXiv", eprint={2112.12468}
}

@article{jahn2021holographic,
  title={Holographic tensor network models and quantum error correction: A topical review},
  author={Jahn, Alexander and Eisert, Jens},
  journal={Quantum Science and Technology},
  year={2021},
  publisher={IOP Publishing},
  doi={10.1088/2058-9565/ac0293},
  archivePrefix = "arXiv", eprint={2102.02619}
}

@article{poulin2005stabilizer,
  title={Stabilizer formalism for operator quantum error correction},
  author={Poulin, David},
  journal={Physical review letters},
  volume={95},
  number={23},
  pages={230504},
  year={2005},
  publisher={APS},
  archivePrefix = "arXiv", eprint={quant-ph/0508131},
  doi={10.1103/PhysRevLett.95.230504}
}

@article{harlow2018tasi,
  title={TASI Lectures on the Emergence of the Bulk in AdS/CFT},
  author={Harlow, Daniel},
  journal={arXiv preprint arXiv:1802.01040},
  year={2018},
  archivePrefix = "arXiv", eprint={1802.01040}
}

@article{pastawski2015holographic,
  title={Holographic quantum error-correcting codes: Toy models for the bulk/boundary correspondence},
  author={Pastawski, Fernando and Yoshida, Beni and Harlow, Daniel and Preskill, John},
  journal={Journal of High Energy Physics},
  volume={2015},
  number={6},
  pages={149},
  year={2015},
  publisher={Springer},
  archivePrefix = "arXiv",eprint={1503.06237},
  doi={10.1007/JHEP06(2015)149}
}

@article{bravyi2011subsystem,
  title={Subsystem codes with spatially local generators},
  author={Bravyi, Sergey},
  journal={Physical Review A},
  volume={83},
  number={1},
  pages={012320},
  year={2011},
  publisher={APS},
  archivePrefix = "arXiv",
  eprint = {1008.1029},
  %archivePrefix = "arXiv",eprint={1008.1029},
  doi={10.1103/PhysRevA.83.012320}
}

@article{jones2003neumann,
  title={von Neumann algebras},
  author={Jones, Vaughan FR},
  journal={Course at Berkeley University},
  year={2003},
  url={https://math.berkeley.edu/~vfr/MATH20909/VonNeumann2009.pdf}
}

@article{Hubeny:2007xt,
    author = "Hubeny, Veronika E. and Rangamani, Mukund and Takayanagi, Tadashi",
    title = "{A Covariant holographic entanglement entropy proposal}",
    eprint = "0705.0016",
    archivePrefix = "arXiv",
    %primaryClass = "hep-th",
    reportNumber = "DCPT-07-13, KUNS-2069",
    doi = "10.1088/1126-6708/2007/07/062",
    journal = "Journal of High Energy Physics",
    volume = "07",
    pages = "062",
    year = "2007"
}

@article{Engelhardt:2014gca,
    author = "Engelhardt, Netta and Wall, Aron C.",
    title = "{Quantum Extremal Surfaces: Holographic Entanglement Entropy beyond the Classical Regime}",
    eprint = "1408.3203",
    archivePrefix = "arXiv",
    %primaryClass = "hep-th",
    doi = "10.1007/JHEP01(2015)073",
    journal = "Journal of High Energy Physics",
    volume = "01",
    pages = "073",
    year = "2015"
}

@article{PhysRevLett.47.979,
  title = {Indirect Evidence for Quantum Gravity},
  author = {Page, Don N. and Geilker, C. D.},
  journal = {Physical Review Letters},
  volume = {47},
  issue = {14},
  pages = {979--982},
  numpages = {0},
  year = {1981},
  month = {Oct},
  publisher = {American Physical Society},
  doi = {10.1103/PhysRevLett.47.979},
  url = {https://link.aps.org/doi/10.1103/PhysRevLett.47.979}
}

@article{Polarski:1995jg,
    author = "Polarski, David and Starobinsky, Alexei A.",
    title = "{Semiclassicality and decoherence of cosmological perturbations}",
    eprint = "gr-qc/9504030",
    archivePrefix = "arXiv",
    reportNumber = "LMPM-95-4",
    doi = "10.1088/0264-9381/13/3/006",
    journal = "Classical and Quantum Gravity",
    volume = "13",
    pages = "377--392",
    year = "1996"
}

@article{Lombardo:2005iz,
    author = "Lombardo, Fernando C. and Lopez Nacir, Diana",
    title = "{Decoherence during inflation: The Generation of classical inhomogeneities}",
    eprint = "gr-qc/0506051",
    archivePrefix = "arXiv",
    doi = "10.1103/PhysRevD.72.063506",
    journal = "Physical Review D",
    volume = "72",
    pages = "063506",
    year = "2005"
}

@article{Guth:1980zm,
    author = "Guth, Alan H.",
    editor = "Fang, Li-Zhi and Ruffini, R.",
    title = "{The Inflationary Universe: A Possible Solution to the Horizon and Flatness Problems}",
    reportNumber = "SLAC-PUB-2576",
    doi = "10.1103/PhysRevD.23.347",
    journal = "Physical Review D",
    volume = "23",
    pages = "347--356",
    year = "1981"
}

@article{Linde:1981mu,
    author = "Linde, Andrei D.",
    editor = "Fang, Li-Zhi and Ruffini, R.",
    title = "{A New Inflationary Universe Scenario: A Possible Solution of the Horizon, Flatness, Homogeneity, Isotropy and Primordial Monopole Problems}",
    reportNumber = "LEBEDEV-81-229",
    doi = "10.1016/0370-2693(82)91219-9",
    journal = "Physics Letters B",
    volume = "108",
    pages = "389--393",
    year = "1982"
}

@article{Albrecht:1982wi,
    author = "Albrecht, Andreas and Steinhardt, Paul J.",
    editor = "Fang, Li-Zhi and Ruffini, R.",
    title = "{Cosmology for Grand Unified Theories with Radiatively Induced Symmetry Breaking}",
    reportNumber = "UPR-0185T",
    doi = "10.1103/PhysRevLett.48.1220",
    journal = "Physical Review Letters.",
    volume = "48",
    pages = "1220--1223",
    year = "1982"
}

@article{Hawking:1975vcx,
    author = "Hawking, S. W.",
    editor = "Gibbons, G. W. and Hawking, S. W.",
    title = "{Particle Creation by Black Holes}",
    doi = "10.1007/BF02345020",
    journal = "Commun. Math. Phys.",
    volume = "43",
    pages = "199--220",
    year = "1975",
    note = "[Erratum: Commun.Math.Phys. 46, 206 (1976)]"
}

@article{Wall:2012uf,
    author = "Wall, Aron C.",
    title = "{Maximin Surfaces, and the Strong Subadditivity of the Covariant Holographic Entanglement Entropy}",
    eprint = "1211.3494",
    archivePrefix = "arXiv",
    %primaryClass = "hep-th",
    doi = "10.1088/0264-9381/31/22/225007",
    journal = "Classical and Quantum Gravity",
    volume = "31",
    number = "22",
    pages = "225007",
    year = "2014"
}

@article{Headrick:2014cta,
    author = "Headrick, Matthew and Hubeny, Veronika E. and Lawrence, Albion and Rangamani, Mukund",
    title = "{Causality \& holographic entanglement entropy}",
    eprint = "1408.6300",
    archivePrefix = "arXiv",
    %primaryClass = "hep-th",
    reportNumber = "DCPT-14-33, BRX-TH-6284",
    doi = "10.1007/JHEP12(2014)162",
    journal = "Journal of High Energy Physics",
    volume = "12",
    pages = "162",
    year = "2014"
}

@article{Almheiri:2016blp,
    author = "Almheiri, Ahmed and Dong, Xi and Swingle, Brian",
    title = "{Linearity of Holographic Entanglement Entropy}",
    eprint = "1606.04537",
    archivePrefix = "arXiv",
    %primaryClass = "hep-th",
    doi = "10.1007/JHEP02(2017)074",
    journal = "Journal of High Energy Physics",
    volume = "02",
    pages = "074",
    year = "2017"
}

@article{Akers:2021fut,
    author = "Akers, Chris and Penington, Geoff",
    title = "{Quantum minimal surfaces from quantum error correction}",
    journal={arXiv preprint arXiv:2109.14618},
    eprint = "2109.14618",
    archivePrefix = "arXiv",
    %primaryClass = "hep-th",
    month = "9",
    year = "2021"
}

@article{Donnelly:2016auv,
    author = "Donnelly, William and Freidel, Laurent",
    title = "{Local subsystems in gauge theory and gravity}",
    eprint = "1601.04744",
    archivePrefix = "arXiv",
    %primaryClass = "hep-th",
    doi = "10.1007/JHEP09(2016)102",
    journal = "Journal of High Energy Physics",
    volume = "09",
    pages = "102",
    year = "2016"
}

@article{Hamilton:2005ju,
    author = "Hamilton, Alex and Kabat, Daniel N. and Lifschytz, Gilad and Lowe, David A.",
    title = "{Local bulk operators in AdS/CFT: A Boundary view of horizons and locality}",
    eprint = "hep-th/0506118",
    archivePrefix = "arXiv",
    reportNumber = "BROWN-HET-1448, CU-TP-1130",
    doi = "10.1103/PhysRevD.73.086003",
    journal = "Physical Review D",
    volume = "73",
    pages = "086003",
    year = "2006"
}

@article{Hamilton:2006az,
    author = "Hamilton, Alex and Kabat, Daniel N. and Lifschytz, Gilad and Lowe, David A.",
    title = "{Holographic representation of local bulk operators}",
    eprint = "hep-th/0606141",
    archivePrefix = "arXiv",
    reportNumber = "CU-TP-1149",
    doi = "10.1103/PhysRevD.74.066009",
    journal = "Physical Review D",
    volume = "74",
    pages = "066009",
    year = "2006"
}

@article{Skenderis:2008dg,
    author = "Skenderis, Kostas and van Rees, Balt C.",
    title = "{Real-time gauge/gravity duality: Prescription, Renormalization and Examples}",
    eprint = "0812.2909",
    archivePrefix = "arXiv",
    %primaryClass = "hep-th",
    reportNumber = "ITFA-2008-50",
    doi = "10.1088/1126-6708/2009/05/085",
    journal = "Journal of High Energy Physics",
    volume = "05",
    pages = "085",
    year = "2009"
}

@article{Christodoulou:2016nej,
    author = "Christodoulou, Ariana and Skenderis, Kostas",
    title = "{Holographic Construction of Excited CFT States}",
    eprint = "1602.02039",
    archivePrefix = "arXiv",
    %primaryClass = "hep-th",
    doi = "10.1007/JHEP04(2016)096",
    journal = "Journal of High Energy Physics",
    volume = "04",
    pages = "096",
    year = "2016"
}

@article{Gesteau:2020hoz,
    author = "Gesteau, Elliott and Kang, Monica Jinwoo",
    title = "{The infinite-dimensional HaPPY code: entanglement wedge reconstruction and dynamics}",
    eprint = "2005.05971",
    journal={arXiv preprint arXiv:2005.05971},
    archivePrefix = "arXiv",
    %primaryClass = "hep-th",
    reportNumber = "CALT-TH-2020-016",
    month = "5",
    year = "2020"
}

@article{Jahn:2020ukq,
    author = "Jahn, Alexander and Zimbor\'as, Zolt\'an and Eisert, Jens",
    title = "{Tensor network models of AdS/qCFT}",
    journal={arXiv preprint arXiv:2004.04173},
    eprint = "2004.04173",
    archivePrefix = "arXiv",
    %primaryClass = "quant-ph",
    month = "4",
    year = "2020"
}

@article{Cao:2021wrb,
    author = "Cao, Chunjun and Pollack, Jason and Wang, Yixu",
    title = "{Hyper-Invariant MERA: Approximate Holographic Error Correction Codes with Power-Law Correlations}",
    journal={arXiv preprint arXiv:2103.08631},
    eprint = "2103.08631",
    archivePrefix = "arXiv",
    %primaryClass = "quant-ph",
    month = "3",
    year = "2021"
}

@article{Kabernik:2019jko,
    author = "Kabernik, Oleg and Pollack, Jason and Singh, Ashmeet",
    title = "{Quantum State Reduction: Generalized Bipartitions from Algebras of Observables}",
    eprint = "1909.12851",
    archivePrefix = "arXiv",
    %primaryClass = "quant-ph",
    reportNumber = "CALT-TH-2019-36",
    doi = "10.1103/PhysRevA.101.032303",
    journal = "Physical Review A",
    volume = "101",
    number = "3",
    pages = "032303",
    year = "2020"
}

@article{Gubser:1998bc,
    author = "Gubser, S. S. and Klebanov, Igor R. and Polyakov, Alexander M.",
    title = "{Gauge theory correlators from noncritical string theory}",
    eprint = "hep-th/9802109",
    archivePrefix = "arXiv",
    reportNumber = "PUPT-1767",
    doi = "10.1016/S0370-2693(98)00377-3",
    journal = "Phys. Lett. B",
    volume = "428",
    pages = "105--114",
    year = "1998"
}

\end{document}